\documentclass[PhD,twoside,openright]{iitmdiss}

\usepackage{afterpage}
\usepackage{times}
\usepackage{t1enc}
\usepackage[export]{adjustbox}
\usepackage{graphicx}
\usepackage{adjustbox }
\usepackage{epstopdf}
\usepackage{subfigure}
\usepackage{latexsym}
\usepackage{amssymb}
\usepackage{amsbsy}
\usepackage{mathtools}
\usepackage{color}
\usepackage{amsthm}
\usepackage{amsfonts}
\usepackage{supertabular}
\usepackage{emptypage}

\theoremstyle{plain}
\newtheorem{theorem}{Theorem}[chapter]
\newtheorem{assumption}{Assumption}[chapter]

\newtheorem{lemma}{Lemma}[chapter]
\newtheorem{proposition}{Proposition}[chapter]
\newtheorem{ctrlobj}{Control Objective}[chapter]

\theoremstyle{definition}
\newtheorem{definition}{Definition}[chapter]
\newtheorem{example}{Example}[chapter]
\newtheorem{remark}{Remark}[chapter]

\theoremstyle{remark}

\newcommand*\bigcdot{\mathpalette\bigcdot@{.5}}
\newcommand\blankpage{%
	\null
	\thispagestyle{empty}%
	\addtocounter{page}{-1}%
	\newpage}
\newcommand{\beqn}{\begin{eqnarray*}} 
\newcommand{\eeqn}{\end{eqnarray*}}
\newcommand{\beq}{\begin{eqnarray}} 
\newcommand{\eeq}{\end{eqnarray}}
\def\deff{\stackrel{\triangle}{=}}

\newcommand{\djc}[1]{\textcolor{black}{#1}}
\newcommand{\kc}[1]{\textcolor{black}{#1}}
\newcommand{\kcw}[1]{\textcolor{black}{#1}}
\newcommand{\kcb}[1]{\textcolor{black}{#1}}
\newcommand{\new}[1]{\textcolor{black}{#1}}

\newcommand{\kcn}[1]{\textcolor{black}{#1}}
\newcommand{\ijc}[1]{\textcolor{black}{#1}}

\newcommand{\thesis}[1]{\textcolor{black}{#1}}

\def\norm#1{\|#1\|}

\newcommand{\blue}[1]{\textcolor{black}{#1}}

\begin{document}

\title{Alternative passive maps in the Brayton-Moser framework: Implications on control and optimization}

\author{Krishna Chaitanya Kosaraju}

\date{2017}
\department{Electrical Engineering}

\maketitle
 \certificate

 \vspace*{0.5in}

 \noindent This is to certify that the thesis titled {\bf Alternative passive maps in the Brayton-Moser framework: Implications on control and optimization}, submitted by {\bf Krishna Chaitanya Kosaraju}, 
   to the Indian Institute of Technology, Madras, for
 the award of the degree of {\bf Doctor of Philosophy}, is a bona fide
 record of the research work done by him under my supervision.  The
 contents of this thesis, in full or in parts, have not been submitted
 to any other Institute or University for the award of any degree or
 diploma.

 \vspace*{1.5in}

 \begin{singlespacing}
 	
 \hspace*{-0.25in}
 \parbox{2.5in}{
 \noindent {\bf Dr. Ramkrishna Pasumarthy} \\
 \noindent Research Guide \\ 
 \noindent Professor \\
 \noindent Dept. of Electrical Engineering\\
 \noindent IIT-Madras, 600 036 \\
 } 
 \hspace*{1.0in} 
 \end{singlespacing}

 \vspace*{0.25in}
 \noindent Place: Chennai\\
 Date: 19th January 2009

\acknowledgements

I would like to express my thanks of gratitude to my thesis advisor, Dr. Ramkrishna Pasumarthy, for giving this opportunity, and all his support. Sir, really thank you for putting up with my laziness. I sincerely thank, Dr. Arun D. Mahindrakar, for introducing me to control theory. Sir, thank you for motivating me to choose a career in research.   I would like to express my gratitude to Singh Sir for helping me look beyond the horizon. I would also like to thank Venky for all his fruitful collaboration.

\noindent There are those with out whom, my stay at IIT-Madras would not have been this enjoyable, to name a few Bhaskar, Srinivas, Vikram, Kruthika, Anup, Asit, Sravan, Sai Krishna, Vijay, Akshit, Sailash, Durgesh, Sharad, Yashraj, Niharika, Sunil, Abhiskek, Lokesh, Gourav, Immanuvel, Mithun, Abraham, Tahiya, Gopal and list goes on ... Thanks a lot for our never ending insane time waste in tea shops,  partying, making plans that never went beyond the night...

\noindent I would like to thank my family for all their constant love and support. There is one person to whom I cant just say thanks and move on. Vasu, saying 'this would not have been possible without your support', is an under statement.

\vfill \hfill ... to our little princess
\newpage 

 \vspace{2cm}
\afterpage{\blankpage}
 \begin{center}
 	\textbf{\large ABSTRACT}
 \end{center}
\addcontentsline{toc}{chapter}{ABSTRACT}

%
%
%
 \noindent In the recent years, passivity theory has gained renewed attention because of its advantages and practicality in modeling of multi-domain systems and constructive control techniques. Unlike Lyapunov theory, passivity theory takes a behavioral approach in its control design methodologies.
 Hence, it provides solutions, which not only achieve the control objectives, but are also easily interpretable in the standard engineering parlance.
 
 \noindent  The fundamental idea in passivity based control (PBC) methodologies is to find a controller that renders the closed-loop system passive. 
  It is well known that, the PBC methodologies that rely on power-conjugate port-variables do not work for control objectives that require bounded power and unbounded energy. This is commonly known as the {\em dissipation obstacle}. One possible alternative that has been well explored, in the case of finite dimensional systems, is Brayton-Moser formulation. 
   However, designing controllers in this framework leads to various difficulties, such as, solving for partial differential equations and finding storage functions satisfying a gradient structure.
   
\noindent In this thesis, we first show that the output port-variable derived from Brayton-Moser formulation is integrable, under the assumption that the input matrix is integrable. The integrated output port-variable is then used to construct a desired closed-loop storage function for the closed-loop system. Secondly, we show that a class of Brayton-Moser systems are contracting. This results in a new passivity property with ``differentiation at both port-variables''. We extend this to a class of contracting nonlinear systems using dynamic feedback and Krasovskii-type storage functions. 

\noindent Systems represented in Brayton-Moser framework possess a pseudo-gradient structure. Another class of problems where pseudo-gradient form naturally appears, is in the primal-dual gradient-methods of convex optimization. This observation motivates us to present passivity based converge analysis for the primal-dual gradient dynamics.

\noindent Brayton-Moser formulation is not a well-established topic in infinite-dimensional systems, albeit dissipation obstacle is more prevalent in these systems. In this thesis, we present modeling and control aspects of infinite-dimensional port-Hamiltonian systems (defined by a Stokes-Dirac structure) in Brayton-Moser framework. We illustrate these methods using (i) stability analysis of Maxwell's equations in $\mathbb{R}^3$, (ii) boundary control of transmission line system modeled using Telegrapher's equations.
%
%
%
%
%
 \pagebreak

\begin{singlespace}
  \tableofcontents
\thispagestyle{empty}
 \listoftables
\addcontentsline{toc}{chapter}{LIST OF TABLES}
\listoffigures
\addcontentsline{toc}{chapter}{LIST OF FIGURES}
\end{singlespace}







\pagebreak
\clearpage

\pagenumbering{arabic}

\chapter{Introduction}
%
The notion of passivity, originating from electrical networks, has been very useful in analyzing stability of a class of nonlinear systems. A system is passive if its energy is bounded from below, and is inherently stable at its natural equilibrium. In the context of state-space representation of nonlinear systems, this allows for a Lyapunov function interpretation of quantities such as allowable, stored and dissipated energy and thus provides a direct relationship between passivity and stability \cite{l2gain}. In the passivity-based control-by-interconnection methodologies, the controller can be understood as a dynamical system interconnected to the physical system that renders the closed-loop system passive \cite{OrtVanMasEsc}. Power conserving interconnection of passive systems is again passive. This has resulted in techniques called control-by-interconnection \cite{ortega2001putting}, in which we assume that controller is a passive system interconnected to the physical system, resulting in system's desired control objective and/or performance. Additionally, passivity-based controllers include sensing, actuation and do not usually require an external power source. Consequently, these are more robust and insensitive to measurements. Classical examples include lead-lag compensators, the gain-setting circuit of feedback amplifiers for voltage/current control and fly-ball governor in speed controller for windmills and steam engines which pre-dates back to 16th century \cite{maxwell1867governors}.

Energy is an intellectually deep but a simple concept, that eluded the best minds for centuries. In physics, energy represents the ability to do work, that exist in various forms such as mechanical, electrical, thermal and chemical. The behavior of a complex system can be described by analyzing the energy transfer among its subsystems. In this regard, the controller can be understood as an energy exchanging device (typically implemented on a computing system) that modifies the behavior of the plant.
Energy-based methods for modeling and control of complex physical systems has been an active area of research for the past two decades. In particular, the port-Hamiltonian based formulation has proven to be effective in modeling and control of complex physical systems from several domains, both finite- and infinite-dimensional \cite{bookgeo}. Port- Hamiltonian systems are inherently passive with the Hamiltonian (as the total energy), which is assumed to be bounded from below, serving as the storage function and the port variables being power-conjugate (force and velocity or voltage and current). This resulted in the development of so-called ``Energy Shaping'' methods for control of physical systems. The fundamental idea in energy shaping is to find a controller that 
 renders the closed-loop system passive. The term `shaping' refers to `assigning a desired energy function to the closed-loop system through control'. This often requires one to solve partial differential equations. In this context, the controller can be interpreted as a system that bridges the gap between the given open-loop and desired closed-loop energy.

%
To analyze passivity of a general nonlinear system, one needs to be crafty in constructing the storage function. To that end, recasting the dynamics into a known framework, such as port-Hamiltonian formulation has lead to passive maps with power-conjugate port-variables (such as voltage and current, Force and velocity). But the standard control by interconnection methodologies, where we assume that both plant and controller are passive, fails for the control-objectives that require bounded power but unbounded energy. In the case of resistive, inductive and capacitive (RLC) circuits, this phenomenon is usually called as {\em dissipation obstacle}. This motivated researchers to search for passive maps that are not necessarily power-conjugate. One possible alternative that has been explored extensively in the finite-dimensional case is the Brayton-Moser framework for modeling electrical networks \cite{BraMos1964I,BraMos1964II,BraMir1964}, which has been successfully adapted towards analyzing passivity of RLC circuits and for control of physical systems by ``power shaping''.

Physical systems in BM framework are modeled as pseudo-gradient systems with respect to a pseudo-Riemannian metric $A$ and a ``mixed-potential'' function $P$ which has units of power \cite{smale2000mathematical}.
\begin{eqnarray*}
A\dot{x}=\nabla_xP +Bu
\end{eqnarray*} 
where $x\in \mathbb{R}^n$ denotes the state variable, $B\in \mathbb{R}^{n\times m}$ and $u\in \mathbb{R}^m$ denotes input and input matrix respectively.
In the case of RLC circuits, the mixed-potential function is the sum of the content of the current carrying resistors, co-content of the voltage controlled resistors and instantaneous power transfer between storage elements. Unlike energy in the port-Hamiltonian formulation, the mixed-potential function in Brayton-Moser formulation is sign-indefinite. Hence, we cannot use this directly as a storage/Lyapunov function to infer any kind of passivity/stability properties. The key step to derive passivity in this framework is to find an equivalent gradient formulation with respect to a matrix $\tilde{A}$ (whose symmetric part is negative-definite) and a positive-definite mixed-potential function $\tilde{P}$. 
\begin{eqnarray*}
\tilde{A}\dot{x}=\nabla_x\tilde{P} +Bu
\end{eqnarray*}
In literature, the new mixed-potential function $\tilde{P}$ and matrix $\tilde{A}$ are together called as ``admissible-pairs''. The passive maps derived from Brayton-Moser framework directly follow from the inherent properties of these admissible-pairs. In the context of electrical networks, the passivity is now achieved with respect to `controlled voltages and derivatives of currents' or `controlled currents and derivatives of the voltages'. These new passive maps lay the groundwork for control by power-shaping methodology. Analogous to the energy-shaping, the idea is to make the closed-loop system passive by assigning a desired power-like function through control. This method has natural advantages over practical drawbacks of energy shaping methods like speeding up the transient response (as derivatives of currents and voltages are used as outputs) and also help overcome the ``dissipation obstacle''. For complete details on various energy and power-based modeling techniques, we refer to \cite{jeltsema2009multidomain}.

%
\vspace{0.5cm}
\noindent\textbf{Motivation and Contributions}

\noindent For systems formulated in the Brayton Moser framework, we aim to explore alternative passive maps and study their impact on control and optimization of dynamical system. Further, the research objectives gave rise to several publications, and are classified into three themes:

\hspace{-0.69cm}
\begin{minipage}{0.05\linewidth}
	\vspace{-5.4cm}
	(i)
\end{minipage}\begin{minipage}{0.95\linewidth}
	{\em Finite dimensional systems}: Energy shaping methods for designing controllers often suffer from dissipation obstacle. The Brayton-Moser formulation was a possible alternative to circumvent this problem. However, even in this framework, designing controllers leads to two chief difficulties. The first one involves solving partial differential equations, which might be a herculean task. We provide an alternate methodology for passive systems with an integrable output port-variable, which does not involve solving for partial differential equations. These results have been published in \cite{chinde2016building}. The second difficulty lies in the fact that the methodology requir-
\end{minipage}
\hspace{-0.5cm}\begin{minipage}{0.05\linewidth}
\vspace{-12.25cm}
(i)
\end{minipage}\begin{minipage}{0.95\linewidth}
es one to find storage functions satisfying the gradient structure. 
This has led to serious restrictions on the scope of problems that can be solved.
In \cite{ICCvdotidot} we have shown that, for systems in Brayton-Moser framework, storage functions that are constructed using Krasovskii's Lyapunov functions yields passive maps that have ``differentiation on both the port variables''. This solves the dissipation obstacle problem and avoids the need to find admissible pairs. To establish the result, we extended the property that, a class of dynamical systems in Brayton-Moser formulations are contracting. We further extended this to class of contracting nonlinear systems via dynamic state feedback \cite{ACC18}.
\end{minipage}
\vspace{0.5cm}

\hspace{-0.5cm}\begin{minipage}{0.05\linewidth}
\vspace{-14.5cm}
(ii)
\end{minipage}\begin{minipage}{0.95\linewidth}
{\em Infinite-dimensional systems}: Similar to finite dimensional systems, infinite dimensional systems also suffer from dissipation obstacle. The existing literature on boundary control of infinite dimensional systems by energy shaping, deals with either lossless systems \citep{Hugo} or partially lossless systems \citep{AleCla}, and thus avoid dissipation obstacle issues. The Brayton-Moser formulation that helped us solve these issues in finite dimensional systems is not a very established topic in infinite dimensional systems. For instance, the authors in \cite{BraMir1964} studied the stability of transmission line system with constant input and nonlinear load elements at boundary. Further, in \cite{DimVan07}, the authors provided Brayton-Moser formulation of Maxwell's equations with zero boundary energy flows. However, both the papers are limited to stability analysis and commonly avoids the boundary control problem. The basic building block to overcome dissipation obstacle is to write the equations in the Brayton-Moser form, which was not fully extended to infinite dimensional systems. However, to effectively use the method, we need to construct admissible pairs, which aids in stability analysis.  In case of infinite-dimensional systems with nonzero boundary energy flows, we need to find these admissible pairs for all individual subsystems, that is, spatial domain and boundary, while preserving the interconnection structure between these subsystems. These results are published in \cite{mtns}, \cite{IFAC}, \cite{book} and \cite{ima2018} and illustrated using `stability of Maxwell's equations' and `boundary control of transmission line system modeled by Telegraphers equations'.  
\end{minipage}
\vspace{0.5cm}

\hspace{-0.5cm}\begin{minipage}{0.05\linewidth}
	\vspace{-9.1cm}
	(iii)
\end{minipage}\begin{minipage}{0.95\linewidth}
	{\em Convex optimization}: Another set of problems where the pseudo-gradient formulation naturally arises is in gradient methods for convex optimization. Gradient-based methods are a well-known class of mathematical routines for solving convex optimization problems. These gradient algorithms have much to gain from a control and dynamical systems perspective, to have a better understanding of the underlying system theoretic properties (such as stability, convergence rates, and robustness). The convergence of gradient-based methods and Lyapunov stability, relate the solution of the optimization problem to the equilibrium point of a dynamical system. Specifically, the primal-dual methods closely resemble the pseudo-gradient structure. Moreover, Brayton-Moser formulation is inherently a pseudo-gradient formulation.  This observation motivates us to look for connections between convex optimization and Brayton-Moser formulation. These results are published in \cite{kosaraju2018stability}.
\end{minipage}
\vspace{0.1cm}

\noindent\textbf{Outline of thesis} This thesis is subdivided into six chapters which are structured as follows:

\vspace{0.25cm}
\hspace{-0.5cm}\begin{minipage}{0.05\linewidth}
\vspace{-1.6cm}
(i)
\end{minipage}\begin{minipage}{0.95\linewidth}
Chapter 2 accommodates most of the prerequisite and background information.
It contains a brief outline on modeling and control aspects in port-Hamiltonian and Brayton-Moser formulations.
\end{minipage}
\vspace{0.5cm}

\hspace{-0.5cm}\begin{minipage}{0.05\linewidth}
\vspace{-4.65cm}
(ii)
\end{minipage}\begin{minipage}{0.95\linewidth}
In Chapter 3, we show that the output port-variable, derived from systems modeled in Brayton Moser framework,   are integrable; under the assumption that the input matrix is integrable. The integrated output port-variable is then used to construct a desired storage function for the closed-loop system. Further, we show that a class of Brayton Moser systems are contracting, resulting in a new passivity property with ``differentiation at both port-variables''. We extended this methodology to a class of nonlinear systems using dynamic feedback and Krasovskii's method.
\end{minipage}
\vspace{0.5cm}

\hspace{-0.5cm}\begin{minipage}{0.05\linewidth}
\vspace{-6.9 cm}
(iii)
\end{minipage}\begin{minipage}{0.95\linewidth}
In Chapter 4, we establish that infinite-dimensional systems are prone to dissipation obstacle. Thereafter, we begin with Brayton-Moser formulation of port-Hamiltonian system defined using Stokes' Dirac structure. In the process, we present its Dirac formulation with a non-canonical bilinear form. Analogous to the finite-dimensional system, identifying the underlying gradient structure of the system is crucial in analyzing the stability. We illustrate this with two examples, (i) stability analysis of Maxwell's equations in $\mathbb{R}^3$ with zero boundary energy flows, (ii) boundary control of transmission line system modeled by Telegraphers equations. Towards the end, we extend the results presented in chapter 3 to infinite-dimensional systems.

\end{minipage}
\vspace{0.5cm}

\hspace{-0.5cm}\begin{minipage}{0.05\linewidth}
\vspace{-7.4cm}
(iv)
\end{minipage}\begin{minipage}{0.95\linewidth}
In Chapter 5, we deal with stability of continuous time primal-dual gradient dynamics of convex optimization problem. Primarily, the convex optimization problem with only affine equality constraints admits a Brayton Moser formulation. Secondly, the inequality constraints are modeled as a state dependent switching system. Finally, the two systems are shown as passive  systems and are interconnected in a power conserving way. This results in a new passive system whose dynamics represents the primal-dual gradient equations of the overall optimization problem. The aforementioned methodology is applied to an support vector machine problem and  simulations are provided for corroboration.

\vspace{0.5cm}
\hspace{-0.5cm}\begin{minipage}{0.05\linewidth}
(v)
\end{minipage}\begin{minipage}{0.95\linewidth}
In Chapter 6, we give concluding remarks and present some future directions.
\end{minipage}
%
\end{minipage}
\vspace{0.5cm}

\chapter{System theoretic Prerequisites}
%
In this chapter, we review few important results from the literature on port-Hamiltonian systems, Brayton Moser formulation, and their geometric properties and limitations. The list contains results that have directly shaped our work that we present in this thesis.  The list is by no means complete.
We start with a input-output port-Hamiltonian system and its Dirac formulation, subsequently we present control by interconnection methodology and its drawback, `dissipation obstacle'. Later on, we introduce Brayton Moser formulation of finite-dimensional topologically complete RLC circuits and presents some results on stability and control in this framework. Throughout the chapter, we illustrate these concepts using a parallel RLC circuit as an example. Towards the end, we advance to infinite-dimensional port-Hamiltonian systems and conclude with some general stability definitions for infinite dimensional systems.
\section{Port-Hamiltonian (pH) system and Dirac structures}
A port-Hamiltonian system with dissipation evolving on an n-dimensional state space manifold $\mathcal{X}$ with input space $\mathcal{U}= \mathbb{R}^m$  and output space $\mathcal{Y}= \mathbb{R}^m$ ($m\leq n$) is represented as
\begin{align}\label{PHDS}
\begin{split}
\dot{x}&=\left[J(x)-R(x)\right]\dfrac{\partial H}{\partial x}+g(x) u\\
y&=g^\top (x)\dfrac{\partial H}{\partial x}
\end{split}
\end{align}
where $x\in \mathcal{X}$ is the energy variable and the smooth function $H(x):\mathcal{X}\rightarrow \mathbb{R}$ represents the total stored energy, otherwise called as Hamiltonian. $u\in \mathcal{U}$ and $y\in \mathcal{Y}$ are called input and output port-variables respectively. The $n\times n$ matrices $J(x)$ and $R(x)$ satisfies $J(x)=-J^\top (x)$ and $R(x)=R^\top(x)\geq 0$. The input matrix $g(x)\in \mathbb{R}^{n\times m}$ and skew-symmetric matrix $J(x)$ capture the system's interconnection structure, where as, the positive semi-definite matrix $R(x)$ captures the dissipation (or resistive) structure in the system. By the properties of $J(x)$ and $R(x)$, it immediately follows that
\begin{align}\label{dPHDS}
\begin{split}
\dfrac{d}{dt}H(x)&= \dot{x}^\top \dfrac{\partial H}{\partial x}\\
&= \left(\left[J(x)-R(x)\right]\dfrac{\partial H}{\partial x}+g(x) u\right)^\top \dfrac{\partial H}{\partial x}\\
&= -\dfrac{\partial H}{\partial x}^\top R(x)\dfrac{\partial H}{\partial x}+u^\top g^\top \dfrac{\partial H}{\partial x}\\
&\leq u^\top y
\end{split}
\end{align}
This implies that the pH system \eqref{PHDS} is passive with port-variables $u$ and $y$. As the Hamiltonian $H$ represents the total energy stored in the system, $\dot{H}$ represents the instantaneous power transfer. Moreover, $u$ and $y$ are called power-conjugate port-variables, meaning, their product $u^\top y$ represents the power flow between the environment and the system. Well-known examples of such pairs are voltage-current in electrical circuits and force-velocity in mechanical systems. Consequently, the equation \eqref{dPHDS} gives the {\em dissipative} inequality
\begin{align}\label{iPHDS}
\begin{split}
\underbrace{H(x(t_1))-H(x(t_0))}_{\text{\large Stored\;\;energy}}&=  \underbrace{-\int_{t_0}^{t_1}\left(\dfrac{\partial H}{\partial x}^\top R(x)\dfrac{\partial H}{\partial x}\right)dt}_{\text{\large Dissipated energy}}+ \underbrace{\int_{t_0}^{t_1}u^\top ydt}_{\text{\large Supplied energy }}
\end{split}
\end{align}
where $t_0\leq t_1$ and $H$, $u^\top y$ represent the storage function and the supply rate respectively.

\textbf{\em Dirac Structures}:
In network theory, the Tellegen's theorem states that the summation of instantaneous power in all the branches of the network is zero. Dirac structure generalizes the underlying geometric structure of Tellegen's theorem (power conservation). Let $\mathcal{F}\times \mathcal{E}$ be the space of these power variables, where the linear space $\mathcal{F}$ is called flow space and $\mathcal{E}=\mathcal{F}^\ast$ is the dual space of $\mathcal{F}$, called as effort space. In the case of network theory, these spaces can be interpreted as spaces of branch voltages and currents, vice-versa (in the case of mechanical systems, they represent generalized forces and generalized velocities). Let $f\in \mathcal{F}$ and $e\in \mathcal{E}$ denotes the flow and effort variables respectively. The power in the total space of port variables ($\mathcal{F}\times \mathcal{E}$) can be defined as 
\beq
P=\left<e|f\right>,\;\; (f,e)\in\mathcal{F}\times \mathcal{E}
\eeq
where $\left<e|f\right>$ denotes the duality product, that is, the linear functional $e\in \mathcal{E}$ acting on $\mathcal{F}$. In the case of $\mathcal{F}=\mathbb{R}^m$ 
\beqn
\left<e|f\right>&=&e^\top f\eeqn
that is, the duality product can be identified with the inner-product defined on $\mathbb{R}^m$. 
\begin{definition}\cite{van2014port}
	Consider a finite-dimensional linear space $\mathcal{F}$ with $\mathcal{E}=\mathcal{F}^\ast$. A subspace $D\subset \mathcal{F}\times \mathcal{E}$ is a (constant) Dirac structure if
	\begin{itemize}
		\item [1.]$\left<e|f\right>=0$, for all $(f,e)\in \mathcal{D}$ \;\;\;\; (Power conservation)
		\item[2.]dim $\mathcal{D}$ $=$ dim $\mathcal{F}$\;\;\;\;\;\;\;\;\;\;\;\;\;\;\;\; (maximal dimension of subspace $\mathcal{D}$ )
	\end{itemize}
\end{definition}
\noindent Next, we present an equivalent definition for the Dirac structure, which will be useful in presenting the Dirac formulation of infinite dimensional systems in Chapter 5.
\begin{lemma}
	A (constant) Dirac structure on $\mathcal{F}\times \mathcal{E}$ is a subspace $\mathcal{D}\subset \mathcal{F}\times \mathcal{E}$ such that 
	\beq
	\mathcal{D}&=&\mathcal{D}^{\perp}
	\eeq
	where $\perp$ denotes the orthogonal complement with respect to the bilinear form $\left<\left<\;,\right>\right>$ given as
	\beq
	\left<\left<(f^a,e^a),(f^b,e^b)\right>\right>&=&\left<e^a|f^b\right>+\left<e^b|f^a\right>, \;\;\;\;(f^a,e^a),(f^b,e^b)\in \mathcal{F}\times \mathcal{E}
	\eeq
	or equivalently
	\beq
	D^{\perp}:=\left\{\left(f^a,e^a\right)|\left<\left<(f^a,e^a),(f,e)\right>\right>=0,(f,e)\in \mathcal{D}\right\}
	\eeq
\end{lemma}
\noindent For further exposition on Dirac structures and their alternate representation see \cite{van2014port,l2gain, courant1990dirac, dorfman1993dirac,dalsmo1998representations}.
%
%
%
%
\section{Control by interconnection}
The typical approach in control of physical systems is about choosing a controller that constraints the time derivative of the Lyapunov function candidate to be a negative (semi-) definite function. The form of the controller has very little to do with model or the physics of the plant, but more on the choice of the Lyapunov function candidate. Control objective with a performance criterion cannot be easily incorporated using this methodology. In the passivity-based control (PBC) methodologies, the controller can be understood as an aggregation of proportional, derivative and integral actions, thus providing a direct relation to the performance criteria. Further, the storage function, which acts as a Lyapunov function for stability analysis, is derived from the physics of the plant. The fundamental idea in PBC methodologies is to find a controller that renders the closed-loop system passive. In this section we briefly present a PBC methodology called control by interconnection \cite{garcia2005control,castanos2009asymptotic,ortega2007control,ortega2008control} and study its limitations.
%

In control by interconnection, we assume that the controller is a port-Hamiltonian system with dissipation (a passive dynamical system usually implemented on a computer) 
\begin{align}\label{cPHDS}
\begin{split}
\dot{x}_c&=\left[J(x_c)-R(x_c)\right]\dfrac{\partial H_c}{\partial x_c}+g_c(x_c) u_c\\
y_c&=g_c^\top (x_c)\dfrac{\partial H_c}{\partial x_c}
\end{split}
\end{align}
interconnected to the physical system \eqref{PHDS} using a standard feedback interconnection
\beqn
\begin{matrix}
	u=-y_c+v&u_c=y+v_c
\end{matrix}
\eeqn
such that the closed loop system
\begin{align}\label{cibPHDS}
\begin{split}
\begin{bmatrix}
\dot{x}\\
\dot{x}_c
\end{bmatrix}&=\begin{bmatrix}
J(x)-R(x) &-g(x)g_c^\top (x_c)\\g_c(x_c)g^\top(x) &J(x_c)-R(x_c)
\end{bmatrix}\begin{bmatrix}
\frac{\partial H}{\partial x}\\\frac{\partial H_c}{\partial x_c}
\end{bmatrix}+\begin{bmatrix}
g(x)&0\\0&g_c(x_c)
\end{bmatrix}\begin{bmatrix}
v\\v_c
\end{bmatrix}\\
\begin{bmatrix}
y\\y_c
\end{bmatrix}&=\begin{bmatrix}
g(x) &0\\0&g_c(x_c)
\end{bmatrix}^\top\begin{bmatrix}
\frac{\partial H}{\partial x}\\\frac{\partial H_c}{\partial x_c}
\end{bmatrix}
\end{split}
\end{align}
is again a port Hamiltonian system with dissipation. Next, we find the invariant functions, called Casimirs $\mathcal{C}(x,x_c)$, that are independent of the closed-loop Hamiltonian $H(x)+H_c(x_c)$, using
\begin{align}\label{casPHDS}
\begin{split}
\begin{bmatrix}
\frac{\partial \mathcal{C}}{\partial x}\\
\frac{\partial \mathcal{C}}{\partial x_c}
\end{bmatrix}^\top\begin{bmatrix}
J(x)-R(x) &-g(x)g_c^\top (x_c)\\g_c(x_c)g^\top(x) &J(x_c)-R(x_c)
\end{bmatrix}&=0.
\end{split}
\end{align}
These Casimirs, that relate the plant state to controller state, are used to shape the closed-loop Hamiltonian  at the desired operating point by replacing $H(x)+H_c(x_c)$ by $H(x)+H_c(x_c)+H_a(\mathcal{C})$. Further, if the Casimirs are of the form 
\beq \label{cas::clp_inv} \mathcal{C}(x,x_c)=x_c-C(x)\eeq
then we can eliminate $x_c$ from  \eqref{cibPHDS} by restricting the closed-loop dynamics to the level set $L=\left\{(x,x_c)|x_c=C(x)+c\right\}$, where $c\in \mathbb{R}$ is a constant. Thereby, we can use $H(x)+H_c(C(x)+c)$ as the new storage function. 

There are mainly two disadvantages in using this methodology. The first one being, the need for solving  partial differential equations given in \eqref{casPHDS} to find the Casimir functionals. This often turns out to be a herculean task. The second disadvantage lies in the existence of the Casimir functional itself. The partial differential equations in \eqref{casPHDS} can be simplified (using \eqref{cas::clp_inv}) to the following set of necessary conditions.
%
\beq
\dfrac{\partial C}{\partial x}^\top J(x)\dfrac{\partial C}{\partial x}&=&J_c(x_c)\\
R(x)\dfrac{\partial C}{\partial x}&=& 0\label{cas::diss_obs}\\
R_c(x_c)&=&0\\
\dfrac{\partial C}{\partial x}^\top J(x)&=&g_c(x_c)g(x)^\top.
\eeq
In the necessary conditions given above, one that hinders us most often is $R(x)\dfrac{\partial C}{\partial x}= 0$. Let us consider a scenario where the $i^{th}$ coordinate $x_i$ of state vector $x$ needs to be controlled. Further assume that the resistive structure of the system imposes $R(x_i)\neq 0$. Then from equation \eqref{cas::diss_obs}, we have
$$R(x_i)\dfrac{\partial C}{\partial x_i}=0\implies \dfrac{\partial C}{\partial x_i}=0.$$
%
This implies that the achievable Casimirs are independent of $x_i$. Hence, $x_i$ cannot be controlled by this methodology.
%
In the case of RLC circuits this is usually called  as the {\em dissipation obstacle}. In the next subsection, we present an equivalent physical interpretation of the dissipation obstacle.
\section{Dissipation Obstacle}
In standard control by interconnection methodologies \citep{ortega2008control}, we assume that both plant and controller \kcn{are} passive. Plants that extract unbounded energy (but bounded power) at nonzero equilibrium, cannot be stabilized under this assumption. 
%
The following example better illustrates this limitation of control by interconnection methodology.
\begin{example}\textbf{(Parallel RLC circuit)}.\label{moti::example_pRLC}
	\begin{figure}[h]
		\centering
		\includegraphics[width=0.6\linewidth]{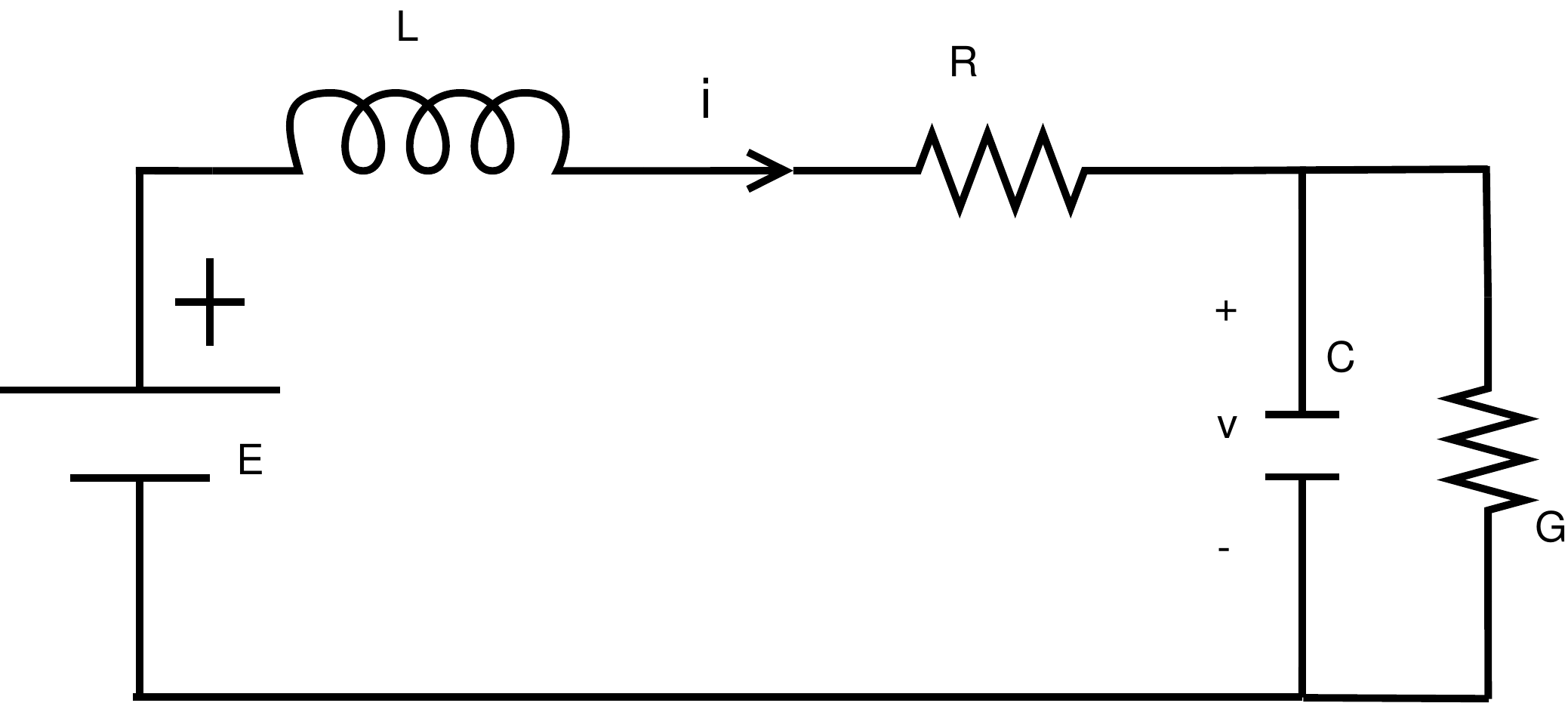}
		\caption{Parallel RLC circuit}
		\label{fig:pRLC}
	\end{figure}
 	Consider the parallel RLC circuit (as shown in Figure \ref{moti::example_pRLC}) with charge $q$ across the capacitor $C$ and flux $\phi$ through the inductor $L$ as the state variables. The dynamics of this system in port-Hamiltonian formulation \eqref{PHDS} with state variables $(q,\phi)$ is
	\beq\label{moti::pRLC_eqn}
	\begin{bmatrix}
		\dot{q}\\ \dot{\phi}
	\end{bmatrix}=\left(\begin{matrix}
	\begin{bmatrix}
		0&-1\\1&0
	\end{bmatrix}&-&\begin{bmatrix}
	R&0\\0&G
\end{bmatrix}
\end{matrix}\right)\begin{bmatrix}
\dfrac{q}{C}\\ \dfrac{\phi}{L}
\end{bmatrix}+\begin{bmatrix}
0\\1
\end{bmatrix}V_s
	\eeq
where $R$ is the series resistance of the inductor $L$, $G$ is the conductance of the capacitor $C$ and $V_s$ is the voltage source. It can be shown that this system is passive with total energy 
\beq \label{ex2.1::ham}
H(q,\phi)=\dfrac{q^2}{2C}+\dfrac{\phi^2}{2L},
\eeq
as storage function and port variables being input $V_s$ and output $i=\dfrac{\phi}{L}$, that is, 
\begin{eqnarray}\label{ex2.1::Hdot}
\dot H\leq V_si.
\end{eqnarray}
We now have the following dissipation inequality
\beq
\underbrace{H(t_1)-H(t_0)}_{stored\;\; energy}= \underbrace{-\int_{t_0}^{t_1}\left(Ri(\tau)^2+Gv(\tau)^2\right)d\tau}_{dissipated\;\; energy} + \underbrace{\int_{t_0}^{t_1}V_s(\tau)i(\tau)d\tau}_{supplied\;\; energy}
\eeq
where $i$ denotes the current through the inductor $L$, and $v$ denotes the voltage across the capacitor $C$. Further, at a non-zero operating point $(v^\ast,i^\ast)$, we have a non-zero supply rate $V_s^\ast i^\ast \neq 0$. This implies that the energy supplied through the controller at the operating point is non-zero, given by
\beq \label{ex2.1::diss_ineq}
\underbrace{\int_{t_0}^{t_1}V^{\ast}_s(\tau)i^{\ast}(\tau)d\tau}_{supplied\;\; energy}=\underbrace{\int_{t_0}^{t_1}\left(Ri^{\ast}(\tau)^2+Gv^{\ast}(\tau)^2\right)d\tau}_{dissipated\;\; energy} .
\eeq
This further indicates that the controller should have an unbounded energy to stabilize the system, violating the assumption that the controller is a passive dynamical system. 
\begin{remark}
From a physics point of view, regulating the current in the inductor $L$ to $i^\ast$ is equivalent to storing $\dfrac{1}{2}L(i^\ast)^2$ energy. Similarly, for the capacitor $C$, regulating voltage to $v^\ast$ is equivalent to storing $\dfrac{1}{2}C(v^\ast)^2$ energy. For instance, let us assume that we have pumped enough energy through the controller to a point where the capacitor and inductor have stored the desired energy, and disconnected the controller. Due to the existence of the resistive elements $R$ and $G$ in the circuit, the energy stored in the circuit dissipates through them. We therefore need to compensate the dissipated energy by supplying it through the controller (given in \eqref{ex2.1::diss_ineq}). This analysis indicates that the limitations in control by interconnection methodology is predicated by resistive structure in the plant. We can now corroborate this from necessary conditions presented in \eqref{cas::diss_obs} for the existence of closed-loop Casimir functional, that is,
\beqn
\begin{bmatrix}
	R&0\\0& G
\end{bmatrix}\begin{bmatrix}
\frac{\partial C}{\partial q}\\\frac{\partial C}{\partial \phi}
\end{bmatrix}=0
\eeqn
which implies $C$ should be independent of the state variable $(q,\phi)$. 
\end{remark}
\end{example}
\begin{remark}
Equation \eqref{ex2.1::diss_ineq} points that, at the operating point, the controller is supplying unbounded energy to the plant but a constant power ($V^\ast_si^\ast$). 
This motivated researchers to look for passive maps with power as storage function. Brayton-Moser is one such framework that provides storage functions related to power. In the next section, we briefly outline the modeling and control aspects of finite-dimensional systems in Brayton Moser formulation.
\end{remark}
\section{Brayton-Moser formulation}
It is well-known that port-Hamiltonian formulation naturally arises as a modeling framework for larger class of physical systems, such as mechanical, electrical and electro-mechanical systems. Another important modeling methodology that has been widely used for RLC networks is Brayton-Moser framework. In this framework, we model the system in pseudo-gradient form using a function, called mixed potential function, which has units of power. The advantage of modeling systems in this framework is that, it presents us a new family of storage functions (derived from mixed-potential function), that can be used to obtain new passive maps. In this section, we present modeling and control of finite dimensional systems in Brayton-Moser framework. The exposition presented here is extracted from \cite{bookgeo,Guido,van2014port,blankenstein2003joined,ElosiaDimOrt,van2011relation}, will be helpful in presenting the Brayton-Moser formulation of infinite dimensional systems in Chapter 5.
\subsection{Energy to co-energy formulation}
In port-Hamiltonian modeling, the dynamics are derived using energy variables; where as in Brayton Moser framework, we model the system using co-energy variables. In the case of network theory, generalized flux and charge represent energy variables; where as, generalized voltages and currents denote the co-energy variables. In this aspect, Brayton-Moser formulation is usually called as co-energy formulation \cite{jeltsema2009multidomain}. Given a port-Hamiltonian system \eqref{PHDS} with energy variable $x$ and Hamiltonian $H(x)$, we define the co-energy variable 
\beq z:=\dfrac{\partial H}{\partial x}.\eeq
Suppose that the mapping between energy variable $x$ and co-energy variables $z$ is invertible, such that
\beq\label{e-coe}
x=\dfrac{\partial H^\ast}{\partial z}(z)
\eeq
where $H^\ast(z)$ represents the co-Hamiltonian, defined through the Legendre transformation of $H(x)$, given by
\beqn
H^\ast(z)&=&z^\top x-H(x).
\eeqn
Differentiating \eqref{e-coe} and using \eqref{PHDS}, we get
\begin{align}\label{co-e_form1}
\begin{split}
\dfrac{\partial^2H^\ast}{\partial z^2}(z)\dot{z}&=[J(x)-R(x)]z+g(x)u\\
y&=g^\top(x) z
\end{split}
\end{align}
 Assume that there exist coordinates $x_1$ and $x_2$ ($x=(x_1,x_2)$), such that the Hamiltonian $H(x)$ can be split as $H_1(x_1)+H_2(x_2)$. Consequently the co-Hamiltonian can also be split as $H_1^\ast(z_1)+H^\ast_2(z_2)$ (where $z=(z_1, z_2)$). Further, assume that 
\beqn
J(x)=\begin{bmatrix}
	0&-B(x)\\B^\top(x) &0
\end{bmatrix}, R(x)=\begin{bmatrix}
R_1(x)&0\\0&R_2(x)
\end{bmatrix}, g(x)=\begin{bmatrix}
g_1(x)\\0
\end{bmatrix}
\eeqn
and there exist functions $P_1(z_1)$ and $P_2(z_2)$ such that
\beqn
R_1(x)z_1&=& \dfrac{\partial P_1}{\partial z_1}(z_1),\\
-R_2(x)z_2&=& \dfrac{\partial P_2}{\partial z_2}(z_2).
\eeqn
Then the system of equations \eqref{co-e_form1} can be written in the pseudo-gradient form
\begin{align}\label{co-e_form2}
\begin{split}
\begin{bmatrix}
-\frac{\partial^2H_1^\ast}{\partial z_1^2}(z)&0\\0&\frac{\partial^2H_2^\ast}{\partial z_2^2}(z)
\end{bmatrix}\begin{bmatrix}
\dot{z}_1\\ \dot{z}_2
\end{bmatrix}&=\begin{bmatrix}
\frac{\partial P}{\partial z_1}\\
\frac{\partial P}{\partial z_2}
\end{bmatrix}+\begin{bmatrix}
-g_1(x)\\0
\end{bmatrix}u\\
y&=g_1^\top(x) z_1
\end{split}
\end{align}
where $P(z)=P_1(z_1)+P_2(z_2)+z_1^\top B(x)z_2$. If $g_1(x)$ and $B(x)$ are constant, then the equations \eqref{co-e_form2} are independent of the energy variable. In this case, the above system of equations closely represents a pseudo-gradient structure.
%
%
%
\subsection{Topologically complete RLC circuits}
\label{sec::BM form of fin RLC}
In this section, we briefly outline the Brayton-Moser formulation of topologically complete RLC circuits and present the underlying geometric structures. The word `topologically complete', indicates that the state space representation of the RLC circuit is completely determined by inductor currents and capacitor voltages.
%
Brayton and Moser in the early sixties \citep{BraMos1964I, BraMos1964II} showed that the dynamics of a 
class (topologically complete) of nonlinear $RLC$-circuits can be
written as
\begin{align}
A (i_L,v_C)
\begin{bmatrix}
\frac{di_L}{dt} \\ \frac{dv_c}{dt}
\end{bmatrix} =
\begin{bmatrix}
\frac{\partial P}{\partial i_L} \\
\frac{\partial P}{\partial v_C}
\end{bmatrix} +
\begin{bmatrix}
B_{E_c}^\top E_c \\
-B_{J_c}^\top J_C
\end{bmatrix}
\label{eq_bm_fd1}
\end{align}
where $i_L\in \mathbb{R}^{n_L}$ and $v_C\in \mathbb{R}^{n_C}$ represent vectors of currents through inductors and voltages across capacitors respectively. $n_L$ and $n_C$ denote the number of inductors and capacitors in the network. $E_C, J_C$ are
respectively the controlled voltage and current sources respectively.  $A(i_L,v_C) = \text{diag} \{L(i_L), -C(v_C) \}$ where $L(i_L)\in \mathbb{R}^{n_L\times n_L}$ and $C(v_C)\in \mathbb{R}^{n_C\times n_C}$ denote inductance and capacitance matrices respectively (both are positive definite matrices).
%
The input matrices $B_{E_c}\in \mathbb{R}^{n_E\times n_L}, B_{J_c}\in \mathbb{R}^{n_J\times n_C}$ (containing elements from the set $\{-1, 0 , 1
\}$) are given by Kirchoff's voltage and current laws. $n_J$ and $n_E$ denote the number of current and voltage sources in the network respectively. $E_C, J_C$ are
respectively the controlled voltage and current sources. $P(i_L,v_C):\mathbb{R}^{n_L\times n_C}\rightarrow \mathbb{R}$ is
called the mixed potential function, defined by
\[
P(i_L,v_C) = F(i_L) - G(v_C) + i_L^\top \gamma\; v_c
\]
Here, $F$ denotes the content of
all the current controlled resistors, $G$ denotes the co-content of all
voltage controlled resistors and $\gamma$ is a skew-symmetric matrix containing elements from $\{-1, 0 , 1 \}$, and represents the network topology.
%
%
As an example, we next present the Brayton-Moser formulation of parallel RLC circuit given in Figure \ref{fig:pRLC}.
\begin{example}\textbf{(Parallel RLC circuit cont'd)}.\label{Example::ch2::prlc::BM}
	Consider the parallel RLC circuit of Figure \ref{moti::example_pRLC}. Let $i$ denote the current through the inductor $L$ and $v$ denote the voltage across the capacitor. The pair $(i,v)$ denotes the co-energy variables.  The kirchhoff voltage and current laws  
		\beq \label{moti::pRLC_eqn}
	\begin{matrix}
		-L\dfrac{di}{dt}&=& Ri+v-V_s\\C\dfrac{dv}{dt}&=&i-Gv
	\end{matrix}
	\eeq
	can be written in Brayton Moser form \eqref{eq_bm_fd1} as
	\begin{align}
	A \begin{bmatrix}
	\frac{di}{dt} \\ \frac{dv}{dt}
	\end{bmatrix} =
	\begin{bmatrix}
	\frac{\partial  P}{\partial i} \\
	\frac{\partial  P}{\partial v}
	\end{bmatrix} +
	\begin{bmatrix}
	-1\\0
	\end{bmatrix}V_s
	\label{BM finite dimenional example}
	\end{align}
	with $A=\text{diag}~\{ -L,C \}$ and
	$P(i,v)$ as the mixed potential function (power function) given by
	\beqn
	P(i,v)&=& -\dfrac{1}{2}Gv^2+vi+\dfrac{1}{2}Ri^2.
	\eeqn
    where $\dfrac{1}{2}Ri^2$ denotes the content of the current controller resistor $R$, $\dfrac{1}{2}Gv^2$ denotes the co-content of the voltage controlled resistor $\dfrac{1}{G}$, and $vi$ represents the instantaneous power transfer between the capacitor and the inductor.
    \end{example}
    \subsection{ Dirac formulation}
   We now present the equivalent Dirac formulation of Brayton-Moser equations of finite-dimensional RLC circuits given in \eqref{eq_bm_fd1} \cite{Guido,fortney2010dirac}.
    Denote by $f\in \mathcal{F}\in \mathbb{R}^{n_L+n_C}$ the space of flows,  $e\in \mathcal{E}:=\mathcal{F}^{\ast}$ the space of efforts, $u\in \mathcal{U}$ the space of input port-variables and  $y\in \mathcal{Y}:=\mathcal{U}^\ast$ the space of output port-variables. 
    Consider the following subspace
\begin{align}\label{Dirac_gen_struc_finite_dim}
\mathcal D = \left \{(f,u, e,y) \in \mathcal F   \times \mathcal U\times \mathcal E \times \mathcal Y :  -A f =  e+B u,~y= -B^\top f\right \}
\end{align}
where  $A = \text{diag} \{L(i_L), -C(v_C) \}$, $B=\text{diag} \{B^\top_{E_c}, -B^\top_{J_c} \}$.
The above defined subspace constitutes a noncanonical Dirac structure, that is $\mathcal{D}=\mathcal{D}^\perp$, $\mathcal{D}^{\perp}$ is the orthogonal complement of $\mathcal{D}$ with respect to the noncanonical bilinear form\\
$<<(f^1,u^1,e^1,y^1),(f^2,u^2,e^2,y^2) >>$
\beq
\hspace{-6mm}= &\hspace{-2mm}\left<e^1|f^2\right>+ \left<e^2|f^1\right>+ \left <f^1 |   (A+A^\top)f^2 \right > + \left<u^1|y^2\right>+\left<u^2|y^1\right>
\label{bileniarform_fin_dim}
\eeq
for $i=1,2~$;
$
	\begin{matrix}
		f^i  \in \mathcal F,  &u^i\in  \mathcal U ,&	e^i  \in \mathcal E, &  y^i\in  \mathcal Y 
	\end{matrix}.
	$
    
    The Brayton-Moser equations \eqref{eq_bm_fd1} can be equivalently described as a dynamical system with respect to the noncanonical Dirac structure $\mathcal{D}$ in \eqref{Dirac_gen_struc_finite_dim} by setting 
    \begin{align}
f=- \begin{bmatrix}
\frac{di_L}{dt} \\ \frac{dv_c}{dt}
\end{bmatrix}, u=\begin{bmatrix}
E_c \\ J_c
\end{bmatrix}, e =
\begin{bmatrix}
\frac{\partial P}{\partial i_L} \\
\frac{\partial P}{\partial v_C}
\end{bmatrix}\;\textrm{and}\; y=
\begin{bmatrix}
-\frac{di_{E_c}}{dt} \\ \frac{dv_{J_c}}{dt}
\end{bmatrix}
\label{eq_bm_fd1a}
\end{align}
where $i_{E_c}$ denotes the current through the voltage sources $E_c$ and $v_{J_c}$ denotes the voltage across the current sources $J_c$. Since the bilinear form \eqref{bileniarform_fin_dim} is non-degenerate, $\mathcal{D}=\mathcal{D}^\perp$ implies 
\beq
<<(f,u,e,y),(f,u,e,y) >>~=~0, ~~\forall (f,u,e,y)\in \mathcal{D}.
\eeq
The bilinear form can further be simplified as
\beq\label{fin_dim_sec_2.43}
\left<e|f\right>+ \left<e|f\right>+ \left <f |   (A+A^\top)f \right > + \left<u|y\right>+\left<u|y\right>&=&0\nonumber \\
\left<e|f\right>+\dfrac{1}{2} \left <f |   (A+A^\top)f \right > + \left<u|y\right>&=&0.
\label{bileniarform_fin_dima}
\eeq
Further, using \eqref{eq_bm_fd1a} in \eqref{bileniarform_fin_dima} gives us the ``balance equation''
\beqn
-\begin{bmatrix}
\frac{\partial P}{\partial i_L} \\
\frac{\partial P}{\partial v_C}
\end{bmatrix}^\top 
\begin{bmatrix}
\frac{di_L}{dt} \\ \frac{dv_c}{dt}
\end{bmatrix}+\dfrac{1}{2}\begin{bmatrix}
\frac{di_L}{dt} \\ \frac{dv_c}{dt}
\end{bmatrix}^\top (A+A^\top)\begin{bmatrix}
\frac{di_L}{dt} \\ \frac{dv_c}{dt}
\end{bmatrix}+\begin{bmatrix}
E_c \\ J_c
\end{bmatrix}^\top \begin{bmatrix}
-\frac{di_{E_c}}{dt} \\ \frac{dv_{J_c}}{dt}
\end{bmatrix}=0
\eeqn
i.e.,
\beq \label{Pdot_fin_dim_ex}
\dot{P}&=&\dfrac{1}{2}\dot{x}^\top (A(x)+A^\top(x) )\dot{x}+u^\top y
\eeq
where $x=(i_L,v_C)$. 
\begin{remark}
In the case of parallel RLC circuit considered in Example \ref{Example::ch2::prlc::BM}, the time derivative of the mixed potential function yields $$\dot{P}=-L\dfrac{di}{dt}^2+C\dfrac{dv}{dt}^2-V_s\dfrac{di}{dt}.$$ One can note that this is not a conserved quantity, not even for $R=G=0$ and $u=0$. That is, mixed potential functional is not conserved, even with zero dissipation and zero power supply.
\end{remark}
%
%
%
%
\subsection{Admissible pairs and stability}
In general, systems in Brayton-Moser framework are modeled as pseudo-gradient systems. The standard representation of a pseudo-gradient system is 
\beq\label{psuedo_grad_sec2.4.4} 
A(x)\dot{x}=\nabla_xP+g(x)u
\eeq
where $x$ denotes the state vector, $A(x)\in \mathbb{R}^{n\times n}$ denotes a pseudo-Riemannian metric (indefinite), $P(x):\mathbb{R}^n\rightarrow \mathbb{R}$, the matrix $g(x)\in \mathbb{R}^{n\times m}$ denotes the input matrix and $u\in \mathbb{R}^m$. If $A(x)$ is positive definite, we call the system \eqref{psuedo_grad_sec2.4.4} a gradient system. One can note that the topologically complete RLC circuits given in \eqref{eq_bm_fd1} take the pseudo-gradient  structure \eqref{psuedo_grad_sec2.4.4} with  $x=(i_L,v_C)$, $g(x)=B$ 
and $u=(E_c,J_c)$. The benefit of modeling a system in pseudo-gradient form is that the function $P$ can be used as a Lyapunov candidate. The time-derivative of $P$ along the trajectories of \eqref{psuedo_grad_sec2.4.4} is
\beq\label{dot_P_sec_2.2.4}
\dfrac{d}{dt}P(x)&=&\nabla_xP ^\top \dot{x}\nonumber\\
&=& \left(A(x)\dot{x}-g(x)u\right)^\top \dot{x}\nonumber \\
&=& \dot{x}^\top A(x)\dot{x}+u^\top y
\eeq
where $y=-g^\top(x) \dot{x}$. From equation \eqref{dot_P_sec_2.2.4}, we can conclude that the system is passive if $P\geq 0$ and $(A(x) + A^\top(x))\leq 0$, with $P$ as the storage function and $u^\top y$ as the supply rate. 
%
In case, $(A(x) + A^\top(x))\le 0$ is not satisfied (see $P$ and $A$ in  parallel RLC circuit  Example \ref{Example::ch2::prlc::BM}), then it is possible
to find new $(\tilde A, \tilde P)$, called an ``admissible pair", (refer \cite{jeltsema2003passivity})
satisfying $(\tilde A(x) + \tilde A^\top(x))\le 0$.  The dynamics \eqref{psuedo_grad_sec2.4.4} can then be
equivalently be written as
\beq\label{admissible_sec_2.4.4}
\tilde{A}\dot{x}&=&\nabla_x\tilde{P}+\tilde{g}(x)u
\eeq
The authors in \cite{BraMos1964I,BraMos1964II,jeltsema2003passivity} have shown that 
\beq\label{ptilde_sec_2.4.4}
\tilde{P}=\lambda P+\dfrac{1}{2}\nabla_xP^\top M \nabla_xP
\eeq 
and 
\beq\label{atilde_sec_2.4.4}
\tilde{A}=\left(\lambda I+\nabla_x^2PM\right)A
\eeq
satisfy the gradient structures \eqref{admissible_sec_2.4.4}. Further, $\lambda \in \mathbb{R}$ and $M\in \mathbb{R}^{n \times n}$ are chosen such that $\tilde{P} \geq 0$ and $\tilde{A}+\tilde{A}^\top \leq 0$.
%
We now	present results on control by power shaping, by finding admissible pairs for the parallel RLC circuit in Example \ref{Example::ch2::prlc::BM} \citep{bookgeo}.
\begin{example}\textbf{(Parallel RLC circuit cont'd)}.\label{Example::ch2::prlc::BM_pass}
In the Brayton-Moser formulation of parallel RLC circuit presented in Example \ref{Example::ch2::prlc::BM}, $P$ and $A$ are both indefinite. 
To deduce the new passivity property (with respect to $V_s$ and $\frac{di}{dt}$), we need to find admissible pairs $\tilde{A}$ and $\tilde{P}$ such that 
\begin{align}
\tilde A \begin{bmatrix}
\frac{di}{dt} \\ \frac{dv}{dt}
\end{bmatrix} =
\begin{bmatrix}
\frac{\partial  \tilde P}{\partial i} \\
\frac{\partial \tilde P}{\partial v}
\end{bmatrix} +
\begin{bmatrix}
-1\\0
\end{bmatrix}V_s.
\label{BM finite dimensional example}
\end{align}
As shown in \cite{ElosiaDimOrt}, the following choice of $\lambda =1$ and $M=\textrm{diag}\{0,\dfrac{2C}{G}\}$ results in
\beqn
\tilde{A}= \begin{bmatrix}
	-L & \dfrac{2C}{G}\\0 & -C
\end{bmatrix} \text{\;\;and\;\;}
\tilde{P}= \dfrac{1}{2G}\left(Gv-i\right)^2+\dfrac{1}{2}\left(R+\dfrac{1}{G}\right)i^2.
\eeqn
This yields the desired dissipation inequality $\dot{\tilde{P}}\leq \frac{di}{dt}V_s$. 
Further, we can achieve the required stabilization via the control voltage \citep{bookgeo,ElosiaDimOrt}
\beq \label{fin_V_s}
V_s=-K(i-i^\ast)+(R+\dfrac{1}{G})i_L^\ast
\eeq
with $K\geq 0$ as a tuning parameter. This controller globally stabilizes the system with Lyapunov function
\beq\label{fin_P_d}
\tilde{P}_d=\dfrac{1}{2G}\left(Gv-i\right)^2+\dfrac{1}{2}(R+\dfrac{1}{G}+K)(i-i^\ast)^2.
\eeq
\end{example}
\begin{remark}
Note that the symmetric part of $\tilde{A}$ is negative definite if and only if $G^2 L \geq C$. Hence, any passivity/stability properties derived using this pair holds only under these constraints. In Chapter 3, we present an alternate methodology that avoids finding admissible pairs, thus eliminating these parameter constraints.
 \end{remark}
\section{Infinite-dimensional port-Hamiltonian systems}\label{Sec:: From inf dim pH-BM}
In this section, we present the Hamiltonian formulation of a distributed parameter system that includes the boundary energy flows. The basic concept needed in the formulation of a port-Hamiltonian system is that of a Dirac structure, which is a geometric object formalizing general power conserving interconnections. To incorporate the power exchanges through boundary, the authors in \cite{stokes} make use of Stokes' theorem along with the properties of exterior derivatives in defining the Dirac structures. Hence the name Stokes-Dirac structure. We start by presenting notation and some key properties in exterior algebra that also help us present our results in Chapter 5.

%
\textbf{\em Notation}: Let $Z$ be an $n$ dimensional Riemannian manifold with a smooth $(n-1)$ dimensional boundary $\partial Z$.  $\Omega^k(Z)$, $k=0,1,\djc{\ldots} ,n$ denotes the space of all exterior $k$-forms on $Z$. The dual space $\left(\Omega^k(Z)\right)^{\ast}$ of $\Omega^k(Z)$ can be identified with $\Omega^{n-k}(Z)$ with a pairing between $\alpha \in \Omega^k(Z)$ and $\beta \in \left(\Omega^k(Z)\right)^{\ast} $  given by $\left<\beta | \alpha \right>=\int_Z \beta \wedge \alpha$. Here, $\wedge$ is the usual wedge product of differential forms, resulting in the $n$\djc{-}form $\beta\wedge \alpha$. Similar pairing\djc{s} can be established between the boundary variables. Further, we denote $\alpha|_{\partial Z}$ to be the $k$-form $\alpha$ evaluated at boundary $\partial Z$. Let $\alpha=(\alpha_1,\alpha_2)\in\mathcal{F}:=\Omega^k(Z)\times \Omega^{l}(\partial Z)$ and $\beta=(\beta_1,\beta_2)\in\mathcal{F}^\ast=\Omega^{n-k}(Z)\times \Omega^{n-1-l}(\partial Z)$. Then, we define the following pairing between $\mathcal{F}$ and $\mathcal{F}^\ast$
\begin{eqnarray}\label{wedge_operation}
\int_{(Z+ \partial Z)}\alpha \wedge \beta:=\int_{Z}\alpha_1 \wedge \beta_1+\int_{\partial Z}\alpha_2 \wedge \beta_2
\end{eqnarray}
\djc{The operator} `$\mathrm d$' denotes the  exterior derivative and maps $k$ forms on $Z$ to $k+1$ forms on $Z$. The Hodge star operator $\ast$ (corresponding to Riemannian metric on $Z$) converts $p$ forms to $(n-p)$ forms.
Given  $\alpha, \beta \; \in \Omega^k(Z) $ and $\gamma \in \Omega^l(Z)$, the wedge product $\alpha \wedge \gamma \in \Omega^{k+l}(Z)$. We additionally have the following properties:
\beq
\alpha \wedge \gamma &=& (-1)^{kl}\gamma \wedge \alpha ~, ~
\ast \ast \alpha = (-1)^{k(n-k)}\alpha \label{Ac},\\
\int_z \alpha \wedge \ast \beta &=& \int_z \beta \wedge  \ast \alpha \label{Ab},\\
\mathrm{d}\left(\alpha \wedge \gamma\right)&=& \mathrm{d}\alpha \wedge \gamma+(-1)^{k} \alpha \wedge \mathrm{d}\gamma. \label{Ad}
\eeq
\djc{For details on the theory of differential forms we refer to \citep{AbrMarRat88}}.
Given a functional $H(\alpha_p,\alpha_q)$, we compute its variation as
\beq \partial H &=& H(\alpha_p+\partial \alpha_p,\alpha_q+\partial\alpha_q)-H(\alpha_p,\alpha_q)\\
&=&   \int_z\left( \delta_{p}H \wedge \partial \alpha_p  + \delta_{q}H \wedge \partial \alpha_q \right)+ \int_{\partial z}\left( \delta_{ \alpha_p|_{\partial z}}H \wedge \partial \alpha_p  + \delta_{\alpha_q|_{\partial z}}H \wedge \partial \alpha_q \right),\nonumber
\label{delta}
\eeq
where $\alpha_p,\;\partial \alpha_p \in \Omega^p(Z)$ and $\alpha_q, \; \partial \alpha_q\in \Omega^q(Z)$; and \\$\delta_{p}H \in \Omega^{n-p}(Z)$, $\delta_{q}H \in \Omega^{n-q}(Z)$ are variational derivative of $H(\alpha_p,\alpha_q)$ with respect to $\alpha_p$ and $\alpha_q$; and $\delta_{\alpha_p|_{\partial z}}H \in \Omega^{n-p-1}(\partial Z)$, $\delta_{\alpha_q|_{\partial z}}H \in \Omega^{n-q-1}(\partial Z)$ constitute variations at boundary. Further, the time derivatives of $H(\alpha_p,\alpha_q)$ is
\beqn \dfrac{dH}{dt}     &=&    \int_Z\left ( \delta_{p}H \wedge  \dfrac{\partial \alpha_p}{\partial t} +\delta_{q}H \wedge  \dfrac{\partial \alpha_q}{\partial t} \right )+ \int_{\partial Z}\left( \delta_{ \alpha_p|_{\partial z}}H \wedge  \dfrac{\partial \alpha_p}{\partial t}  + \delta_{\alpha_q|_{\partial z}}H \wedge \dfrac{\partial \alpha_q}{\partial t} \right).
\eeqn
Let $G: \Omega^{n-p}(Z)\rightarrow \Omega^{n-p}(Z)$ and $R: \Omega^{n-q}(Z)\rightarrow \Omega^{n-q}(Z)$. We call $G\geq 0$, if and only if $\forall \alpha_p \in \Omega^p(Z)$
\beq\label{eqn::G}
\int_Z \left ( \alpha_p \wedge \ast G \alpha_p \right )=\int_Z \left<\alpha_p,G\alpha_p\right>\text{Vol} \geq 0
\eeq
where the inner product is induced by the Riemmanian metric on $Z$ and $\text{Vol}\in \Omega^n(Z)$ such that $\int_Z \left ( \text{Vol} \wedge \ast \text{Vol} \right )=1$. 
$G$ is said to be symmetric if $\left<\alpha_p |  G \alpha_p\right>=\left<G\alpha_p |\alpha_p\right>$. Given $u(z,t):Z \times \mathbb{R}\rightarrow \mathbb{R}$, we denote $\frac{\partial u}{\partial t}(z,t)$ as $u_t$, similarly $\frac{\partial u}{\partial z}(z,t)$ as $u_z$ and $u^\ast(z)$ represents the value of $u(z,t)$ at equilibrium. Furthermore, for $P(z,u,u_z):Z\times \mathbb{R}\times \mathbb{R}^n\rightarrow \mathbb{R}$, we denote $\frac{\partial P}{\partial u_z}$ as $P_{u_z}$.

\textbf{\em Stokes-Dirac structure}: 
Define the linear space $\mathcal{F}_{p,q}=\Omega^p(Z)\times \Omega^q(Z)\times \Omega^{n-p}(\partial Z)$ called the space of flows and $\mathcal{E}_{p,q}=\Omega^{n-p}(Z)\times \Omega^{n-q}(Z)\times \Omega^{n-q}(\partial Z)$, the space of efforts, with integers $p,q$ satisfying $p+q=n+1$. \ijc{Let $(f_p,f_q,f_b)\in \mathcal{F}_{p,q}$ and $(e_p,e_q,e_b)\in \mathcal{E}_{p,q}$}.
Then, the linear subspace $\mathcal{D}\subset \mathcal{F}_{p,q}\times \mathcal{E}_{p,q}$
\beq\label{ph_2_stokes}
\mathcal{D}=\left\{
\left(f_p,f_q,f_b,e_p,e_q,e_b\right)\in  \mathcal{F}_{p,q}\times \mathcal{E}_{p,q}\, \djc{\bigg|}\right.&&\left. \,\begin{bmatrix}
	f_p\\ f_q
\end{bmatrix}=    \begin{bmatrix}
0 & (-1)^r\mathrm{d}\\\mathrm{d} & 0
\end{bmatrix}\djc{\begin{bmatrix}
	e_p\\ e_q
\end{bmatrix}},\right.\\&&\left.
\begin{bmatrix}
	f_b\\ e_b
\end{bmatrix}=\begin{bmatrix}
1 & 0\\0 & -(-1)^{n-q}
\end{bmatrix}\begin{bmatrix}
e_p|_{\partial Z}\\ e_q|_{\partial Z}
\end{bmatrix}
\right\},  \nonumber
\eeq
\djc{with} $r=pq+1$, \djc{is a Stokes-Dirac} structure, \citep{stokes}  with respect to the \djc{bilinear} form 
\beqn 
\left< \left< \left(f_p^1,f_q^1,f_b^1,e_p^1,e_q^1,\right.\right.\right.&&\left.\left.\left.\hspace{-9mm}e_b^1\right),\left(f_p^2,f_q^2,f_b^2,e_p^2,e_q^2,e_b^2\right) \right> \right>  =  \\&& \left< e_p^2|f_p^1\right>+\left< e_p^1|f_p^2\right>+\left< e_q^2|f_q^1\right>+\left< e_q^1|f_q^2\right>+\left< e_b^2|f_b^1\right>+\left< e_b^1|f_b^2\right>
\eeqn
\ijc{where $
\begin{matrix}
	(f_p^i,f_q^i,f_b^i)  \in \mathcal F_{p,q}& \text{and}  & (e_p^i,e_q^i,e_b^i)\in  \mathcal E_{p,q} &\text{for} & i=1,2 
\end{matrix}$}.

\textbf{\em Infinite-Dimensional Port-Hamiltonian Systems}: 
Consider a distributed\djc{-}parameter port\djc{-}Hamiltonian system on $ \Omega^p(Z)\times \Omega^q(Z)\times \Omega^{n-p}(\partial Z)$, with energy variables $\left(\alpha_p ,\alpha_q\right) \in \Omega^p(Z)\times \Omega^q(Z)$ representing two different physical energy domains interacting with each other. The \djc{total stored energy is defined as} 
$$
H:=\int_Z \text{H} \in \mathbb{R}, 
$$
where $\text H$ is the Hamiltonian density \djc{(energy per volume element)}. Let $G\geq 0$ and $R \geq 0$ (satisfying \eqref{eqn::G}) represent the dissipative terms in the system.  \djc{Then, setting $f_p = -(\alpha_{p})_t $, $f_q = -(\alpha_q)_t$, and  $e_p = \delta_p H$, $e_q = \delta_q H$, the system}
\beq
-\frac{\partial}{\partial t}\left[\begin{matrix}
	\alpha_p\\ \alpha_q
\end{matrix}\right]=\begin{bmatrix}
\ast G & (-1)^r\mathrm{d}\\\mathrm{d} & \ast R
\end{bmatrix}\begin{bmatrix}
\delta_pH\\ \delta_qH
\end{bmatrix} \label{pH_2},
\left[\begin{matrix}
	f_b\\ \ e_b
\end{matrix}\right]=\begin{bmatrix}
1 & 0\\0 & -(-1)^{n-q}
\end{bmatrix}\begin{bmatrix}
\delta_pH\djc{|_{\partial Z}}\\ \delta_qH\djc{|_{\partial Z}}
\end{bmatrix}, 
\eeq
\djc{with $r=pq+1$}, represents an infinite-dimensional port-Hamiltonian system with dissipation.
%
The time\djc{-}derivative of the Hamiltonian is computed as
\beq \label{inf_Hdot_sec2.5}
\frac{dH}{dt} \le \int_{\djc{\partial Z}} e_b \wedge f_b.
\eeq
 \begin{remark}
 Equation \eqref{inf_Hdot_sec2.5} means that the increase in energy in the spatial domain is less than or equal to power supplied to the system through its boundary.  This implies that the system is passive with respect to the boundary variables $e_b$, $f_b$ and \djc{storage function $H$} (where $H$ is assumed to be bounded from below).
 \end{remark}
\newpage
\section{Stability of Infinite dimensional systems}\label{sec::stability}
In the case of infinite dimensional systems, it is not sufficient enough to show the positive definiteness of the Lyapuov function and the negative definiteness of its time derivative (as in the case of finite dimensional systems), to prove Lyapunov stability. In infinite dimensional systems, one must specify the norm associated with stability argument because stability with respect to a norm does not imply that it is stable with respect to another norm. \ijc{Let $\mathcal{U}_{\infty}$ be the configuration space of a distributed parameter system, and $\norm{\cdot}$ be a norm on $\mathcal{U}_{\infty}$}.
\begin{definition}\label{def::stability}
	Denote by  $\;U^{\ast}\in \mathcal{U}_{\infty}$ an equilibrium configuration for a distributed parameter system on $\mathcal{U}_{\infty}$. Then, $U^\ast$ is said to be stable in the sense of Lyapunov with respect to the norm $\norm{\cdot}$ if, for every $\epsilon\geq 0$ there exist $\delta\geq 0$ such that, 
	\beqn
	\begin{matrix}
		\norm{U(0)-U^\ast}\leq \delta & \implies & \norm{U(t)-U^\ast}\leq\epsilon
	\end{matrix}
	\eeqn
	for all $t\geq 0$, where $U(0)\in \mathcal{U}_\infty$ is the initial configuration of the system. \kcn{We state the following stability theorem for infinite-dimensional systems, which is also referred to as Arnold's theorem for stability of infinite-dimensional systems.}
\end{definition}
\begin{theorem}\label{thm::stab}\textit{(Stability of an infinite-dimensional system \citep{swaters}):} Consider a dynamical system $\dot{U}=f(U)$ on a linear space $\mathcal{U}_\infty$, where $U^{\ast}\in \mathcal{U}_{\infty}$ is an equilibrium. Assume there exists a solution to the system and suppose there exists function $P_d:\mathcal{U}_\infty\rightarrow \mathbb{R}$ such that
	\beq
	\begin{matrix}
		\delta_{U}P_d(U^\ast)=0 & \text{and} & \dfrac{\partial P_d}{\partial t} \leq 0.
	\end{matrix}
	\eeq
	Denote $\Delta U=U-U^\ast$ and $\mathcal{N}(\Delta U)= P_d(U^\ast +\Delta U)-P_d(U^\ast)$. Suppose that there exists a positive triplet $\alpha$, $\gamma_1$ and $\gamma_2$ satisfying 
	\beq
	\gamma_1\norm{\Delta U}^2 \leq  \mathcal{N}(\Delta U) \leq \gamma_2\norm{\Delta U}^\alpha.
	\eeq 
	Then $U^\ast $ is a stable equilibrium.
\end{theorem}
\newpage
\section{Notes on Chapter 2}

\vspace{0.5cm}

\noindent\begin{minipage}{0.05\linewidth}
\vspace{-.8cm}
(i)
\end{minipage}\begin{minipage}{0.95\linewidth}
Port-Hamiltonian formulation presented in \eqref{iPHDS} is called as `Input-State-Output port-Hamiltonian system'. For a general overview see \cite[Chapter 6]{l2gain}.
\end{minipage}
\vspace{0.5cm}

\noindent\begin{minipage}{0.05\linewidth}
    \vspace{-1.5cm}
    (ii)
\end{minipage}
\begin{minipage}{0.95\linewidth}
Brayton-Moser formulation is one alternative framework that gives port-variables, that are not power-conjugates. Regarding more information on these alternate passive maps, see \cite{venkatraman2009energy,venkatraman2010energy}.
\end{minipage}
\vspace{.5cm}
    
\noindent\begin{minipage}{0.05\linewidth}
\vspace{-0.8cm}
(iii)
\end{minipage}
\begin{minipage}{0.95\linewidth}
The analysis presented in the first part of Section 2.4, Brayton-Moser formulation, is extracted from \cite[Chaper 11]{van2014port}. 
\end{minipage}
\vspace{.5cm}

\noindent\begin{minipage}{0.05\linewidth}
\vspace{-1.6cm}
(iv)
\end{minipage}
\begin{minipage}{0.95\linewidth}
Brayton-Moser formulation of topologically complete RLC circuits and their Dirac formulation can be found in \cite{Guido}. For more on control by power shaping methodology see \cite{ElosiaDimOrt,ortega2003power}.
\end{minipage}
\vspace{.5cm}

\noindent\begin{minipage}{0.05\linewidth}
	\vspace{-0.8cm}
	(v)
\end{minipage}
\begin{minipage}{0.95\linewidth}
For more on distributed-parameter port-Hamiltonian systems defined using Stokes-Dirac structure, see \cite{stokes}.
\end{minipage}
\vspace{.5cm}

\noindent\begin{minipage}{0.05\linewidth}
	\vspace{-1.5cm}
	(vi)
\end{minipage}
\begin{minipage}{0.95\linewidth}
	Infinite dimensional stability theorem presented in Theorem \ref{thm::stab} is taken from \cite{ZhenBaoOmer}. For LaSalle's invariance principle for infinite-dimensional systems see  \citep[Theorem 5.19]{RamThe}.
\end{minipage}

\afterpage{\blankpage}

%
%
	\addtocounter{page}{1}%
 \chapter{Control of finite dimensional system}\label{ch:con_fin_dim_sys}
Energy shaping methods for designing controllers often suffer from dissipation obstacle. The Brayton-Moser formulation is a possible alternative to get around this problem. However, even in this framework, designing controllers leads to two chief difficulties. The first one involves solving partial differential equations, which might be a herculean task. Particularly, with regard to energy shaping methods, the authors of \cite{donaire2016shaping} demonstrated that the need for solving these can be eliminated by finding new passive maps whose output port-variable is integrable. This idea motivates us to look for power based passive maps whose output port variable is integrable; we show that the output port-variable derived from Brayton-Moser formulation is integrable under the assumption that the co-vectors corresponding to the columns of the input matrix are all closed. It is worth noting that this methodology eventually leads to a PI-type controller. The second difficulty lies in the fact that the \thesis{power shaping} methodology requires one to find storage functions satisfying the gradient structure. This difficulty is again alleviated by finding novel passive maps using Karsovskii type storage functions. Precisely, we show that for a class of systems modeled in Brayton-Moser form, this idea leads to a passivity property with ``differentiation at both the ports''. In addition, we also  generalize these results to a larger class of non-linear systems. The details of the aforementioned issues, and their resolutions, will form the body of this chapter. 
%
%
%
%
%
%
\section{Power Shaping}\label{powershaping}
 Power shaping stabilization is a method where the storage function is derived from power of the system instead of the total stored energy.
The first step in this framework is to prove that the plant is passive, which requires finding {\em admissible pairs}. In the context of electrical networks, the passivity property is now established with respect to voltage and derivative of current, or current and derivative of voltage (see Example \ref{Example::ch2::prlc::BM}). The next key step in control by power shaping is to `assign a desired power-like function' to the closed-loop system through control, such that the closed-loop system is passive. Similar to the case of energy shaping, this often requires solving of partial differential equations.
%
%
However, this method has natural advantages over practical drawbacks of energy shaping methods like speeding up the transient response (as derivatives of currents and voltages are used as outputs) and also help overcome the ``dissipation obstacle''.
\subsection{Brayton-Moser formulation}
Physical systems in Brayton-Moser framework are modeled as pseudo-gradient systems using a function called mixed potential function which has units of power. 
In the case of RLC networks, the mixed potential function is the sum of the content of the current carrying resistors, co-content of the voltage controlled resistors and instantaneous power transfer between storage elements \cite{jeltsema2009multidomain}.
Consider the standard representation of a system in Brayton-Moser formulation
\beq
Q(x)\dot{x}=\nabla_x P(x)+G(x)u\label{BM1}
\eeq
where $x \in \mathbb{R}^{n}$ denotes the system state vector and $u \in \mathbb{R}^{m}$ denotes the input vector ($m \leq n$). $P: \mathbb{R}^{n} \rightarrow \mathbb{R}$ is a scalar function of the state, which has the units of power also referred to as mixed potential function, 
$Q(x): \mathbb{R}^{n} \rightarrow \mathbb{R}^{n} \times \mathbb{R}^{n}$ and $G(x): \mathbb{R}^{n} \rightarrow \mathbb{R}^{n} \times \mathbb{R}^{m}$. The time derivative of the mixed potential functional is
 \beqn
\dfrac{d}{dt}P(x) &=& \nabla_xP(x)\cdot \dot x\\
&=& (Q(x)\dot x-G(x)u)\cdot \dot x\\
&=& \dot{x}^\top Q(x)\dot{x}-u^\top G(x)^\top \dot{x}
\eeqn
This suggests that if $P(x)\geq 0$ and $Q(x)\leq 0$, the system \eqref{BM1} is passive with storage function $P(x)$ and  power variables are  $u$, $y=-G(x)^\top \dot{x}$.
 But, in general $P(x)$ and $Q(x)$ can be indefinite \cite{ortega2003power}. In that case one needs to find a new $(\tilde{P}, \tilde{Q})$ called ``admissible pair'', satisfying the pseudo-gradient structure \eqref{BM1}.
 
We state this formally in the following assumption. Towards the end of the chapter, we aim to relax this by finding new passive maps.
 \begin{assumption}\label{assum:admisible_pairs}
 	For the given system in Brayton-Moser form \eqref{BM1}, there exists $\tilde{P}(x)\geq0$ and $\tilde{Q}(x)\leq 0$ such that
 	\begin{equation}\label{BM_adm}
 	\tilde{Q}(x)\dot{x}=\nabla_x \tilde{P}(x)+\tilde{G}(x)u.
 	\end{equation}
 	 Such pairs of $\tilde{P}$ and $\tilde{Q}$ are called {\em admissible pairs} for \eqref{BM1}.
 \end{assumption}
%
\noindent This assumption leads to the following passivity property, which also helps us avoid the dissipation obstacle problem.
%
 %
 %
 %
 \begin{proposition} \label{prop:: power_BM_passivity}
 	Consider the system in Brayton-Moser form \eqref{BM1} satisfying Assumption \ref{assum:admisible_pairs}.  Then the system is passive with input $u$, output given by $y_{PB}=-\tilde{G}(x)^{T}\dot{x}$ and storage function $\tilde{P}$.
 \end{proposition}
\begin{proof}
	Time differential of $\tilde{P}$ is given by
	\begin{eqnarray}\label{pass}
	\dot{\tilde{P}}=(\nabla_x \tilde{P})^{T} \dot{x}
	&=&\dot{x}^{T}\tilde{Q}\dot{x}+u^{T}y_{PB}\nonumber\\
	&\leq& u^{T}y_{PB},
	\end{eqnarray}
	where $y_{PB}$ is given by
	\begin{equation}\label{PBO}
	y_{PB}=-\tilde{G}(x)^{T}\dot{x}
	\end{equation} which is referred as power balancing (shaping) output \cite{ortega2003power}.
\end{proof}
\noindent In the next subsection, we present \thesis{a} different approach for power shaping by utilizing the ``differentiation'' on output port-variable. 
 \subsection{Control methodology using integral outputs}
The word `shaping' in `energy shaping' and `power shaping' methods, which fall under passivity based control (PBC) methodologies, refers to the modification of closed-loop storage function through control. There are several ways to achieve this shaping and one among them is called control by interconnection (CBI).
 %
To begin with, in CBI, it is assumed that the controller is a passive dynamical system interconnected to the physical system. Then, closed-loop invariant functions called Casimirs are determined which relate the system and controller state variables. The closed-loop Hamiltonian is thus restricted to the level-sets given by Casimir functionals. Now, one seeks a Casimir functional such that the minima of the closed-loop Hamiltonian coincides with the desired operating point of the system. Finding such Casimir functionals is yet another hindrance, apart from the dissipation obstacle drawback mentioned earlier, as it involves solving a system of partial differential equations. We would also like to point out that other PBC methodologies, such as `energy/power balancing' and `interconnection and damping assignment', are also plagued by similar difficulties. 
\begin{definition}\label{def::inte}(Integrable) Consider $x\in \mathbb{R}^n$ and  $g(x)\in \mathbb{R}^{n\times m}$. Let $g_k(x)$ be the $k^{th}$ column of $g(x)$ and $g_{kl}(x)$ denotes the $l^{th}$ element of vector $g_k(x)$ where, $k\in\{1\cdots n\}$ and $l\in\{1\cdots m\}$. Denote $g^k(x)\deff \sum_{l=1}^{m}g_{kl}dx^l$, $k\in\{1\cdots n\}$.	
 We call the matrix $g(x)$ integrable if $1-forms$ $g^k(x)$, $\forall k\in\{1\cdots n\}$ are closed. This is equivalent to the following: the matrix $g(x)$ is integrable if $\nabla_xg_k(x)=(\nabla_xg_k(x))^\top $, $\forall k\in\{1\cdots n\}$.
\end{definition}
In this subsection, we make use of a method which was proposed for energy shaping in \cite{donaire2016shaping, mehra2017control,gogte2012passivity,satpute2014geometric, borja2015shaping}. The authors of these papers overcome  technical difficulties, essentially similar to the ones just mentioned in the last paragraph, resulting from Interconnection and Damping Assignment-Passivity based Control (IDA-PBC) methodology. This is accomplished in two steps; as a first step, the authors find new passive maps whose output port variable is integrable and secondly, they relax the assumption that the closed-loop system adheres to port-Hamiltonian structure. Adopting this idea to the case of power shaping, we first show that the output port-variable $y_{PB}$, derived in Proposition \ref{prop:: power_BM_passivity}, is integrable under the assumption that the input matrix $\tilde{G}$ is integrable. 
Secondly, we do not constrain the closed-loop storage function to the gradient equation \eqref{BM_adm}. We begin by restating the following assumption:
\begin{assumption}\label{assum:int_G}
	The new input matrix $\tilde{G}(x)$ is Integrable.
\end{assumption}

\begin{lemma}\label{prop:int_output}
	Consider the system in Brayton-Moser form \eqref{BM1} satisfying Assumption \ref{assum:int_G}. The power balancing output $y_{PB}$ given in equation \eqref{PBO} is integrable.
\end{lemma}
\begin{proof}
	From Assumption \ref{assum:int_G}, we have that $\tilde{G}(x)$ is integrable. From Definition \ref{def::inte}, the $1-forms$'s corresponding to column vectors of $\tilde{G}(x)$ are closed. 
	Therefore, Poincar$\acute{\text{e}}$'s Lemma ensures the existence of a function $\Gamma(x):\mathbb{R}^n\rightarrow\mathbb{R}^m$ such that $\tilde{G}(x)=-\nabla_x\Gamma(x)$. The time derivative of $\Gamma(x)$ is
	\beq
	\dot{\Gamma}&=&\nabla_x\Gamma^\top \dot{x}
	=-\tilde{G}(x)^\top \dot{x}=y_{PB}. \label{gamma_dot}
	\eeq
	Using equation \eqref{PBO} we conclude the proof.
\end{proof}
\begin{ctrlobj}
	The objective is to stabilize the system \eqref{BM_adm} at the equilibrium point $(x^\ast, u^\ast)$ satisfying
	\beq
	\nabla_x \tilde{P}(x^\ast )+\tilde{G}(x^\ast)u^\ast=0. \label{operating point}
	\eeq
\end{ctrlobj}
\noindent The usual methodology to achieve this objective involves finding a new storage function $P_d$ for the closed loop system satisfying
\beq
\tilde{Q}\dot{x}=\nabla_x P_d ~~\text{and}~~
x^\ast= \text{arg\;min}_xP_d \label{clp0}
\eeq
%
where closed loop potential function $P_d$ is \thesis{the} difference of power function $\tilde{P}$  and power supplied by the controller. This was employed in \cite{ElosiaDimOrt}, where the power supplied by controller is found by solving PDE's. \thesis{As stated earlier, we adopt a similar procedure given in \cite{donaire2016shaping, mehra2017control,gogte2012passivity,satpute2014geometric},  in which, the authors have utilized the integrability of the output port-variable in energy shaping of mechanical systems. Also recently in \cite{borja2015shaping} similar idea is used for systems in the port-Hamiltonian form.} 
As mentioned earlier, we do not restrain the closed-loop system to satisfy the gradient structure \eqref{clp0}. Instead, we desire to find a closed loop storage function $P_d$ satisfying
\beq
\dot{P}_d\leq 0 ~~\text{and}~~
x^\ast= \text{arg\;min}_xP_d. \label{clp}
\eeq
\begin{remark}A remark on equations \eqref{clp0} and \eqref{clp}. In \eqref{clp}, we are looking for a Lyapunov function that helps us prove stability. Where as in \eqref{clp0}, we want a Lyapunov function that satisfies the gradient structure. Note that, having the closed-loop system withholding this gradient structure automatically leads to stability, but this may results in solving for partial differential equations and hence not desirable. 
\end{remark}
%
%
In lemma \ref{prop:int_output} we have proved that the power balancing output is integrable. Using this, the desired closed loop potential function $P_d$ is constructed in the following way
\begin{equation}\label{Pd}
P_{d}=k\tilde{P}+\frac{1}{2}||\Gamma(x)+a||_{k_{I}}^{2}
\end{equation}
where $ k>0$, $a \in \mathbb{R}^{m}$, $k_{I} \in \mathbb{R}^{m \times m}$ with $k_{I}>0$. Further  $a$ is chosen such that \eqref{clp} is satisfied, which implies
\beq\label{def}
\nabla_xP_d(x^\ast)=0 ~~~~~ \nabla^2_xP_d(x^\ast)\geq 0
\eeq
which upon solving gives
\begin{equation}\label{const}
a:=kk_{I}^{-1}\tilde{G}^{\dagger}(x^\ast)\nabla_x \tilde{P}(x^\ast)-\Gamma (x^\ast)
\end{equation}
where $\tilde{G}^\dagger$ represents pseudoinverse of $\tilde{G}$.
\begin{proposition}\label{prop::BM_pass_main_res}
	 Consider the system (\ref{BM1}) satisfying the Assumptions \ref{assum:admisible_pairs} and \ref{assum:int_G}. We define the mapping $u:\mathbb{R}^{n} \rightarrow \mathbb{R}^{m}$
	\beq\label{cont}
	u:=\dfrac{1}{k}\left(v+\alpha \tilde{G}^\top\dot{x}-k_I(\Gamma(x)+a)\right)
	\eeq
	where $\alpha >0$, $\nabla \Gamma(x):=-\tilde{G}(x)$ and $v\in \mathbb{R}^{m}$.
	Then system \eqref{BM1} in closed loop is passive with storage function $P_d$ \eqref{Pd} satisfying \eqref{def}, input $v$ and output $y_{PB}$. Further with $v=0$ the system \eqref{BM1} is stable with Lyapunov function $P_d(x)$ and $x^\ast$ as stable equilibrium point. Furthermore, if $y_{PB}=0 \implies \lim\limits_{t\rightarrow \infty}x(t)\rightarrow x^\ast$, then $x^\ast$ is asymptotically stable.
\end{proposition}
\begin{proof}
	The time derivative of closed loop potential function \eqref{Pd} along the trajectories of \eqref{BM1} is
	\begin{eqnarray}
	\dot{P_{d}}&=&k\dot{\tilde{P}}+y_{PB}^{T}k_{I}(\Gamma(x)+a)\nonumber\\
	&\leq&y_{PB}^{T}[ku+k_{I}(\Gamma(x)+a)]\nonumber\\
	&=&y_{PB}^{T}(v-\alpha y_{PB})\nonumber\\
&= &y_{PB}^{T}v-\alpha y_{PB}^Ty_{PB}\nonumber\\
	&\leq &y_{PB}^{T}v,\nonumber
	\end{eqnarray}
	where we used equations \eqref{pass}, \eqref{PBO}, \eqref{cont} in arriving at the result. This proves that the closed loop is passive with storage function $P_d$ \eqref{Pd}, input $v$ and output $y_{PB}$. Further for $v=0$ we have
	\beqn
	\dot{P}_d\leq -\alpha y_{PB}^Ty_{PB}
	\eeqn
	
\noindent	and at equilibrium
	\beq
	u^\ast = -\dfrac{k_I}{k}\left(\Gamma(x^\ast)+a\right). \label{ustar}
	\eeq
	Finally from \eqref{const} and $\eqref{ustar}$ we can show that $(x^\ast, u^\ast)$ satisfy \eqref{operating point}. This concludes the system \eqref{BM1} is stable with Lyapunov function $P_d$ and $x^\ast$ as equilibrium point \cite{khalil1996noninear}. Furthermore, if $\dot{P}_d=0 \implies y_{PB}=0 \implies \lim\limits_{t\rightarrow \infty}x(t)\rightarrow x^\ast$. Finally, we conclude the proof by invoking LaSalle's invariance principle. 
\end{proof}
\begin{remark}The choice of closed loop potential function is obviously not unique. Instead of \eqref{clp} we can have $P_d$ in the following way:
\beq\label{genPd}
P_d(x)&=& k\tilde{P}(x)+f(\Gamma)
\eeq
where $f(\Gamma):\mathbb{R}^m\rightarrow \mathbb{R}$ has to be chosen such that \eqref{def} is satisfied. One such choice for $f(\Gamma)$ is $\frac{1}{2}||\Gamma(x)+a||_{k_{I}}^{2}$. For general $P_d$ of the form  \eqref{genPd}, the control $u$ in \eqref{cont} will take the form
\beq
u=\dfrac{1}{k}\left(v+\alpha \tilde{G}^\top\dot{x}-\nabla_{\Gamma}f(\Gamma)\right).
\eeq
Further one can choose $f(\Gamma)$ such that the controller gives the desired performance.
\end{remark}
\noindent We now present a physical example and demonstrate the control methodology developed in this subsection .
\begin{example}\textbf{(Building Temperature control.)}\label{buildingzone}
Thermal zone is an important component of HVAC subsystem. Although, there are different zone modeling strategies, for control purpose, lumped parameter models are commonly used \cite{ma2012predictive}. Lumped parameter models have resistance-capacitance (RC) interconnected network which represents interaction between zones and between zone and ambient. The capacitances represent the total thermal capacity of the wall and  zone. The resistances are
used to represent the total resistance that the wall offers to
the flow of heat from one side to other. To illustrate the proposed approach, we consider a simple two-zone case separated by a wall, where the surface is modeled as a 3R2C \cite{deng2010building} network as shown in Fig. \ref{fig:zonemodel}.
\begin{figure}[h!]
	\centering
	\includegraphics[width=0.8\linewidth]{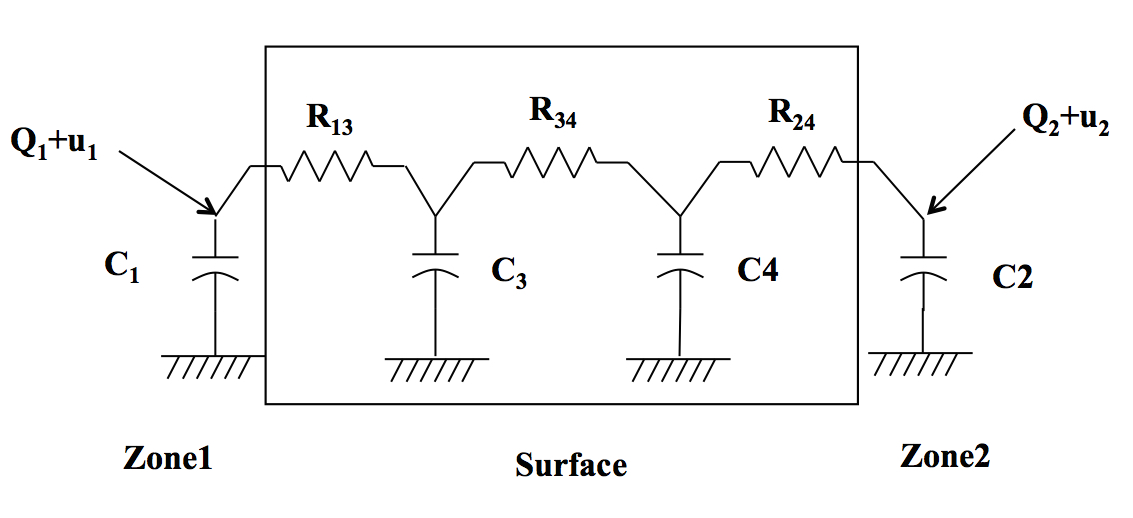}
	\caption{Lumped RC network model: Two zone case}
	\label{fig:zonemodel}
\end{figure}
	The nonlinear thermal model for the two zone case is given by \cite{chinde2016building}
	\begin{eqnarray}\label{dyn}
	C_1\dot{T_1}&=&\dfrac{T_3-T_1}{R_{31}}+\dfrac{(T_{\infty}-T_1)}{R_{10}}+u_1c_p(T_s-T_1)\nonumber\\		C_2\dot{T_2}&=&\dfrac{T_4-T_2}{R_{42}}+\dfrac{(T_{\infty}-T_2)}{R_{10}}+u_2c_p(T_s-T_2)\nonumber\\			
	C_3\dot{T_3}&=&\dfrac{T_1-T_3}{R_{13}}+\dfrac{(T_{4}-T_3)}{R_{34}}\\	C_4\dot{T_4}&=&\dfrac{T_2-T_4}{R_{42}}+\dfrac{(T_{3}-T_4)}{R_{34}}\nonumber
	\end{eqnarray}
In the above model, the inputs $u_1$ and $u_2$ denotes the mass flow rates. $T_{\infty}$, $T_s$ are ambient and supply air temperatures. Note that the inputs are coupled with the states (Temperatures $T_{1}$,$T_{2}$).
\end{example}
\noindent The above system of equations (\ref{dyn}) can be written in the Brayton-Moser form (\ref{BM_adm}) with
$x=\begin{bmatrix}
T_{1},
T_{2},
T_{3},
T_{4}
\end{bmatrix}^{\top}$,
%
and
\beq\label{PQG}
P(x)&=&\frac{(T_{3}-T_{1})^{2}}{2R_{31}}+\frac{(T_{4}-T_{2})^{2}}{2R_{42}}+\frac{(T_{3}-T_{4})^{2}}{2R_{34}}\nonumber\\&&+
\frac{(T_{\infty}-T_{1})^{2}}{2R_{10}}+\frac{(T_{\infty}-T_{2})^{2}}{2R_{20}}.\\
Q(x)&=& \text{diag}[-C_1,-C_2,-C_3,-C_4]~ \text{and}~ \nonumber\\
G(x)&=&\begin{bmatrix}
	-c_p(T_s-T_1)& 0 & 0 & 0\\
	0 & -c_p(T_s-T_2)& 0 & 0 \\
\end{bmatrix}^\top.\nonumber
\eeq
It is easily verified  $P(x)$, $Q(x)$ and $G(x)$ defined in \eqref{PQG} satisfy Assumption \ref{assum:admisible_pairs} and \ref{assum:int_G}. From Proposition \ref{prop:: power_BM_passivity}, system \eqref{dyn} is passive with input $u=[u_1,u_2]^\top$ and power balancing output
 \beqn
 y=\begin{bmatrix}
 	c_p(T_{s1}-T_1)\dot{T}_1 &c_p(T_{s2}-T_2)\dot{T}_2
 \end{bmatrix}^\top,
 \eeqn
 further from lemma \ref{prop:int_output} we have \beqn
 \Gamma_i=-\dfrac{c_p}{2}(T_{si}-T_{i})^2 \;\;\text{for $i$\;=\;1,2}.
 \eeqn
 \begin{ctrlobj}\label{ctrlobj::build_sys}
 	The control objective is to stabilize a given equilibrium point $[T_1^\ast,T_2^\ast]$ satisfying  \eqref{operating point} where
 	\beq\label{uast11}
 	\begin{matrix}
 		u_{1}^{*}&=&-\frac{1}{c_{p}(T_{s1}-T_{1}^{*})}\left(\frac{(T_{3}^{*}-T_{1}^{*})}{R_{31}}+\frac{(T_{\infty}-T_{1}^{*})}{R_{10}}\right)\\
 		u_{2}^{*}&=&-\frac{1}{c_{p}(T_{s2}-T_{2}^{*})}\left(\frac{(T_{4}^{*}-T_{2}^{*})}{R_{42}}+\frac{(T_{\infty}-T_{2}^{*})}{R_{20}}\right).
 	\end{matrix}
 	\eeq
 \end{ctrlobj}
From Proposition \ref{prop::BM_pass_main_res}, we can show that for $a=-kk_{I}^{-1}u^{*}-\Gamma(x^\ast)$ the control input \eqref{cont} takes the form
\beq
\begin{matrix}
	u_{1}&=&-\frac{\alpha}{k}c_{p}(T_{s1}-T_{1})\dot{T_{1}}-\frac{k_{1}}{{k}}\left(\Gamma_{1}-\Gamma_{1}^{*}-\frac{k}{k_1}u_1^{*}\right)\\
	u_{2}&=&-\frac{\alpha}{k}c_{p}(T_{s2}-T_{2})\dot{T_{2}}-\frac{k_{2}}{{k}}\left(\Gamma_{2}-\Gamma_{2}^{*}-\frac{k}{k_2}u_2^{*}\right)
\end{matrix}\label{cont2}
 \eeq
and asymptotically stabilizes the system \eqref{dyn} to equilibrium $[T_1^\ast,T_2^\ast]$ using $P_d$ \eqref{Pd} as lyapunov function.

\begin{remark}
	To achieve the results presented in this section, we principally made two assumptions, (i) Admissible pairs $(\tilde{P},\tilde{A})$ satisfying equation \eqref{BM_adm} exist, (ii) the new input matrix $\tilde{G}$ is integrable. In the following section, we aim to relax the former assumption by finding new passive maps.
\end{remark}
\section{Differential passivity like properties}

 Finding admissible pairs is not always a feasible task. Existence needs sufficient dissipation at all storage elements. (For example: Admissible pairs for series RLC circuits do not exist as there is no dissipation across capacitor \cite{garcia2004new}.) \blue{ Higher the number of storage elements, the more difficult it is to find the admissible pairs. Additionally, the results thus obtained \blue{would} be conservative in light of the restrictions they put on the elements of the system (eg: $L\geq R^2C$ \cite{garcia2004new} and $G^2L\geq C$ in Example \ref{Example::ch2::prlc::BM}).} 
 These restrictions are mainly due to imposing `gradient structure'  when all we need is passivity with differentiation at at least one of the  port variables (see Example \ref{Example::ch2::prlc::BM}). This has led us to search for new storage functions. In any stabilization problem, whether it is to stabilize a system to equilibrium point or to an operating point, the velocities have to go to zero. This has motivated us to look for storage functions in terms of velocities, with its minimum at $zero$. 
A good candidate would be a quadratic function (sum of squared velocities). In this section, we will show that such kind of storage functions ultimately leads to a new passivity property with differentiation at both the ports variables. Using these new passive map, a PI \blue{like} controller is constructed to solve the stabilization problem. To begin with, we start with the parallel RLC circuit presented in Example \ref{moti::example_pRLC}.
\begin{example}\textbf{(Parallel RLC circuit cont'd)}.
	Consider the following storage function for parallel RLC circuit in Example \ref{moti::example_pRLC} \beqn
	S(i_t,v_t)=\dfrac{1}{2}pi_t^2+\dfrac{1}{2}qv_t^2
	\eeqn
	where $p,q\geq0$,  $i_t=\dfrac{di}{dt}$ and $v_t=\dfrac{dv}{dt}$. The time differential of $S(i_t, v_t)$ along the trajectories of \eqref{moti::pRLC_eqn} is
	\beqn
	\dfrac{d}{dt}S(i_t, v_t)&=&pi_ti_{tt}+qv_tv_{tt}\\
	&=& -\frac{p}{L}Ri_t^2-\frac{q}{C}Gv_t^2+\left(\frac{q}{C}-\frac{p}{L}\right)i_tv_t+E_ti_t
	\eeqn
	In first equality we substituted for $\dfrac{d^2i}{dt^2}\deff i_{tt}$ and $\dfrac{d^2v}{dt^2}\deff v_{tt}$ from \eqref{moti::pRLC_eqn}. In second we just rearranged the terms. It can now be easily seen that for the choice $q=C$ and $p=L$ we can write $\dot{S}$ as
	\beq\label{ch3:passivity}
	\dfrac{d}{dt}S(i_t, v_t)&=&-Ri_t^2-Gv_t^2+i_t\dfrac{d}{dt}V_s
	\leq \frac{dV_s}{dt}\frac{di}{dt},
	\eeq
	with the above choice of $p,q$ we can rewrite the storage function $S(i_t, v_t)$ as
	\beq\label{storage}
	S(i_t, v_t)=\dfrac{1}{2}Li_t^2+\dfrac{1}{2}Cv_t^2
	\eeq
	 We now have the following proposition. 
\end{example}
\begin{proposition}\label{prop_pRLC}
The parallel RLC circuit with dynamics \blue{\eqref{moti::pRLC_eqn}} is passive with storage function $S(i_t, v_t)$ \eqref{storage} and \blue{port variables} $\frac{dV_s}{dt}$ and $\frac{di}{dt}$.
\end{proposition}
\begin{proof}
	The proof of the proposition directly follows from \eqref{ch3:passivity} and \eqref{storage}.
\end{proof}
\begin{remark}
		It is noteworthy to mention that, $S(i_t, v_t)$ is defined on the tangent space of \eqref{moti::pRLC_eqn} and has units Power/Time. Where as the mixed potential function in equation \eqref{BM1} has units of Power.
\end{remark}
%
%
%

 %
 %
 Systems with `dissipation obstacle'\cite{ortega2001putting} can be stabilized using the Brayton Moser framework, where passivity is obtained by {\em differentiating one of the port variables}. This has led to power shaping methods for control, but the solutions (if exists) obtained impose constraints on the physical parameters of the system (as shown in Example \ref{Example::ch2::prlc::BM}). The passivity property presented in Proposition \ref{prop_pRLC} have {\em differentiation at both the port variables}, does not impose any constraints on system parameters.
%
%
%
%
%
%
Using this new passive map (in Proposition \ref{prop_pRLC}), we present a control methodology for solving the stabilization problem.
\begin{ctrlobj}
Regulate the voltage across the capacitor of the parallel RLC circuit \eqref{moti::pRLC_eqn} to $v^\ast$.
At this operating point we have:
\beq \label{eqili}
v^\ast +Ri^\ast =V_s^\ast \;\;\;\;\;
i^\ast=Gv^\ast.
\eeq	
\end{ctrlobj}
%
\begin{proposition}\label{prop_control_finite}
	The state feed back controller of the form
	\beq\label{cont_pRLC}
	V_s&=&-\underbrace{K_P(i(t)-i^\ast)}_{\text{Proportional}}-\underbrace{K_I\int_0^t\left(i(\tau)-i^\ast\right)d\tau}_{Integral}
	+Ri^\ast+v^\ast, \;\;\  K_P,\;K_I\geq 0.
	\eeq 
	asymptotically stabilizes the system \eqref{moti::pRLC_eqn} at equilibrium \eqref{eqili}. 
	\beqn\label{cont_storage}
	S_d&=& \dfrac{1}{2}Li_t^2+\dfrac{1}{2}Cv_t^2+\dfrac{K_I}{2}(i-i^\ast)^2,\;\;\; K_I \geq 0.
	\eeqn
	 as Lyapunov function.
\end{proposition}
\begin{proof}
	The time derivative of $S_d$ along the trajectories of \eqref{moti::pRLC_eqn} is
	\beq \label{Timeder_Pd}
	\dfrac{d}{dt}S_d&=&\dfrac{d}{dt}S+K_I(i-i^\ast)i_t\nonumber \\
	&=&-Ri_t^2-Gv_t^2+i_t\dfrac{d}{dt}V_s+K_I(i-i^\ast)i_t\\
	&\leq & \left(\dfrac{d}{dt}V_s+K_I(i-i^\ast)\right)i_t\nonumber \\
	&=& \left(\dfrac{d}{dt}V_s+K_I(i-i^\ast)\right)i_t.\nonumber 
	\eeq
	Choosing $V_s$ of the form \eqref{cont_pRLC} gives in
	 \beq\label{{cont_pRLC}_der} \dfrac{d}{dt}V_s=-K_I(i-i^\ast)-K_Pi_t,\eeq and using this in $\dfrac{d}{dt}S_d$ results in
	\beq
	\frac{d}{dt}S_d &\leq & -K_Pi_t^2.
	\eeq
	Further from \eqref{Timeder_Pd}, we can say that $\exists \alpha>0$ satisfying 
	\beq
	\frac{d}{dt}S_d &\leq & -\alpha\left(i_t^2+v_t^2\right).
	\eeq
	Which implies $\dfrac{d}{dt}S_d= 0$ $\implies$  $i(t)=i_c$ and $v(t)=v_c$ (where $i_c$ and $v_c$ are constant), from \eqref{moti::pRLC_eqn} $V_s(t)=Ri_c+v_c$ is a constant. Substituting this in \eqref{{cont_pRLC}_der}, we get $i_c=i^{\ast}$, $\implies$ $v_c=v^\ast$ and $V_s=V_s^\ast$.
%
\end{proof}
%
%
\subsection{Topologically complete RLC circuits}\label{sec::Topologically complete RLC circuits}
We now present the new passivity property for a larger class of RLC circuits called Topologically complete RLC circuits \cite{jeltsema2003passivity}.  Let the column vectors $i$ and $v$ denote the currents passing though all the inductors and voltage across all the capacitors respectively. 
The dynamics of a complete RLC circuit with regulated
voltage sources in series with inductors is described by
\beq
-L\frac{di}{dt}&=& \frac{\partial P}{\partial i}-B_sV_s\nonumber\\
C\frac{dv}{dt}&=& \frac{\partial P}{\partial v} \label{com_RLC}
\eeq
where the mixed potential function $P(i, v)$ is given by
\beq \label{Mixpot}
P(i,v)&=&i^\top\Gamma v+G( i)-J( v)
\eeq
where $G(i)\geq 0$, $J(v)\geq 0$ represent (possibly) non linear dissipative elements and $\Gamma+\Gamma^\top =0$. Note that $L$ and $C$ are assumed to be constant. Consider the following storage function
\beq\label{com_storage}
S(i_t, v_t)=\dfrac{1}{2}\dfrac{di}{dt}^\top L\dfrac{di}{dt}+\dfrac{1}{2}\dfrac{dv}{dt}^\top C\dfrac{dv}{dt}
\eeq
\begin{proposition}\label{prop_main}
	Let $\nabla_{i}^2G, \nabla_{v}^2J$ be positive semidefinite, then we have the following.
	 The system of equations \eqref{com_RLC} representing the dynamics of a complete RLC circuit, is passive with respective to the storage function $S(i_t, v_t)$ defined in \eqref{com_storage} and ports $B_s^\top\dfrac{di}{dt}$ and $\dfrac{dV_s}{dt}$.
\end{proposition}
\begin{proof}
	The time derivative of $S(i_t, v_t)$ can be simplified as
	\beqn
\dfrac{d}{dt}	S &=& i_t^\top\left(-\nabla_{i}^2Pi_t-\nabla_{vi}^2Pv_t+B_s\dfrac{dV_s}{dt}\right)+v_t^\top\left(\nabla_{iv}^2Pi_t+\nabla_{v}^2Pv_t\right)\\
	&=& -i_t^\top \nabla_{i}^2Pi_t+v_t^\top \nabla_{v}^2Pv_t+i_t^\top B_s\dfrac{dV_s}{dt}
	\eeqn
	From \eqref{Mixpot} we get
	\beq\label{com_Sdot}
\dfrac{d}{dt}S(i_t, v_t)&=&-i_t^\top \nabla_{i}^2Gi_t-v_t^\top \nabla_{v}^2Jv_t+i_t^\top B_s\dfrac{dV_s}{dt}
	\eeq
	From \eqref{com_RLC}, \eqref{com_storage}, \eqref{Mixpot} and \eqref{com_Sdot}, we have
	\beq
	\dfrac{d}{dt}S(i_t, v_t)&\leq&i_t^\top B_s\frac{dV_s}{dt}.
	\eeq
\end{proof}
\begin{remark}
In the Proposition \ref{prop_main}:
The nonlinear dynamical system given by \eqref{com_RLC} with input $V_s=0$ is contracting with metric diag$\{L,C\}$ \cite{lohmiller1998contraction,forni2014differential}. In the next section, we will utilize this for generalizing this passivity property to a class of contracting nonlinear systems. 
\end{remark}
\begin{remark}
	 In deriving the result of Proposition \ref{prop_main}, we assumed that the input matrix $B_s$ as constant. The result is not obvious for a system with a state dependent input matrix $B_s$. That is the system represented by equations \eqref{com_RLC} is not passive with port variables $B_s(x)^\top\frac{di}{dt}$ and $\frac{dV_s}{dt}$.
\end{remark}
\section{A class of nonlinear system}
So far we were looking at systems that have been modeled in Brayton-Moser formulation. A different, but related, class of systems are contracting systems (a term coined in the seminal paper \cite{lohmiller1998contraction}). The analysis in these systems pertains to, study the convergence between two trajectories rather than a trajectory to a particular solution. This notion of convergence/stability has given rise to a new passivity concept called as differential passivity \cite{forni2013differential,van2013differential,crouch1987variational}, which is similar to the passive maps derived in Proposition \ref{prop_pRLC} and \ref{prop_main}. Consider a nonlinear system of the form
\begin{equation}\label{gen_sys}
\dot{x}=f(x)+g(x)u
\end{equation}	
where $x\in \mathbb{R}^n$ is the state vector , $u\in \mathbb{R}^m$ ($m<n$) is the control input. $f(x):\mathbb{R}^n\rightarrow \mathbb{R}^n$ and $g(x):\mathbb{R}^n\rightarrow \mathbb{R}^m$ are smooth functions. In this subsection, we aim to derive the passive maps presented in Proposition \ref{prop_pRLC} and \ref{prop_main} to a class of nonlinear systems characterized by the following assumptions.
%
%
%
\begin{assumption}\label{ass::A1}
	For a given $f(x)$, there exist a symmetric positive definite matrix $M\in \mathbb{R}^{n \times n}$ satisfying
	\begin{equation}\label{A1}
	M\dfrac{\partial f}{\partial x}+\dfrac{\partial f}{\partial x}^\top M < 0
	\end{equation}
	This implies the dynamical system $\dot{x}=f(x)$ is contracting.
\end{assumption}
\begin{assumption}\label{ass::A2}
	The full-rank left annihilator of input matrix also left annihilates its Jacobian. If $g^{\perp}$ denotes left annihilator of the input matrix $g(x)$, that is, $ g^{\perp}g=0$ then 
  	\begin{equation}\label{A2}
	g^{\perp}\dfrac{\partial g}{\partial x}=0
	\end{equation}
\end{assumption}
\begin{assumption}\label{ass::A3}
	$Mg(x)$ is Integrable.
\end{assumption}
\noindent The second method of Lyapunov has been widely used for stability analysis of dynamical systems \cite{khalil1996noninear}. This method revolves around finding a suitable Lyapunov function that decreases along the system trajectories. Further, positive definite quadratic functions of state variables are usually a good candidates. The classical Krasovskii's method \cite{krasovskiicertain} of generating Lyapunov functions also bears a similar form in terms of velocities (instead of states) and forms a candidate function for stability analysis. 
In the following proposition, we show that the nonlinear systems with $u=0$ satisfying Assumption \ref{ass::A1} is contracting with a Krasovskii-type Lyapunov function \cite{lohmiller1998contraction,forni2014differential}.
\begin{proposition}\cite{lohmiller1998contraction}\label{prop::Contrating}
Consider system \eqref{gen_sys} with input $u=0$ satisfying Assumption \ref{ass::A1}. Then the resulting dynamical system is contracting.
\end{proposition}
	\begin{proof}
	Consider the Krasovskii Lyapunov function
	\begin{equation}\label{kras_lypunov}
	    V(x,\dot{x})=\dfrac{1}{2}\dot{x}^\top M \dot{x}.
	\end{equation}
	Then the time derivative of \eqref{kras_lypunov} along the trajectories of \eqref{gen_sys} with $u=0$ is
	\begin{align*}
	\small
	    \dfrac{d}{dt}V &= \dot{x}^\top M\ddot{x}= \dot{x}^\top M\left(\dfrac{\partial f}{\partial x}\dot{x}\right)= \dot{x}^{\top}\left(M\dfrac{\partial f}{\partial x}+\dfrac{\partial f}{\partial x}^\top M\right)\dot{x} \leq 0
	\end{align*}
	    This implies the dynamical system $\dot{x}=f(x)$ is contracting in $\mathbb{R}^n$ with respect to the metric $M$. 
	 \end{proof}
\begin{remark}
	    In Assumption \ref{ass::A1}, one can consider a state dependent Riemannian metric $M(x)$, and replace equation \eqref{A1} with
	 \begin{eqnarray*}
	 M\dfrac{\partial f}{\partial x}+\dfrac{\partial f}{\partial x}^\top M +\dot{M}< 0.
	 \end{eqnarray*}
\end{remark}

 Static and dynamic feedback techniques \cite{l2gain, nijmeijer1990nonlinear} are well-known methodologies that are widely used in deriving passivity properties. For a class of nonlinear systems \eqref{gen_sys}, we use this dynamic feedback techniques and storage functions of Krasovskii-type \eqref{kras_lypunov} to achieve passive maps similar to the ones derived in Proposition \eqref{prop_main}.  The following lemma will be instrumental in formulating our result .
 	\begin{lemma}\label{prop::alpha}
 	Consider an input matrix $g(x)$ satisfying Assumption \ref{ass::A2} and an $\alpha\in \mathbb{R}^{m\times m}$. Then
	\begin{equation}\label{stab_cond}
	    	\dot{g}+g \alpha=0
	\end{equation}
	if and only if $\alpha$ satisfies 
		\begin{equation}\label{alpha_1}
	\alpha=-\left(g^\top g\right)^{-1}g^\top \dot{g}.
	\end{equation}
	\end{lemma}
	\begin{proof}
	The {\em only if} part of the proof: consider the following full rank matrix $\begin{bmatrix}g^{\perp}\\g^\top \end{bmatrix}$. Now left multiplying $\left(\dot{g}+g \alpha\right)$ in \eqref{stab_cond} by $\begin{bmatrix}g^{\perp}\\g^\top \end{bmatrix}$ yields
	\begin{eqnarray*}
	\begin{bmatrix}g^{\perp}\\g^\top \end{bmatrix}\left(\dot{g}+g \alpha\right)=\begin{bmatrix}g^{\perp}\left(\dot{g}+g \alpha\right)\\g^\top \left(\dot{g}+g \alpha\right)\end{bmatrix}
	&=&\begin{bmatrix}g^{\perp}\dot{g}\\g^\top \left(\dot{g}-g \left(g^\top g\right)^{-1}g^\top \dot{g}\right)\end{bmatrix}\\
	&=&\begin{bmatrix}g^{\perp}\dfrac{\partial g}{\partial x}\dot{x}\\ \left(g^\top\dot{g}-g^\top g \left(g^\top g\right)^{-1}g^\top \dot{g}\right)	\end{bmatrix}\\
	&=&\begin{bmatrix}0\\ \left(g^\top\dot{g}-g^\top \dot{g}\right)	\end{bmatrix}\\
	&=&0
	\end{eqnarray*}
	By construction $\begin{bmatrix}g^{\perp}\\g^\top \end{bmatrix}$ is full rank matrix, hence $\dot{g}+g \alpha=0$.\\ \noindent The {\em if} part of the proof:
$$
	\dot{g}+g\alpha=0 \implies g^\top g \alpha=g^\top \dot{g}\implies \alpha=-(g^\top g)^\top g^\top \dot{g}.
$$
hence $\alpha=-(g^\top g)^\top g^\top \dot{g}\iff \dot{g}+g \alpha=0$.
	\end{proof}
\begin{figure}
	\centering
	\includegraphics[width=.75\linewidth]{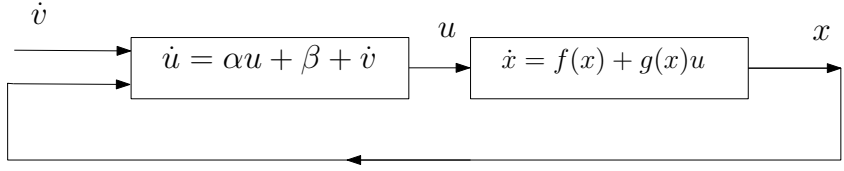}
	\caption{Interconnection of  dynamic state feedback \eqref{input_dyn}  to system \eqref{gen_sys}.}
	\label{fig:dyn_state_feedback}
\end{figure}
%
%
	Consider the following dynamic state feedback \cite{nijmeijer1990nonlinear} for system \eqref{gen_sys} (see Fig. \ref{fig:dyn_state_feedback})
	\begin{equation}\label{input_dyn}
	    \dot{u}= \alpha u+\beta +\dot{v}
	\end{equation}
	with $\alpha$ defined as in lemma \ref{prop::alpha}, $\beta=-g^\top M \dot{x}$ and $v\in \mathbb{R}^m$. The use of $\dot{v}$ in \eqref{input_dyn} rather than $v$ as new port variable will evident in the later part of the subsection. We now have the following passivity property for the nonlinear system \eqref{gen_sys}.
	\begin{theorem}\label{thm::main}
	Let the Assumptions \ref{ass::A1} and \ref{ass::A2} are satisfied. Then the system \eqref{gen_sys} together with \eqref{input_dyn} is passive with input $\dot{v}$ and output $y=g^\top M \dot{x}$.
	\end{theorem}
	\begin{proof}
	Consider  storage function of the form \eqref{kras_lypunov}. The time derivative of \eqref{kras_lypunov} along the trajectories of \eqref{gen_sys} and \eqref{input_dyn} is
	\begin{eqnarray*}
	\dfrac{d}{dt}V&=& \dot{x}^\top M\ddot{x}\\
	    &=& \dot{x}^\top M\left(\dfrac{\partial f}{\partial x}\dot{x}+\dot{g}u+g\dot{u}\right)\\
	    &=& \dot{x}^\top M\left(\dfrac{\partial f}{\partial x}\dot{x}+\dot{g}u+g\left(\alpha u+\beta +\dot v\right)\right)\\
	    &=& \dot{x}^{\top}\left(M\dfrac{\partial f}{\partial x}+\dfrac{\partial f}{\partial x}^\top M\right)\dot{x}+\dot{x}^\top M\left(\left(\dot{g}+g \alpha)\right)u+g\beta+g\dot v \right)\\
	    &\leq & \dot{v}^\top y
	\end{eqnarray*}
	where $y=g^\top M \dot{x}$ is also referred to as power shaping output. In step 1 and 2 we use system dynamics \eqref{gen_sys} and controller dynamics \eqref{input_dyn} respectively. In step 4 and 5 we used Proposition \ref{prop::Contrating} and lemma \ref{prop::alpha} respectively.
	\end{proof}
	\begin{lemma}\label{prop::output_integrability}
	The output $y=g^\top M \dot{x}$ given in Theorem \eqref{thm::main} is integrable.
\end{lemma}
\begin{proof}
		From Assumption \ref{ass::A3}, we have that $Mg(x)$ is integrable. From Definition \ref{def::inte}, the $1-forms$'s corresponding to column vectors of matrix $Mg(x)$ are closed. 
	Therefore, Poincar$\acute{\text{e}}$'s Lemma ensures the existence of a function
%
	$\Gamma(x):\mathbb{R}^n\rightarrow \mathbb{R}^n$ such that
	\beqn
	\nabla_x\Gamma(x)=Mg(x).
	\eeqn
	This implies,
	\begin{eqnarray*}
	\dot{\Gamma}= \nabla_x\Gamma(x)^\top \dot{x}
	=(Mg)^\top \dot{x}
	= g^\top M \dot{x}
	= y
	\end{eqnarray*}
This implies $y$ in Theorem \ref{thm::main} is integrable.
\end{proof}
	\begin{remark}
		From Theorem \ref{thm::main} and lemma \ref{prop::output_integrability}, we can say that nonlinear systems \eqref{gen_sys} satisfying Assumptions \ref{ass::A1}-\ref{ass::A3} are passive with port variables $\dot{u}$ and $\dot \Gamma$. Similar kind passivity properties exists in literature, namely differential passivity \cite{van2013differential} and incremental passivity \cite{l2gain}.
	\end{remark}
	%
	\begin{ctrlobj}
		To stabilize the system \eqref{gen_sys} at an non-trivial operating point $(x^{\ast},u^{\ast})$ satisfying
		\begin{equation}\label{Control_objective}
		f(x^\ast)+g(x^\ast)u^\ast =0
		\end{equation}
	\end{ctrlobj}
To achieve this control objective we follow a similar methodology proposed in Proposition \ref{prop::BM_pass_main_res}. That is, we start with finding a closed-loop storage function $V_d$ satisfying 
\beq
\dot{V}_d\leq 0 ~~\text{and}~~
x^\ast= \text{arg\;min}_xV_d \label{clp_d}
\eeq
	%
and one relevant choice would be
	\begin{equation}\label{clP_str}
	    V_d(x)=\dfrac{1}{2}k_1\dot{x}^\top M \dot{x}+\dfrac{1}{2}||\Gamma(x)-\Gamma(x^\ast)||^2_{k_i}.
	\end{equation}
	\begin{figure}
	\centering
	\includegraphics[width=0.75\linewidth]{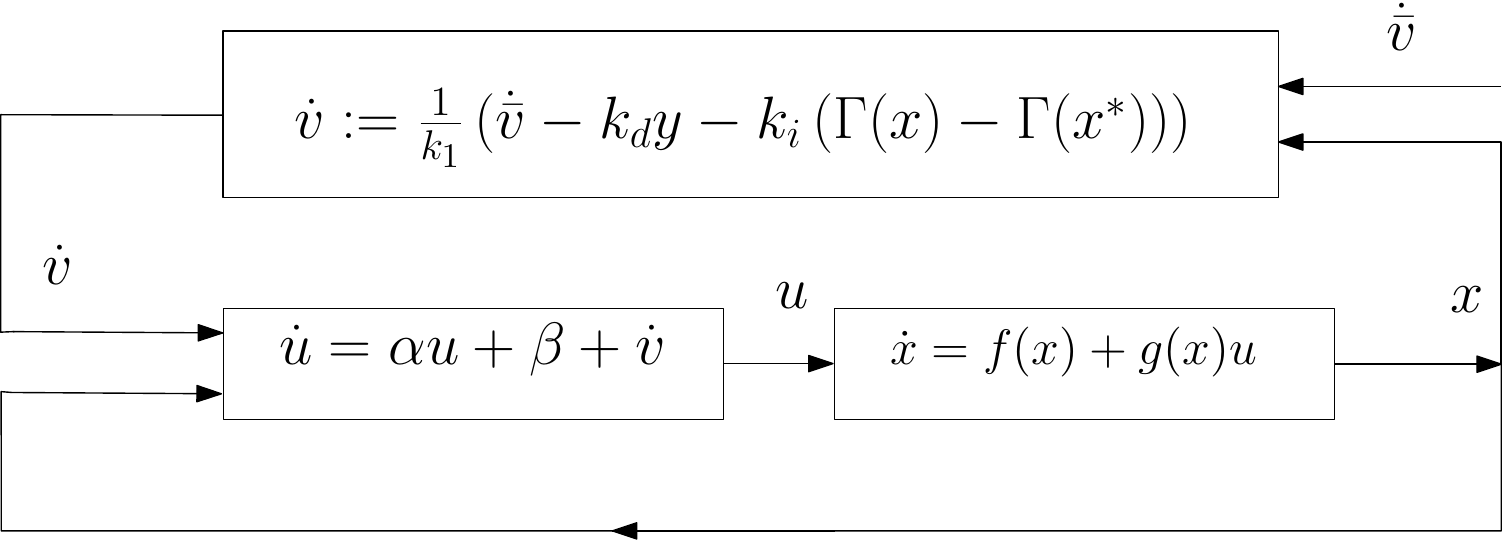}
	\caption{Interconnecting the controller \eqref{input_v_dyn} to dynamic state feedback system in Fig. \ref{fig:dyn_state_feedback}.}
	\label{fig:dyn_state_feedback_controller}
\end{figure}
	\begin{proposition}\label{porp::control}
	Suppose the system \eqref{gen_sys} together with \eqref{input_dyn} satisfies Assumptions \ref{ass::A1}-\ref{ass::A3}. We define the mapping $\dot v:\mathbb{R}^n\rightarrow\mathbb{R}^m$
	\begin{equation}\label{input_v_dyn}
	    \dot v:=\dfrac{1}{k_1}\left(\dot{\bar{v}}-k_d y-k_i \left(\Gamma(x)-\Gamma(x^{\ast})\right)\right)
	\end{equation}
	where $y=g^\top M \dot{x}$ and $\bar{v}\in \mathbb{R}^m$. Then the system of equation \eqref{gen_sys} and \eqref{input_dyn} are passive with port variables $\dot{\bar{v}}$ and $\dot{\Gamma}(x)$ (see Fig. \ref{fig:dyn_state_feedback_controller}). Further for $\dot{\bar{v}}=0$, the system is stable and $x^\ast$  as the stable equilibrium point. Furthermore if $y=0 \implies \lim_{t\rightarrow \infty}x(t)\rightarrow x^{\ast}$, then $x^\ast$ is asymptotically stable.
	\end{proposition}
	\begin{proof}
	The time derivative of the closed loop storage function \eqref{clP_str} is
	\begin{eqnarray*}
	\dfrac{d}{dt}V_d&=& k_1\dot{V}+y^\top k_i(\Gamma(x)-\Gamma(x^\ast))\\
	&\leq & y^\top\left(k_1\dot{v}+ k_i(\Gamma(x)-\Gamma(x^\ast))\right)\\
	&\leq & y^\top \dot{\bar{v}}
	\end{eqnarray*}
		This proves that the closed loop system is passive with storage function $V_d$, input $\dot{\bar{v}}$ and output $y$. Further for $\dot{\bar{v}}=0$ we have 
	\begin{equation*}
	    \dot{V}_d\leq -k_d y^\top y
	\end{equation*}
	and at equilibrium $x=x^\ast$ we have $\dot{v}=0$, further using this in \eqref{input_dyn} we can show that $\dot u=0$. This implies $(x^\ast,u^\ast)$ satisfy the control objective \eqref{Control_objective}, further concluding that system \eqref{gen_sys} is asymptotically stable with Lyapunov function $V_d$ and $x^\ast$ as the equilibrium point.
	\end{proof}
	\begin{remark}
At the desired operating point one can show that  $\dot{u}-\alpha u-\beta=0$.  Hence, we have considered $\dot{u}= \alpha u+\beta +\dot{v} $, instead of $\dot{u}= \alpha u+\beta +v$ in equation \eqref{input_dyn}.
	\end{remark}
\begin{remark}
	Systems that are contracting always forget their initial conditions. That is, their final behaviour is always independent of the initial conditions. Hence, one need not worry about the initial conditions of the control input $u$ while implementing the control law \eqref{input_dyn} together with \eqref{input_v_dyn}.
\end{remark}
We now illustrate this methodology using building HVAC system in Example \ref{buildingzone}.
\begin{proposition}
The systems of equations \eqref{dyn} and \eqref{input_dyn}, with $\alpha$ and $\beta$ defined as 
\begin{eqnarray}\label{alphabeta}
\alpha=\begin{bmatrix}
\frac{\dot{T}_1}{(T_s-T_1)} &0\\0&\frac{\dot{T}_2}{(T_s-T_2)}
\end{bmatrix}\text{and}\;\; \beta=\begin{bmatrix}
c_p(T_1-T_s)\dot{T}_1\\c_p(T_2-T_s)\dot{T}_2
\end{bmatrix}
\end{eqnarray}
respectively, are passive with port variables $\dot{v}$ and $y$. where
\begin{eqnarray}
y(T)&=& c_p\begin{bmatrix}
\left(T_s-T_1\right)\dot{T}_1\\ \left(T_s-T_1\right)\dot{T}_2
\end{bmatrix}.
\end{eqnarray}
\end{proposition}
\begin{proof}
Let $C=\text{diag}\left\{C_1,C_2,C_3,C_4\right\}$.
 One can prove that the system \eqref{dyn} satisfies Assumption \ref{ass::A1} given in equation \eqref{A1} by choosing $M=\text{diag}\left\{C_1,C_2,C_3,C_4\right\}$. \\
 The input matrix of \eqref{dyn} is $g(T)=[g_1(T),g_2(T)]$, where 
\begin{eqnarray*}
g_1(T)=\begin{bmatrix}
\dfrac{c_p}{C_1}(T_s-T_1)&0&0&0
\end{bmatrix}^\top,\;\;
g_2(T)=\begin{bmatrix}
0&\dfrac{c_p}{C_2}(T_s-T_2)&0&0
\end{bmatrix}^\top.
\end{eqnarray*}
Using left annihilator of $g(T)$, that is 
$$g^{\perp}(T)=\begin{bmatrix}0&0&1&0\\ 0&0&0&1\end{bmatrix}$$
one can show that
\begin{eqnarray}\label{inpu_matrix}
\begin{matrix}
g^{\perp}\dfrac{\partial g_1 }{\partial T}=0 &
g^{\perp}\dfrac{\partial g_2 }{\partial T}=0
\end{matrix}
\end{eqnarray}
Hence the input matrix $g(T)$ satisfies Assumption \ref{ass::A2}. Now, we can use lemma \ref{prop::alpha} and show that $\alpha$ takes the same form, given in \eqref{alphabeta}. Finally from Theorem \ref{thm::main}, using
 \begin{eqnarray}\label{Storage_fun_1}
V(T)&=&\dfrac{1}{2}\dot{T}^\top M\dot{T}\\
&=&\dfrac{1}{2}\left(C_1\dot{T}_1^2+C_2\dot{T}_2^2+C_3\dot{T}_3^2+C_4\dot{T}_4^2\right)\nonumber
\end{eqnarray}
as storage function, the system of equations \eqref{dyn}, together with input dynamics \eqref{input_dyn} given by
\begin{eqnarray}\label{udot}
\begin{matrix}
\dot{u}_1=\left(\dfrac{u_1}{(T_s-T_1)}-c_p(T_s-T_1)\right)\dot{T}_1+\dot{v}_1\\
\dot{u}_2=\left(\dfrac{u_2}{(T_s-T_2)}-c_p(T_s-T_1)\right)\dot{T}_1+\dot{v}_2
\end{matrix}
\end{eqnarray}
are passive with port variables $\dot{v}$ and $y$.
\end{proof}
Now we can consider $v=[v_1,v_2]^\top$ as input for the combined equations \eqref{dyn}, \eqref{udot} and provide a control strategy using Proposition \eqref{porp::control}. Consider $a_1=(T_1^\ast-T_s)^2$, $a_2=(T_2^\ast-T_s)^2$, $k_d\geq0$ and $k_i>0$.
\begin{proposition}
The state feedback controller
\begin{eqnarray}\label{Control_port}
\begin{matrix}
\dot{v}_1\hspace{-3mm}&=&\hspace{-3mm}-k_dc_p\left(T_s-T_1\right)\dot{T}_1+\dfrac{1}{2}k_ic_p\left(\left(T_s-T_1\right)^2-a_1\right)\\ 
\dot{v}_2\hspace{-3mm}&=&\hspace{-3mm}-k_dc_p\left(T_s-T_1\right)\dot{T}_2+\dfrac{1}{2}k_ic_p\left(\left(T_s-T_2\right)^2-a_2\right)
\end{matrix}
\end{eqnarray}
asymptotically stabilizes the system of equations \eqref{dyn} and \eqref{udot} to the operating point $(T^\ast,u^\ast)$ satisfying \eqref{Control_objective}.
\end{proposition}
\begin{proof}
With $M=\text{diag}\{C_1,C_2,C_3,C_4\}$ and input matrix $g(T)$ in \eqref{inpu_matrix}, one can verify Assumption \ref{ass::A3}. Hence from lemma  \ref{prop::output_integrability}, we can show that 
\begin{eqnarray}\label{gamma}
\Gamma(T)=-\dfrac{1}{2}c_p\begin{bmatrix}(T_1-T_s)^2\\(T_2-T_s)^2\end{bmatrix} 
\end{eqnarray}
satisfies $\dot{\Gamma}(T)=y(T)$. Further proof directly follows from Proposition \ref{porp::control} using $\Gamma(T)$ in \eqref{gamma}. It can also be proved by taking the time derivative of Lyapunov function \eqref{clP_str} along the trajectories of \eqref{dyn} and \eqref{udot} as shown below
\begin{eqnarray*}
\dot{V}_d&=&k_1\dot{T}^\top M\ddot{T}+k_i(\Gamma(T)-a)^\top \dot{\Gamma}(T)\\
&=& -\dfrac{k_1}{R_{13}}\left(\dot{T}_1-\dot{T}_3\right)^2-\dfrac{k_1}{R_{24}}\left(\dot{T}_2-\dot{T}_4\right)^2\\&&-\dfrac{k_1}{R_{34}}\left(\dot{T}_3-\dot{T}_4\right)^2-\dfrac{k_1}{R_{10}}\left(\dot{T}_1^2+\dot{T}_2^2\right)\\
&&+\dot{T}^\top M\dfrac{d}{dt}\left(g(T)u\right)+k_i(\Gamma(T)-a)^\top y(T)\\
&\leq &\dot{T}^\top\left(\dot{g}u+g\dot{u}\right)+k_i(\Gamma-a)^\top y\\
&= &\dot{T}^\top M\left(\dot{g}u+g(\alpha u+\beta+v)\right)+k_i(\Gamma-a)^\top y\\
&\leq &\dot{T}^\top M\left((\dot{g}+g\alpha)u+gv\right)+k_i(\Gamma-a)^\top y\\
&=& \dot{T}^\top Mg v+k_i(\Gamma-a)^\top y\\
&=& y^\top\left( v+k_i(\Gamma-a)\right)\\
&=& -k_dy^\top y.
\end{eqnarray*}
In step 2 and 4 we use system dynamics \eqref{dyn} and controller dynamics respectively. In step 5 we used $\dot{g}+g\alpha=0$ given in Proposition \ref{stab_cond}. Finally in step 6 we have used the control strategy \eqref{Control_port}. Now one can infer that there exist an $\alpha>0$, such that 
\begin{eqnarray*}
\dot{V}_d&\leq& -\alpha \left(\left(\dot{T}_1-\dot{T}_3\right)^2+\left(\dot{T}_2-\dot{T}_4\right)^2+\left(\dot{T}_3-\dot{T}_4\right)^2\right.\\&&\left.+\dot{T}_1^2+\dot{T}_2^2\right).
\end{eqnarray*}
 $\dot{V}_d=0$ implies $\dot{T}_1$, $\dot{T}_2$, $\dot{T}_3$ and $\dot{T}_4$ are identically zero. Using this in \eqref{dyn}, we get $u_1$ and $u_2$ as constant. From \eqref{udot} we get $v=0$, substituting this in \eqref{Control_port} we get that $T_1=T_1^\ast$, and $T_2=T_2^\ast$. Finally, we conclude the proof by invoking LaSalle's invariance principle.
\end{proof}
{\em Simulation results:}
The parameter values used for the simulation study are given in \cite{deng2010building}. The trajectories of zone temperatures for the two zone case is shown in Fig. \ref{fig:tres} and the effectiveness of controller is shown by zone temperatures reach their respective reference temperature values. The control inputs to the zones and the time evolution of port variables is shown in Fig. \ref{fig:control} and Fig. \ref{fig:ports}. Zone 2 needs higher control effort to reach reference temperature compared to zone 1 due to the higher difference in initial and reference values.
\begin{figure}[h!]
	\centering
\includegraphics[width=0.9\linewidth]{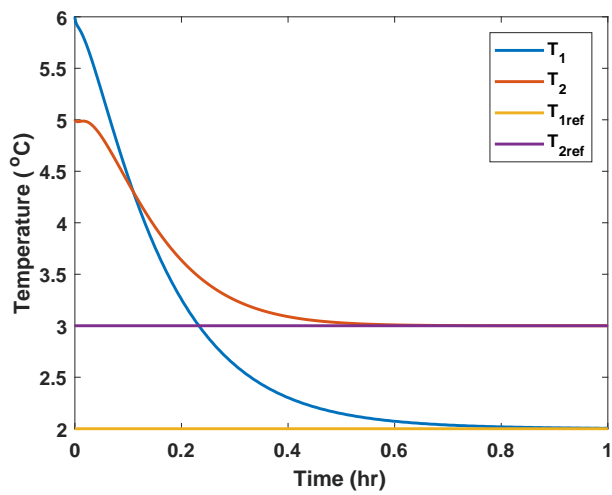}
	\caption{Trajectories of zone temperatures ($T_{\text{1ref}}=2.5$, $T_{\text{2ref}}=6$)}
	\label{fig:tres}
\end{figure}
\begin{figure}[h!]
	\centering
	\includegraphics[width=0.9\linewidth]{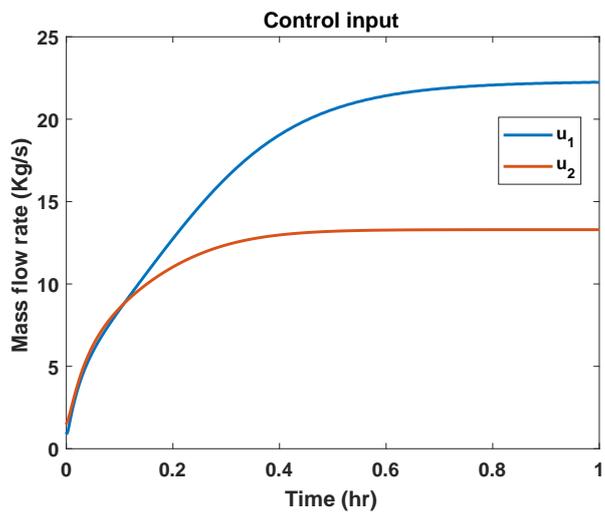}
	\caption{Time evolution of mass flow rate $u$.}
	\label{fig:control}
\end{figure}
\begin{figure}[h!]
	\centering
	\includegraphics[width=0.9\linewidth]{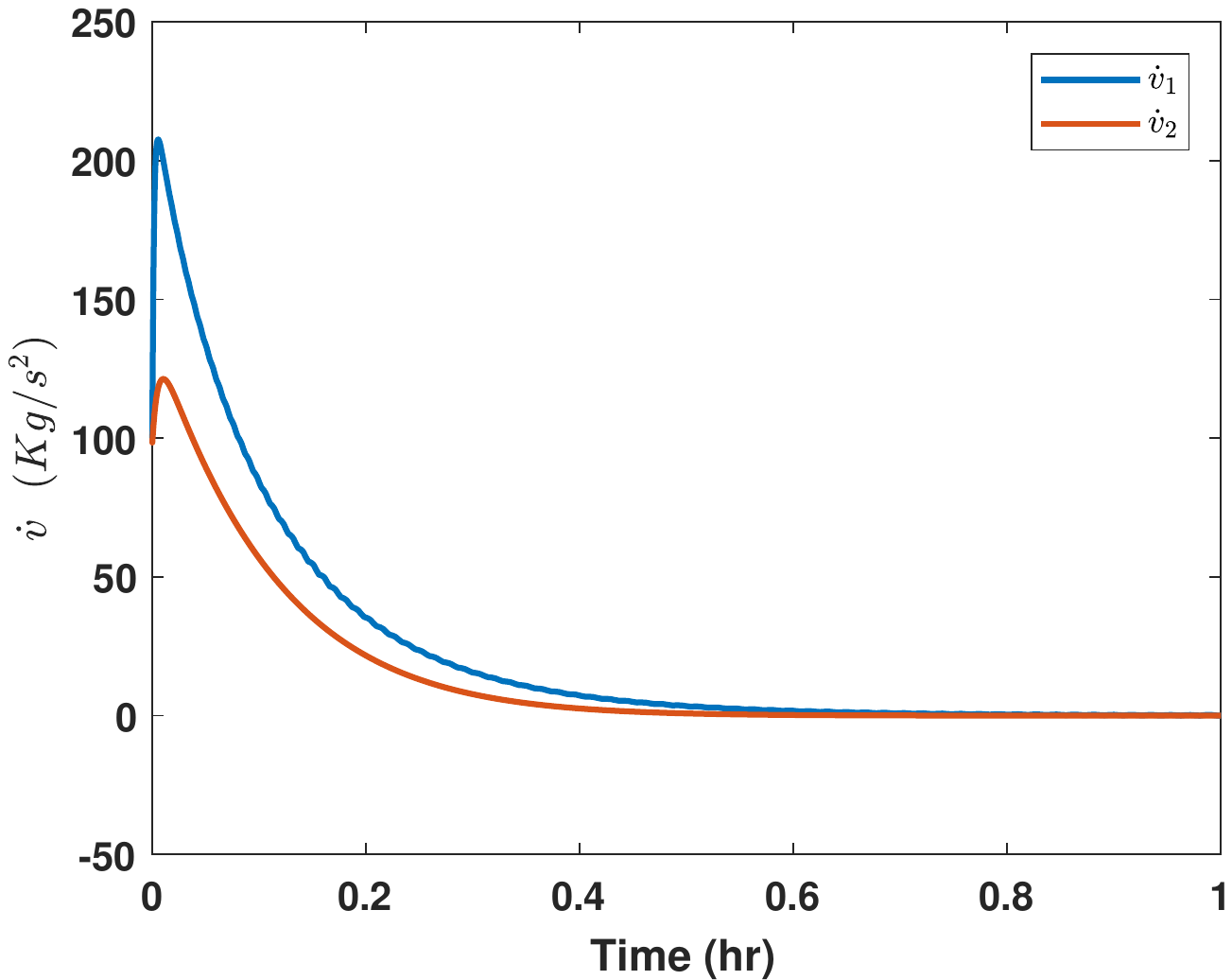}
	\caption{Time evolution of port variable $\dot{v}$ .}
	\label{fig:ports}
\end{figure}
\section{Final Remarks}
In this chapter, we  discuss the issues of finding closed-loop storage function and admissible pairs in control by power shaping. Firstly, we  present a methodology for constructing closed-loop storage function by utilizing the assumption that input matrix is integrable. Secondly, the need for finding admissible pairs is addressed by introducing storage functions similar to Krasovskii-type Lyapunov functions. The use of such storage functions has led to new passive maps, which are used for controller design. These passive maps have differentiation on both the port variables, hence the controller resulted also helped us avoid dissipation obstacle problem. 

%
%
\afterpage{\blankpage}
	\addtocounter{page}{1}%
\chapter{Infinite dimensional system}
 Modeling electrical networks in Brayton-Moser framework is a well-established theory \citep{BraMos1964I,BraMos1964II} and has proven useful in studying the Lyapunov stability of RLC networks. The formulation was extended in \citep{BraMir1964}, to the infinite-dimensional case where the authors developed a pseudo gradient framework to analyze the stability of a transmission line with non-zero boundary conditions. Later control theorists borrowed this framework to generate new passive maps \citep{jeltsema2003passivity,ortega2003power,Jeltsemaa06apower,jeltsema2003dual,DimCleOrtSch} when usual passive maps with energy as storage function render ineffective due to {\em pervasive dissipation} \citep{ElosiaDimOrt}.  Even though BM formulation is well established in finite dimensional systems, it is not  fully extended to the infinite dimensional case.  The existing literature on boundary control of infinite dimensional systems by energy shaping  (in the Hamiltonian case), deals with either lossless systems \citep{Hugo} or partially lossless systems as in \citep{AleCla}, and thus avoids dissipation obstacle issues. Recently in \citep{DimVan07}, the authors presented Brayton-Moser formulation of Maxwell's equations with zero boundary energy flows. However, the admissible pairs given impose restrictions on their spatial domain (such as $\norm{\frac{\partial }{\partial z}}\leq $ 1). \\The main contributions of this chapter are as follows:
\vspace{0.5cm}

\hspace{-0.5cm}\begin{minipage}{0.05\linewidth}
\vspace{-3.8cm}
(i)
\end{minipage}\begin{minipage}{0.95\linewidth}
 {\em {BM formulation:}} In this chapter, we first motivate the need for BM formulation by proving the existence of dissipation obstacle in infinite-dimensional systems using transmission line system as an example. Thereafter, we begin with Brayton-Moser formulation of port-Hamiltonian system defined using Stokes' Dirac structure. In the process, we present its Dirac formulation with a non-canonical bilinear form, similar to the finite dimensional case \citep{Guido}.
\end{minipage}
\vspace{0.5cm}

\hspace{-0.5cm}
\begin{minipage}{0.05\linewidth}
\vspace{-3.1cm}
(ii)
\end{minipage}\begin{minipage}{0.95\linewidth}
  {\em Zero boundary energy flows:} Analogous to the finite-dimensional system, identifying the underlying gradient structure of the system is crucial in analyzing the stability. 
 Therefore we identify alternative Brayton-Moser formulations called admissible pairs, that helps in the stability analysis, with Maxwell's equations as an example.
\end{minipage}

\hspace{-0.5cm}
\begin{minipage}{0.05\linewidth}
\vspace{-5.4cm}
(iii)
\end{minipage}\begin{minipage}{0.95\linewidth}
 {\em {Non-zero boundary energy flows and passivity:}} In case of infinite-dimensional systems with nonzero boundary energy flows, to find admissible pairs for the overall interconnected system, we have to find these admissible pairs for all individual  subsystems, that is, spatial domain and boundary, while preserving the interconnection between these subsystems. To illustrate this, we use the transmission line system (modeled by Telegrapher's equations) where the boundary is connected to a finite dimensional circuit at both ends. This ultimately leads to a new passive map with controlled current and derivatives of the voltage at boundary as port variables.
\end{minipage}
\vspace{0.5cm}

\hspace{-0.5cm}
\begin{minipage}{0.05\linewidth}
\vspace{-3.1cm}
(iv)
\end{minipage}\begin{minipage}{0.95\linewidth}
{\em Boundary control:} Using the new passive map, a {\em passivity based controller} is constructed to solve a boundary control problem (employing control by interconnection), where the original passive maps derived using energy as storage function does not work due to the existence of pervasive dissipation. The control objective is achieved by  generating Casimir functions of the overall systems. 
\end{minipage}
\vspace{0.5cm}

\hspace{-0.5cm}
\begin{minipage}{0.05\linewidth}
\vspace{-3.1cm}
(v)
\end{minipage}\begin{minipage}{0.95\linewidth}
 {\em {Alternative passive maps:}} The passive maps obtained from Brayton Moser formulation, as we have seen earlier in finite-dimensional systems (presented in Chapter 3), impose constraints on systems parameters. We therefore extend the alternative maps methodology developed in Chapter 3.2 (for infinite dimensional systems), and present boundary control methodology using Maxwell's equations.
\end{minipage}
\section{Motivation/Examples}\label{sec::moti}
In this section we show the existence of dissipation obstacle in infinite-dimensional systems, using transmission line system (with non-zero boundary conditions) as an illustrating example.

    \begin{example}\label{moti::example_TL}
    Let $0<z<1$ represent the spatial domain of the transmission line with $L$, $C$, $R$, and $G$ denoting the specific inductance, capacitance, resistance, and conductance respectively. We further assume that these are independent of the spatial variable $z$. Denote by $i(z,t)$ and $v(z,t)$ the line current and line voltage of the transmission line system. Consider the transmission line system (modeled using telegraphers equations) interconnected to the boundary as shown in Figure \ref{fig::abc}. The dynamics of this system are
	\begin{figure}[htp]
		\begin{center}
			\includegraphics[scale=.53]{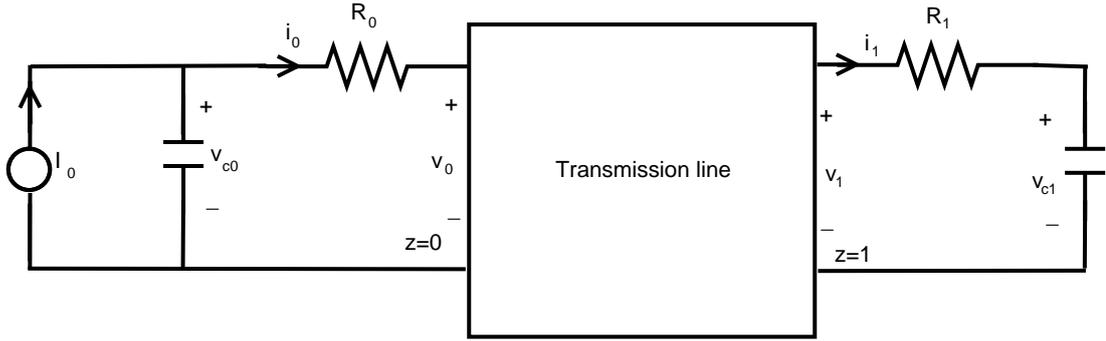} \caption{Transmission line system} \label{fig::abc}
		\end{center}
	\end{figure}
	\begin{align}
	\begin{matrix}
	-Li_t&=& v_z+Ri\\
	Cv_t&=& -Gv-i_z
	\end{matrix}\hspace{8mm}0<z<1\label{telegraphers}\\
	\begin{matrix}
	I_0&=& C_0v_{0t}+i_0\\
	v_0	&=& v_{C_0}-i_0R_0
	\end{matrix}\hspace{13mm}z=0\label{z0}\\
	\begin{matrix}
	i_1&=& C_1v_{C_1t}\\
	v_1&=& R_1i_1+v_{C_1}
	\end{matrix}\hspace{12mm}z=1.\label{z1}
	\end{align}
    	where $v_{C0}$ and $v_{C1}$ denote voltages across the capacitors $C_0$ and $C_1$ respectively and  $I_0$ represents the current source at $z=0$. Additionally, the boundary voltages and currents are denoted by $v_0=v(0,t)$, $i_0=i(0,t)$, $v_1=v(1,t)$ and $i_1=i(1,t)$.
\end{example}
   
   \begin{proposition}\label{prop::dissp}
	The transmission line system described by (\ref{telegraphers}-\ref{z1}) cannot be stabilized \kcn{at} any non-trivial equilibrium point with passive maps obtained by using the total energy
	\beq\label{moti::energy}
	E=\dfrac{1}{2}\int_{0}^{1}\left(Li^2+Cv^2\right)dz+\dfrac{1}{2}C_0v_{c_0}^2+\dfrac{1}{2}C_1v_{c_1}^2
	\eeq
	as the storage function.
\end{proposition}
\begin{proof}Differentiating \eqref{moti::energy} along the trajectories of (\ref{telegraphers}-\ref{z1}), 
we arrive at the following inequality
	\beq\label{mot_pmaps}
	\dot E \leq I_0 v_{C_0}.
	\eeq
	{\em Equilibrium points:} At equilibrium, equations \ (\ref{telegraphers}-\ref{z1}) evaluate to
		\begin{align}
		\begin{matrix}
		i^\ast_z +Gv^\ast&=&0, \hspace{2mm}& Ri^\ast +v^\ast_z&=&0
		\end{matrix}\hspace{10mm}0<z<1\label{eq_telegraphers}\\
		\begin{matrix}
		I_0^\ast &=& i_0^\ast, \hspace{2mm} & v_0^\ast	&=& v_{C_0}^\ast-i_0^\ast R_0
		\end{matrix}\hspace{11mm}z=0\label{eq_z0}\\
		\begin{matrix}
		i_1^\ast &=& 0, \hspace{2mm} & v_1^\ast&=& v_{C_1}^\ast
		\end{matrix}\hspace{12mm}z=1.\label{eq_z1}
		\end{align}
		Finally, solving the partial differential equations in \eqref{eq_telegraphers}, and using the boundary conditions \eqref{eq_z0} and \eqref{eq_z1}, the solution for $i^\ast(z), v^\ast(z)$ takes the form
		\beq
		\begin{matrix}\label{sol_trans_eq}
			i^\ast(z)=\dfrac{G}{\omega}v_{C_1}^\ast \sinh(\omega (1-z)), &
			v^\ast(z)= v_{C_1}^\ast \cosh(\omega (1-z)
		\end{matrix}
		\eeq
		where $\omega =\sqrt{RG}$.
		Using equations  (\ref{eq_z0}-\ref{sol_trans_eq}) it can be shown that the supply rate $I_0^\ast v_{C_0}^\ast  \neq 0$ at equilibrium. This implies \kcn{that at} the equilibrium, the system extracts infinite energy from the controller, thus proving the existence of dissipation obstacle \cite{ortega2001putting}. 
\end{proof}
This problem can be circumvented either by relaxing the assumption that controller has to be passive \citep{KoopDim} or by finding new passive maps \citep{ElosiaDimOrt,venkatraman2009energy,venkatraman2010energy}. In this chapter, we use the latter approach. 
It can be seen from \eqref{mot_pmaps} that ``adding a differentiation" on the output port variable obviates the dissipation obstacle. Recall from Chapter 3, that the port-variables realized from Brayton-Moser framework has this property. 
We hence start with Brayton-Moser formulation of an infinite-dimensional port-Hamiltonian system and derive their admissible pairs, which aids in establishing stability.
%
%
%
%
\section{The Brayton-Moser formulation}
\label{Sec:: From inf dim pH-BM}
In this section\footnote{The notation used in this section, is introduced in Chapter 2.5.}, we present Brayton-Moser formulation of infinite-dimensional port-Hamiltonian system \eqref{pH_2} defined using Stokes' Dirac structure \eqref{ph_2_stokes}, thereby giving its Dirac formulation with a non-canonical bilinear form (refer \citep{Guido} for the finite dimensional equivalent).
To begin with, we assume that the mapping from the energy variables $(\alpha_p, \alpha_q)$ to the co-energy variables
$(e_p, e_q) = (\delta_p H, \delta_q H)$ is invertible. This means the inverse transformation from the co-energy variables
to the energy variables can be written as $(\alpha_p, \alpha_q) =( \delta_{e_p}H^{\ast}, \delta_{e_q}H^{\ast})$.
$H^{\ast}$ is the co-energy of $H$ obtained by $H^{\ast}(e_p, e_q) = \int_Z \left ( e_p \wedge \alpha_p + e_q \wedge \alpha_q\right ) - H(\alpha_p, \alpha_q) $.
Further, assume that the Hamiltonian $H$ splits as $H(\alpha_p, \alpha_q) = H_p(\alpha_p)+ H_q(\alpha_q)$, with the co-energy variables given by
$e_p = \delta_p H_p,~ e_q = \delta_q H_q$. Consequently the co-Hamiltonian can also be split as $H^{\ast}(e_p, e_q) = H_p^{\ast}(e_p) + H_q^{\ast}(e_q)$. We can now rewrite the spatial dynamics of the infinite-dimensional port-Hamiltonian system, in terms of the co-energy variables as
\beq
\begin{bmatrix}
	\ast \delta^2_{e_p} H^{\ast} & 0 \\
	0 & \ast \delta^2_{e_q} H^{\ast}
\end{bmatrix}
\left[\begin{matrix}
	-\frac{\partial e_p}{\partial t} \\ -\frac{\partial e_q}{\partial t}
\end{matrix}\right]&=&\begin{bmatrix}
\ast G & (-1)^r\mathrm{d}\\\mathrm{d} & \ast R
\end{bmatrix}\begin{bmatrix}
e_p\\ e_q
\end{bmatrix}.
\label{proof_Mix_1}
\eeq
For simplicity, we assume that the relation between the energy and co-energy variables is linear and is given as
\begin{equation}
\alpha_p = \ast \epsilon \;  e_p \; \text{and}\;  \alpha_q = \ast \mu\; e_ q \label{rel_e,co-e} 
\end{equation}
where $\mu(=\delta^2_{e_q} H^{\ast})$, $\epsilon(= \delta^2_{e_p} H^{\ast}) \in \mathbb{R}$. Applying \djc{the} Hodge star \djc{operator to} both sides \djc{of} \eqref{proof_Mix_1} and arranging terms using \eqref{rel_e,co-e}, \djc{we get}
\beq-\epsilon \dot{e}_p &=& \ast \left((-1)^{r} \mathrm{d} e_q+G  \ast  e_p\right)(-1)^{(n-p)\times p},\nonumber \\
-\mu \dot{e}_q &= &\ast \left(\mathrm{d}e_p+R \ast e_q\right)(-1)^{(n-q)\times q}. \label{proof_Mix_a} \eeq
Next, we find a mixed-potential function $P=\int_Z\text{P}(e_p,e_q)$ such that \eqref{proof_Mix_a} can take the pseudo-gradient structure \citep{DimVan07}. \\
{\em The lossless case:} We  first consider the case of a system \djc{that} is lossless, that is, when $R$ and $G$ are identically equal to zero in \eqref{pH_2}. To begin with, we also neglect the boundary terms by setting them to zero. Define $P$ to be a functional of the form $P=\int_Z\text{P}(e_p,e_q)$, where
\beq
\text{P}(e_p,e_q):= e_q\wedge de_p. 
\eeq 
Its variation is given as
\begin{align*} 
\delta P &= 
\int_Z\left(\text{P}(e_p+\partial e_p, e_q+\partial e_q)-\text{P}(e_p,e_q)\right)
= \int_Z\left(e_q \wedge \mathrm{d}\partial e_p +\partial e_q \wedge \mathrm{d} e_p +\cdots\right). 
\end{align*}
Using the relation  $e_q\wedge \mathrm{d} \partial e_p=(-1)^{pq}\partial e_p \wedge \mathrm{d}e_q+(-1)^{n-q}\mathrm{d}\left(e_q \wedge \partial e_p\right)$, and the identity $\eqref{delta}$, we have 
$$ 
\delta_{e_q}P=\mathrm{d}e_p(-1)^{(n-q)\times q}, ~~ \delta_{e_p}P=(-1)^{pq}\mathrm{d}e_q(-1)^{(n-p)\times p}.
$$
Finally, utilizing the above equation; \eqref{proof_Mix_1} can be rewritten in the BM-type as
\beq 
\begin{bmatrix}
	-\mu & 0\\ 0 & \epsilon
\end{bmatrix}\frac{\partial}{\partial t}\begin{bmatrix}
e_q\\ \ e_p
\end{bmatrix}&=& \begin{bmatrix}
\ast \delta_{e_q}P\\ \ast \delta_{e_p}P
\end{bmatrix}.
\label{proof_Mix}
\eeq
{\em Including dissipation:} One may allow for dissipation by defining the content and co-content functions as follows. Consider instead a functional $P=\int_Z\text{P}$ defined as
\beq
\text{P}(e_p,e_q)= e_q \wedge \mathrm{d}e_p+\underbrace{\text F(e_q)\text{Vol}}_\text{content}- \underbrace{\text G(e_p)\text{Vol}}_\text{co-content} \label{mix_pot_1}
\eeq
where $\text{Vol}\in \Omega^n(Z)$ such that $\int_Z\text{Vol}\wedge \ast \text{Vol} =1$, the content $\text F(e_q)$ and the co-content $\text G(e_p)$ functions are defined respectively as
\begin{eqnarray}
\begin{matrix}
\text F(e_q)&=&\int_0^{e_q}\left<\hat{e}_p(e_q^{'}),de_q^{'}\right>, &
\text G(e_p)&=&\int_0^{e_p}\left<\hat{e}_q(e_p^{'}),de_p^{'}\right> 
\end{matrix}
\end{eqnarray}
where the inner product $\left<\cdot ,\cdot\right>$ is induced by the Riemannian metric defined on $Z$. In the case of linear dissipation  \eqref{pH_2}, that is $\hat{e}_p(e_q)=Re_q$ and $\hat{e}_q(e_p)=Ge_p$ we have
\beq
\text{P}(e_p,e_q)&=&e_q \wedge \mathrm{d}e_p+\int_0^{e_q}\left<Re_q^{'},de_q^{'}\right>\text{Vol}-\int_0^{e_p}\left<Ge_p^{'},de_p^{'}\right>\text{Vol} \nonumber\\
&=&e_q \wedge \mathrm{d}e_p+\dfrac{1}{2}\left<Re_q,e_q\right>\text{Vol}-\dfrac{1}{2}\left<Ge_p,e_p\right>\text{Vol} \nonumber\\
&=& e_q \wedge \mathrm{d}e_p+\underbrace{\frac{1}{2}R e_q \wedge \ast e_q}_\text{content}- \underbrace{\frac{1}{2}G e_p \wedge \ast e_p}_\text{co-content} \label{mix_pot}
\eeq
where in the third step we have used \eqref{eqn::G}. The variation in $P$  is computed as
\beqn
\delta P &=&\hspace{-1mm} \int_Z\left(e_q \wedge \mathrm{d}\partial e_p +\partial e_q \wedge \mathrm{d} e_p + \frac{1}{2}(e_q \wedge R \ast \partial e_q+\partial e_q \wedge \ast e_q)\right.\\&&- \left.\frac{1}{2}(e_p \wedge G \ast \partial e_p+\partial e_p \wedge \ast e_p\right) \\
&=& \hspace{-1mm}\int_Z\left( \partial e_q \wedge \mathrm{d}e_p+\partial e_p \wedge (-1)^{pq} \mathrm{d} e_q+ \frac{1}{2}(e_q \wedge R \ast \partial e_q\hspace{-1mm}+\hspace{-1mm}\partial e_q \wedge \ast e_q)\hspace{-1mm}\right. \\&&\left.- \frac{1}{2}(e_p \wedge G \ast \partial e_p\hspace{-1mm}+\hspace{-1mm}\partial e_p \wedge \ast e_p)+(-1)^{n-q}\mathrm{d}\left(e_q \wedge \partial e_p\right) \right)\\
&=&\int_Z \partial e_q \wedge \left(\mathrm{d} e_p+R \ast e_q \right)+\partial e_p \wedge \left((-1)^{pq} \mathrm{d} e_q-G \ast e_p\right)\\&&+(-1)^{n-q}\int_{\partial Z}\left(e_q \wedge \partial e_p\right)
\eeqn
where we have used the relation  $e_q\wedge \mathrm{d} \partial e_p=(-1)^{pq}\partial e_p \wedge \mathrm{d}e_q+(-1)^{n-q}\mathrm{d}\left(e_q \wedge \partial e_p\right)$, together with properties of the wedge and the Hodge star operator defined in \eqref{Ab} and \eqref{Ad}. 
Finally, by making use of  \eqref{delta} we can write 
\beq
\begin{bmatrix}
	\delta_{e_p}P\\ \delta_{e_q}P\\\delta_{e_p|_{\partial z}}P\\ \delta_{e_q|_{\partial z}}P
\end{bmatrix} &=& \begin{bmatrix}
\left((-1)^{pq} \mathrm{d} e_q-G \ast e_p\right)(-1)^{(n-p)\times p}\\(\mathrm{d}e_p +R\ast e_q)(-1)^{(n-q)\times q}\\(-1)^{n-q}e_q|_{\partial z}\\0
\end{bmatrix}.
\label{proof_Mix_2}
\eeq
The system of equations (\ref{proof_Mix_1}) can be written in a concise way, similar to \eqref{proof_Mix} as
\begin{align}
A \djc{u_t} = \ast \delta_u P
\label{eq_BMI}
\end{align}
where $u=\djc{[}e_p,e_q\djc{]}^\top$ and 
$ A = \djc{\text{diag}(\epsilon,-\mu)}$. Note that if the linearity between energy and co-energy variables is not assumed \eqref{rel_e,co-e} then $A$ takes the form $\text{diag}(-\delta^2_{e_q} H^{\ast},\delta^2_{e_p} H^{\ast})$.\\
{\em Including boundary energy flow:} The system of equations \eqref{pH_2} together with boundary terms can be rewritten as
\beq \label{BM_boundary}
\begin{matrix}
	\mathcal{A}U_t&=&\ast \delta_UP+B\ast e_b\\
	\dot{f}_b&=&B^\top U_t\left(=\dot{e}_p|_{\partial z}\right)
\end{matrix}
\eeq
where $U=[u;u|_{\partial z}]$, $B=[O_1\; I\; O_2]^\top$ and $\mathcal{A}=diag\{A,O_3\}$ with $O_1,\;O_2,\;O_3$ denoting zero matrices of order $(n+1\times n-q),\; (n-p \times n-q),\; (n+1 \times n+1)$ respectively and $I$ identity matrix of order $(n-q)$.
\subsection {\kcw{The Dirac formulation}}\label{subsec::dirac formulation}
 In this section, we aim to find an equivalent Dirac structure formalism of the Brayton-Moser equations of the infinite-dimensional system \eqref{BM_boundary},  (for an overview of Dirac structure of infinite dimensional systems we refer to \citep{YanhanMas}). As we shall see such a formulation would result in a noncanonical Dirac structure. 
Denote by $f  \in \mathcal F:=\Omega^{n-p}(Z)\times\Omega^{n-q}(Z)\times\Omega^{n-p}(\partial Z)\times\Omega^{n-q}(\partial Z)$ as the space of flows and $e \in \mathcal{E}:= \mathcal{F}^\ast$, as the space of effort variables.
\begin{theorem}\label{noncan_dirac}
    Consider the following subspace
\begin{align}\label{Dirac_gen_struc}
\mathcal D = \left \{(f,f_y, e,e_u) \in \mathcal F   \times \mathcal Y\times \mathcal E \times \mathcal S :  -\mathcal A f =\ast  e+Be_u,~f_y=\ast B^\top f\right \}
\end{align}
where $\mathcal{S}$, $\mathcal{Y}$ represents space of port variables  $e_{u}$ and $f_y$ respectively defined on $\partial Z$. 
The subspace $\mathcal{D}$ constitutes a noncanonical Dirac structure, that is $\mathcal{D}=\mathcal{D}^\perp$, $\mathcal{D}^{\perp}$ is the orthogonal complement of $\mathcal{D}$ with respect to the bilinear form\\
$<<(f^1,f_{y}^1,e^1,e_u^1),(f^2,f_{y}^2,e^2,e_u^2) >>$
\beq
\hspace{-6mm}= &\hspace{-2mm}\left<e^1|f^2\right>+ \left<e^2|f^1\right>+ \int_{\left(Z+\partial Z\right)}   \left (f^1 \wedge \ast \mathcal A  f^2+f^2 \wedge \ast \mathcal A  f^1 \right ) + \left<e_u^1|f_y^2\right>+\left<e_u^2|f_y^1\right>
\label{bileniarform}
\eeq
\kc{	where $\mathcal{A}:\mathcal F\rightarrow \mathcal{F}$, for $i=1,2~$;
$
	\begin{matrix}
		f^i  \in \mathcal F,  &f_y^i\in  \mathcal Y ,&	e^i  \in \mathcal E, &  e_{u}^i\in  \mathcal S 
	\end{matrix}.
	$}
\end{theorem}
\begin{proof}
    We follow a similar procedure as in \citep{stokes}. We first show that $\mathcal{D}\subset \mathcal{D}^{\perp}$, and secondly $\mathcal{D}^{\perp}\subset \mathcal{D}$. 
\newline
{\underline{\em Case (i)} $\mathcal{D}\subset \mathcal{D}^{\perp}$} :\\
 Consider $(f^1,f_{y}^1,e^1,e_u^1)\in \mathcal{D}$, it suffices to show $(f^1,f_{y}^1,e^1,e_u^1)\in \mathcal{D}^{\perp}$ then $\mathcal{D}\subset \mathcal{D}^{\perp}$. Now consider any $(f^2,f_{y}^2,e^2,e_u^2)\in \mathcal{D}$ i.e. satisfying \eqref{Dirac_gen_struc}, substituting in the bilinear form \eqref{bileniarform} gives
        $<<(f^1,f_{y}^1,e^1,e_u^1),(f^2,f_{y}^2,e^2,e_u^2) >>$
        \beqn
&= &\left<e^1|f^2\right>+ \left<e^2|f^1\right>+ \int_{\left(Z+\partial Z\right)}   \left (f^1 \wedge \ast \mathcal A  f^2+f^2 \wedge \ast \mathcal A  f^1 \right ) + \left<e_u^1|f_y^2\right>+\left<e_u^2|f_y^1\right>\\
&= &\left<e^1|f^2\right>+ \left<e^2|f^1\right>- \int_{\left(Z+\partial Z\right)}   \left (f^1 \wedge \ast\left(\ast e^2+Be_u^2\right)+f^2 \wedge \ast\left(\ast e^1+Be_u^1\right) \right ) \\
&&+ \left<e_u^1|\ast B^\top f^2\right>+\left<e_u^2|\ast B^\top f^1\right>\\
&= &\left<e^1|f^2\right>+ \left<e^2|f^1\right>-\left<e^1|f^2\right>- \left<e^2|f^1\right> -\left<e_u^1\ast B^\top f^2\right>-\left<e_u^2|\ast B^\top f^1\right> \\ &&+\left<e_u^1|\ast B^\top f^2\right>+\left<e_u^2|\ast B^\top f^1\right>\\
&=&0
\eeqn
where in step 2 we used the properties of wedge product  \eqref{Ac} and \eqref{Ad}, that is,
\beq\label{dric_proof_prop}
\begin{matrix}
    f^1\wedge\ast \ast e^2=e^2\wedge f^1\\
    f^2\wedge\ast \ast e^1=e^1\wedge f^2\\
    f^1\wedge \ast B e^2_u=B e^2_u\wedge \ast f^1=e^2_u\wedge \ast B^\top f^1\\
     f^2\wedge \ast B e^1_u=B e^1_u\wedge \ast f^2=e^1_u\wedge \ast B^\top f^2
\end{matrix}
\eeq
This implies $(f^1,f_{y}^1,e^1,e_u^1)\in \mathcal{D}^{\perp}$ implying $\mathcal{D}\subset \mathcal{D}^{\perp}$.\\
{\underline{\em Case (ii)} $\mathcal{D}^{\perp}\subset \mathcal{D}$}  :\\
 Consider $(f^1,f_{y}^1,e^1,e_u^1)\in \mathcal{D}^{\perp}$ and  if we show that  $(f^1,f_{y}^1,e^1,e_u^1)\in \mathcal{D}$ then we are through. Now consider any $(f^2,f_{y}^2,e^2,e_u^2)\in \mathcal{D}$, implies
\beq \label{case2_dirac_proof}<<(f^1,f_{y}^1,e^1,e_u^1),(f^2,f_{y}^2,e^2,e_u^2) >>=0 \eeq
which upon simplifying the left hand side of \eqref{case2_dirac_proof} we get
\beqn
&= &\left<e^1|f^2\right>+ \left<e^2|f^1\right>+ \int_{\left(Z+\partial Z\right)}   \left (f^1 \wedge \ast \mathcal A  f^2+f^2 \wedge \ast \mathcal A  f^1 \right ) + \left<e_u^1|f_y^2\right>+\left<e_u^2|f_y^1\right>\\
&= &\left<e^1|f^2\right>+ \left<e^2|f^1\right>- \int_{\left(Z+\partial Z\right)}   \left (f^1 \wedge \ast \left(\ast e^2+Be_u^2\right) \right ) + \int_{\left(Z+\partial Z\right)}   \left (f^2 \wedge \ast \mathcal A  f^1 \right )\\&&+ \left<e_u^1|\ast B^\top f^2\right>+\left<e_u^2|f_y^1\right>\\
&=& \int_{\left(Z+\partial Z\right)}   \left (f^2 \wedge \ast \left(\mathcal A  f^1 +\ast e^1+Be_u^1\right)\right )+\left<e_u^2|\left(f_y^1-\ast B^\top f^1\right)\right>
\eeqn
where in step 2 we used the fact that $(f^2,f_{y}^2,e^2,e_u^2)\in \mathcal{D}$, and in step 3 we used the wedge operator properties in \eqref{dric_proof_prop}. From \eqref{case2_dirac_proof}, for all $f^2$, $e^2_u$
\beq
\int_{\left(Z+\partial Z\right)}   \left (f^2 \wedge \ast \left(\mathcal A  f^1 +\ast e^1+Be_u^1\right)\right )+\left<e_u^2|\left(f_y^1-\ast B^\top f^1\right)\right>=0.
\eeq
This clearly implies 
\beqn
\mathcal A  f^1 +\ast e^1+Be_u^1=0\\
f_y^1-\ast B^\top f^1=0
\eeqn
proving that $(f^1,f_{y}^1,e^1,e_u^1)\in \mathcal{D}$.
\end{proof}
\begin{proposition}
The port-Hamiltonian system \eqref{pH_2} or  Brayton-Moser equations \eqref{BM_boundary} can be equivalently written as a dynamical system with respect to the noncanonical Dirac structure $\mathcal{D}$ in Theorem \ref{noncan_dirac} by setting
\beq\label{ffyeeu}
(f,f_y,e,e_u)=\left ( -U_t,~-\ast \dot{f}_b,~ \delta_U\text{P},~\ast e_b\right ).
\eeq
Moreover, the noncanonical bilinear form \eqref{bileniarform} evaluates to the ``power balance equation'' 
\beq\label{gen::bal eqn}
\dfrac{\partial}{\partial t}\mathcal P &=&  \int_Z u_t\wedge \ast A u_t - \int_{\partial Z}\left(e_b\wedge \dot{f}_b\right).
\eeq
\end{proposition}
\begin{proof}
    The first part of the Proposition can be verified by using \eqref{ffyeeu} in the Dirac structure \eqref{Dirac_gen_struc}. For the second part, consider the following. The bilinear form \eqref{bileniarform} is assumed to be non-degenerate, hence $\mathcal D=\mathcal{D}^{\perp}$ implies
    \beqn
    <<(f,f_y,e,e_u),(f,f_y,e,e_u)>>=0,~~~\forall\;\; (f,f_y,e,e_u)\in \mathcal{D} 
    \eeqn
and can be simplified to
\beq\label{Power_balence}
\left<e|f\right>+ \int_{\left(Z+ \partial Z\right)}   \left (f \wedge \ast \mathcal A  f\right ) + \left<e_u|f_y\right>&=&0
\eeq
finally using \eqref{ffyeeu} we arrive at the power balance equation \citep{Guido,blankenstein2003joined}, given in \eqref{gen::bal eqn}. We can now interconnect \eqref{BM_boundary} to other BM systems defined at the boundary $\partial Z$ using these new port variables $e_b$ and $-\dot{f}_b$.
\end{proof}
\subsection{A Passivity argument} 
Once we have written down the equations in the BM framework (sometimes also referred to as the pseudo gradient form \new{\citep{DimVan07}}) we can pose the following question; does the mixed potential function serve as a storage function (or a \kcn{Lyapunov} function) to infer passivity (or equivalently stability) properties of the system? A first look at the balance equation \eqref{gen::bal eqn} might suggest that the system in the BM form \eqref{BM_boundary} is passive with $\mathcal P$ serving as the storage function and port variables $-\dot f_{b}$ and $e_b$.
Similar to the exposition presented in Chapter 3 for finite-dimensional systems and also see \citep{BraMir1964,DimVan07} for infinite-dimensional systems, this is not the case, as the mixed potential function $\mathcal P$, and its time derivative \eqref{gen::bal eqn} are  sign in-definite and hence does not serve as a storage function.
This motivates our quest for finding a new $\mathcal{P}\geq 0$ and $\mathcal{A}\leq 0$, called as admissible pairs, enabling us to derive certain new passivity/stability properties (analogous to the ones presented in Equation \eqref{BM_adm} for finite-dimensional systems). \ijc{This work} aims to answer these issues.
\section{Systems without boundary interaction}\label{sec:WOB}
To infer stability properties of the system \eqref{eq_BMI}, let us begin with the case of zero energy flow through the boundary of the system. The mixed-potential function \eqref{mix_pot} is not positive definite. Hence, we cannot use it as a Lyapunov or storage functional. Moreover, the \djc{rate of change} of this function is computed as
$$
\dot{\djc{P}}= \int_{\djc{Z}} \left ( -\mu \dot{e_p} \wedge \ast \dot{e_p} +\epsilon \dot{e_q} \wedge \ast \dot{e_q}  \right ), \nonumber
$$
\djc{which clearly is not sign-definite.}
We thus need to look for other {\em admissible pairs} $(\tilde{A}$, $\tilde{P})$ like in the case of
	finite-dimensional systems \eqref{BM_adm} \citep{jeltsema2003passivity} that can be used to
	prove stability of the system while preserving the
	dynamics of \eqref{eq_BMI}. 
	Moreover, the admissible pair should be such
that the symmetric part of $\tilde{A}$ is negative semi-definite. This
can be achieved in the following way \citep{BraMir1964,DimVan07}. 
\subsection{Admissible pairs} \label{infi_admissible_pairs}
Consider the functional $\tilde{P}=\int_Z\tilde{\text{P}}$ of the form
\beq\label{mix_pot_adm_gen}
\tilde{P} = \lambda P+\frac{1}{2}\int_Z \left(\delta_{e_p}P \wedge M_1\ast \delta_{e_p}P+\delta_{e_q}P \wedge M_2\ast \delta_{e_q}P\right),
\eeq
with $\lambda\in \mathbb{R}$ is \djc{an} arbitrary constant and the symmetric \djc{mappings} $M_1:\Omega^p(Z) \rightarrow \Omega ^p(Z)$ and $M_2:\Omega^q(Z) \rightarrow \Omega ^q(Z)$ are linear. Here, the aim is to find $\lambda$, $M_1$ and $M_2$ such that 
\beq
\dot{\tilde{P}} = u_t^\top \tilde{A} u_t  \le  -K ||u_t||^2 \le 0,\label{BM_admis}
\eeq
\kc{where $K\ge0$ represents the magnitude of smallest eigenvalue of $\tilde{ A}$}. 
If we can find such a \djc{pair} $(\tilde { P}, \tilde { A})$, which satisfies \eqref{BM_admis}, then we can conclude stability of the system \eqref{eq_BMI}.
\begin{theorem}\label{genPA} The system of equations \eqref{eq_BMI} have the alternative BM representation 
	$\tilde{A}u_t=\ast \delta_u \tilde{P}$ with $\tilde{P}$ defined as in \eqref{mix_pot_adm_gen} and
	\beq\label{genricA}
	\tilde{A}  &&\deff \begin{bmatrix}
		-\mu\left(\lambda I+	R^\top M_1\right) & \epsilon M_2 \ast \mathrm{d} (-1)^{(n-p)\times p}\\
		-\mu (-1)^qM_1\ast \mathrm{d} & \epsilon\left(\lambda I-G^\top M_2\right)
	\end{bmatrix}.
	\eeq
	The new mixed potential function satisfies, $\tilde{P}\geq 0$ for $-\norm{M_1R}_s < \lambda  <\norm{M_2 G}_s$, where $\norm{\cdot}_s$ denotes the spectral norm. Additionally, for systems with $p=q$ and $\epsilon M_2 =\mu M_1$; symmetric part of $\tilde{A}$ is negative definite.
\end{theorem}
\begin{proof}
We start with finding the variational derivative of $\tilde{P}$. Consider the term $\delta_{e_p}P \wedge M_1\ast \delta_{e_p}P$
\begin{eqnarray*}
&=&\left((-1)^{pq} \mathrm{d} e_q -\ast G e_p\right) \wedge M_2 \ast \left((-1)^{pq} \mathrm{d} e_q -\ast G e_p\right)\\
&=& \mathrm{d} e_q \wedge M_2 \ast \mathrm{d}e_q -(-1)^{pq} \mathrm{d}e_q \wedge M_2 \ast \ast G e_p -(-1)^{pq} \ast G e_p \wedge M_2 \ast \mathrm{d}e_q\\&& +\ast G e_p \wedge M_2 \ast \ast G e_p\\
&=&\mathrm{d} e_q \wedge M_2 \ast \mathrm{d}e_q-(-1)^p\mathrm{d}e_q \wedge M_2 G e_p+e_p\wedge\ast G^\top M_2Ge_p.
\end{eqnarray*}
The variation in first term $\mathrm{d} e_q \wedge M_2 \ast \mathrm{d}e_q$ is
\begin{flalign*}
&\mathrm{d}(e_q+\partial e_q) \wedge \ast M_2 \mathrm{d}(e_q+\partial e_q)-\mathrm{d} e_q \wedge M_2 \ast \mathrm{d}e_q\\&=\mathrm{d} \partial e_q \wedge \ast M_2 \mathrm{d} e_q + \mathrm{d} e_q \wedge \ast M_2 \mathrm{d} \partial e_q\textbf{}+ \cdots&&\\
&= 2 \mathrm{d} \partial e_q \wedge \ast M_2 \mathrm{d} e_q+\cdots&&
\end{flalign*}
the variation in the second term $\mathrm{d}e_q \wedge M_2 G e_p$ is
\begin{flalign*}
&\mathrm{d}(e_q+\partial e_q) \wedge M_2 G (e_p+\partial e_p)- \mathrm{d} e_q \wedge M_2 G e_p\\
&=\mathrm{d}e_q \wedge M_2 G \partial e_p + \mathrm{d}\partial e_q \wedge M_2 G e_p+\cdots&&\\
&=\partial e_p \wedge G^ \top M_2 \mathrm{d}e_q (-1)^{(n-p)\times p} + \mathrm{d}\partial e_q \wedge M_2 G e_p+\cdots&
\end{flalign*}
and finally the variation in the last term $e_p\wedge\ast G^\top M_2Ge_p$ is given by
\begin{flalign*}
&(e_p+\partial e_p) \wedge \ast G^\top M_2 G (e_p+\partial e_p)- e_p \wedge \ast G^\top  M_2 G e_p =\partial e_p \wedge 2  \ast G^\top M_2 G e_p.&&
\end{flalign*}
By the properties of the exterior derivative,
\begin{flalign*}
\mathrm{d}(\partial e_q \wedge \ast M_2 \mathrm{d}e_q) &=\mathrm{d} \partial e_q \wedge \ast M_2 \mathrm{d} e_q+\partial e_q \wedge (-1)^{(n-q)} \mathrm{d} \ast \mathrm{d} M_2 e_q&&\\
\mathrm{d}(\partial e_q \wedge M_2 G e_p) &= \mathrm{d}\partial e_q \wedge M_2 G e_p + (-1)^{n-q} \partial e_q \wedge M_2 G \mathrm{d} e_p&&
\end{flalign*}
the variation in $\delta_{e_p}P \wedge M_1\ast \delta_{e_p}P$ can be simplified to as
\begin{flalign*}  
&\partial e_q \wedge 2 \left(  (-1)^{p} \mathrm{d} \ast \mathrm{d} M_2 e_q -M_2 G\mathrm{d}e_p\right) +\partial e_p \wedge 2 \left( (-1)^{pq+1}G^\top M_2 \mathrm{d}e_q + \ast G^\top M_2 G e_p  \right)\\
&=\partial e_q \wedge 2(-1)^{(n-p)\times p}M_2 \mathrm{d} \ast  \left( (-1)^{pq} \mathrm{d}e_q - \ast G e_p  \right) &&\\&\hspace{0.5cm}+\partial e_p \wedge  -2G^\top M_2\left( (-1)^{pq} \mathrm{d}e_q - \ast G e_p  \right).&&
\end{flalign*}
Similarly the variation in $\delta_{e_q}P \wedge M_1\ast \delta_{e_q}P$ is calculated as 
\begin{flalign*}  
&\partial e_q \wedge 2 \left( R^\top M_1 \mathrm{d}e_p +\ast R^\top M_1 R e_q \right) +\partial e_p \wedge 2 \left((-1)^q \mathrm{d} \ast \mathrm{d} M_1 e_p +(-1)^{pq} M_1 R \mathrm{d}e_q \right)\\
&= \partial e_q \wedge 2R^\top M_1 \left( \mathrm{d}e_p +\ast  R e_q \right) +\partial e_p \wedge 2(-1)^qM_1\mathrm{d} \ast \left(\mathrm{d}  e_p + \ast  R \mathrm{d}e_q \right).&&
\end{flalign*}  
Together the variational derivative of $\tilde{P}$ can be computed as
\beqn
\delta\tilde{P} \deff \begin{bmatrix}
	\lambda I+	R^\top M_1 & M_2 \mathrm{d} \ast (-1)^{(n-p)\times p}\\
	(-1)^qM_1\mathrm{d} \ast & \lambda I-G^\top M_2
\end{bmatrix}\begin{bmatrix}
	\delta_{e_p}P\\
	\delta_{e_q}P
\end{bmatrix}.
\eeqn
Further
\beqn
\ast \delta\tilde{P} &=& \begin{bmatrix}
	\lambda I+	R^\top M_1 & M_2 \ast \mathrm{d} (-1)^{(n-p)\times p}\\
	(-1)^qM_1\ast \mathrm{d} & \lambda I-G^\top M_2
\end{bmatrix}\left(\ast \begin{bmatrix}
	\delta_{e_p}P\\
	\delta_{e_q}P
\end{bmatrix}\right)\\
&=& \begin{bmatrix}
	\lambda I+	R^\top M_1 & M_2 \ast d (-1)^{(n-p)\times p}\\
	(-1)^qM_1\ast d & \lambda I-G^\top M_2
\end{bmatrix}\begin{bmatrix}
	-\mu & 0\\ 0 & \epsilon
\end{bmatrix}\begin{bmatrix}
	\dot{e}_q\\\dot{e}_p
\end{bmatrix}\\
&=& \begin{bmatrix}
	-\mu\left(\lambda I+	R^\top M_1\right) & \epsilon M_2 \ast \mathrm{d} (-1)^{(n-p)\times p}\\
	-\mu (-1)^qM_1\ast \mathrm{d} & \epsilon\left(\lambda I-G^\top M_2\right)
\end{bmatrix}\begin{bmatrix}
	\dot{e}_q\\\dot{e}_p
\end{bmatrix}
=\tilde{A}u_t.
\eeqn
This concludes the first part of the proof. We next to show the positive definiteness of $\tilde{P}$. Before that we simplify $P$ in \eqref{mix_pot} as follows:
\beqn
\text P(e_p,e_q)&=&e_q \wedge \mathrm{d}e_p+\frac{1}{2}R e_q \wedge \ast e_q- \frac{1}{2}G e_p \wedge \ast e_p  \\
&=&\frac{R^{-1}}{2}\left(\ast R e_q \wedge \ast \ast  R e_q  +\mathrm{d}e_p \wedge \ast \ast R e_q + \ast Re_q \wedge \ast \mathrm{d}e_p\right.\\&&\left. +\mathrm{d}e_p \wedge \ast \mathrm{d}e_p- \mathrm{d}e_p \wedge \ast \mathrm{d}e_p\right)-\frac{1}{2} Ge_p \wedge \ast e_p  \\
&=& \frac{R^{-1}}{2}\left(\delta_{e_p}P \wedge \ast   \delta_{e_p}P \right) -\frac{R^{-1}}{2}\mathrm{d}e_p \wedge \ast \mathrm{d}e_p -\frac{1}{2} Ge_p \wedge \ast e_p 
\label{p1}
\eeqn
for $-\norm{M_1R}_s < \lambda  <0$ we have
\beqn
\tilde{\text P}  &=&  \frac{\lambda R^{-1}+ M_1 }{2} \left(\delta_{e_p}P  \wedge  \ast   \delta_{e_p}P \right) - \frac{\lambda R^{-1}}{2}\mathrm{d}e_p  \wedge  \ast \mathrm{d}e_p   - \frac{\lambda I}{2} Ge_p  \wedge  \ast e_p  \\&&+ \frac{M_2}{2}  \left( \delta_{e_q}P  \wedge  \ast  \delta_{e_q}P \right) \\&>&  0.
\eeqn
In a similar way we can show that
\beq \text P(e_p,e_q)= -\frac{G^{-1}}{2}\left( \delta_{e_q}P \wedge \ast  \delta_{e_q}P \right) +\frac{G^{-1}}{2}\mathrm{d}e_q \wedge \ast \mathrm{d}e_q  +\frac{1}{2} Re_q \wedge \ast e_q \label{p2}
\eeq
hence for $0 < \lambda  <\norm{M_2G}_s$ we have
\beqn \tilde{\text P}  &=&   -\frac{\lambda G^{-1}  -  M_2 }{2}\left( \delta_{e_q}P \wedge \ast  \delta_{e_q}P \right)  + \frac{\lambda G^{-1}}{2}\mathrm{d}e_q  \wedge  \ast \mathrm{d}e_q  + \frac{\lambda}{2} Re_q  \wedge  \ast e_q \\&&+ \frac{M_1}{2} \left( \delta_{e_p}P  \wedge  \ast   \delta_{e_p}P \right) \\&>&  0 
\eeqn
concluding that $\tilde P$ is positive definite for $-\norm{M_1R}_s < \lambda  <\norm{M_2 G}_s$. Furthermore, with $p=q$ and $\epsilon M_2 =\mu M_1$ one can prove that symmetric part of $\tilde{A}$ is negative definite.
\end{proof}
\begin{remark}
 	Note that, if we do not restrict $M_1$ and $M_2$ such that $\epsilon M_2 =\mu M_1$ in Theorem \ref{genPA}, then for the symmetric part of $\tilde{A}\leq 0$ will lead to constraints on spatial domain like $	\sigma^{-1}\sqrt{\epsilon \mu^{-1}}\norm{\ast \mathrm{d}}<1$, as given in \citep{DimVan07}. 
\end{remark}
\subsection{Stability of Maxwell's equations}
\begin{example}[Maxwell's equations] \label{Example::Maxwell}
    \djc{Consider an electromagnetic medium with spatial domain $Z \subset \mathbb R^3$ with a smooth two-dimensional boundary $\partial Z$. The energy variables \ijc{($2$-forms on $Z$)} are the \djc{electric field} induction $\mathcal{D}=\frac{1}{2}\mathcal{D}_{ij}z_i \wedge z_j$ and the magnetic field induction $\mathcal{B}= \frac{1}{2}\mathcal{B}_{ij}z_i \wedge z_j$ on $Z$. The associated co-energy variables are electric field intensity $\mathcal{E}$ and magnetic field intensity $\mathcal{H}$. These co-energy variables ($1$-forms) are linearly related to the energy variables through the constitutive relationships of the medium as}
\beq 
\ast \mathcal{D} = \epsilon \mathcal{E}, \
\ast \mathcal{B} = \mu \mathcal{H},
\label{Max_en_coen_rel}
\eeq
\djc{where $\epsilon(z,t)$ and $\mu(z,t)$ denote the electric permittivity and the magnetic permeability, respectively.}\\
{\em Hamiltonian formulation:} The Hamiltonian $H$ can be written as 
\beq
H(\mathcal{D},\mathcal{B})&=&\int_Z \frac{1}{2}\left(\mathcal{E}\wedge \mathcal{D}+ \mathcal{H}\wedge \mathcal{B}\right).\label{H_maxwell}
\eeq
Therefore, $\delta_{\mathcal{D}}H=\mathcal{E}$ and $\delta_{\mathcal{B}}H=\mathcal{H}$.
Taking into account dissipation in the system, the dynamics can be written in the port-Hamiltonian form as
\beq
\hspace*{5mm}-\frac{\partial }{\partial t}\begin{bmatrix}
	\mathcal{D}\\    \mathcal{B}
\end{bmatrix}
&=&\begin{bmatrix}
	0& -\mathrm{d}\\ \mathrm{d}& 0
\end{bmatrix}\begin{bmatrix}
\delta_{\mathcal{D}}H\\ \delta_{\mathcal{B}}H
\end{bmatrix}\!+\!\begin{bmatrix}
J_d\\0
\end{bmatrix}=\begin{bmatrix}
	\ast \sigma & -\mathrm{d}\\ \mathrm{d}& 0
\end{bmatrix}\begin{bmatrix}
\delta_{\mathcal{D}}H\\ \delta_{\mathcal{B}}H
\end{bmatrix}
\label{Maxwell eqn_dirac}
\eeq
where $\ast J_d= \sigma \mathcal{E}$, $J_d$ denotes the current density and $\sigma(z,t)$ is the specific conductivity of the material.
In addition, we define the boundary variables as $f_b = \delta_D H|_{\partial Z}$ and $e_b = \delta_B H|_{\partial Z}$. Hence, we obtain $\frac{d}{dt} H \le \int_{\partial Z} \mathcal H \wedge \mathcal E$.
For $n=3$, $p=q=2$, and $\alpha_p=\mathcal{D}$, $\alpha_q=\mathcal{B}$ with $H$ given in \eqref{H_maxwell}, Maxwell\djc{'s} equations given in \eqref{Maxwell eqn_dirac} forms a \djc{Stokes-Dirac} structure. \\
{\em The Brayton-Moser form of Maxwell's equations:} \djc{In order to write the Maxwell's equations in BM form, we proceed with defining the corresponding mixed-potential functional}
\begin{equation} 
P=\int_{\djc{Z}} \mathcal{H}\wedge \mathrm{d}\mathcal{E}-\dfrac{1}{2}\sigma \mathcal{E} \wedge \ast \mathcal{E}, \label{P_maxwell}
\end{equation} 
which \djc{yields} the following BM form 
\begin{align}
\begin{bmatrix}
-\mu I_3 & 0\\ 0 & \epsilon I_3
\end{bmatrix}\begin{bmatrix}
\mathcal{H}_t\\ \mathcal{E}_t
\end{bmatrix}=\begin{bmatrix}
\ast \mathrm{d} \mathcal{E}\\ -\sigma \mathcal{E}+\ast \mathrm{d} \mathcal{H}\end{bmatrix} 
=\begin{bmatrix}
\ast \delta_{\mathcal{H}}P\\ \ast \delta_{\mathcal{E}} P\end{bmatrix}. 
\label{grad_maxwell}
\end{align}
Next, we present the stability analysis of Maxwell's equations \eqref{grad_maxwell}, using the admissible pairs provided in Section \ref{infi_admissible_pairs}.
\begin{proposition}
The system of equations \eqref{grad_maxwell} constitute alternate Brayton Moser formulation  $\tilde{A}\dot{x}=\ast\delta_u\tilde{P}$, where $\tilde{P}$ is as defined in \eqref{mix_pot_adm_gen} and  $\tilde{A}$ is defined as
 \begin{align}\label{Maxwell_Atilde}
\tilde{A}= \begin{bmatrix}
-\mu \lambda I & \epsilon M_2 \ast  \mathrm{d}  \\
-\mu    M_1 \ast \mathrm{d} &  \epsilon \left(\lambda I-\sigma M_2\right)
\end{bmatrix}.
\end{align} Additionally, \eqref{grad_maxwell} is stable 
if $\lambda$, $M_1>0$, and $M_2>0$ are selected such that $\epsilon M_2 =\mu M_1$ and  $0 < \lambda < \sigma \norm{M_2}_s$.
\end{proposition}
\begin{proof}
	The first part of the proof is straight forward from Theorem \ref{genPA}. The positive definiteness of $\tilde{P}$ can be seen by rewriting it as
	\begin{align*} 
	\tilde{P}=& \int_z    \delta_{\mathcal{E}}P \wedge  \frac{\sigma M_2  -\lambda I}{2 \sigma} \ast \delta_{\mathcal{E}}P +\frac{1}{2 \sigma}d\mathcal{H} \wedge \ast d \mathcal{H} + \frac{1}{2} \left(\delta_{\mathcal{H}}P \wedge M_1\ast \delta_{\mathcal{H}}P\right)
	\geq  0. 
	\end{align*}
	Under the condition $0 < \lambda < \sigma \norm{M_2}_s$, the \djc{time-derivative} of $\tilde P$ is
	$$
	\dot{\tilde{P}} =-\int_Z \left(\mu \lambda \mathcal{H}_t \wedge \ast \mathcal{H}_t+\mathcal{E}_t \wedge \ast (\sigma M_2-\lambda I)\mathcal{E}_t\right) \leq 0.
	$$
	Denote $U=(\mathcal{H},\mathcal{E})$, $\Delta U =(\Delta \mathcal{H}, \Delta\mathcal{E} )$ and consider the norm
	\beq\label{maxwell_norm}
	\norm{\Delta U}^2 &=&\int_Z\left((\Delta \mathcal{E}-\ast \mathrm{d}\Delta \mathcal{H})\wedge \ast( \Delta \mathcal{E}-\ast \mathrm{d}\Delta \mathcal{H})+\mathrm{d}\Delta\mathcal{H}\wedge \ast \mathrm{d}\Delta\mathcal{H}\right.\\&&\left.+\mathrm{d}\Delta\mathcal{E}\wedge \ast \mathrm{d}\Delta\mathcal{E}\right).\nonumber
	\eeq
	One can easily show that the system of equations \eqref{grad_maxwell} are stable at equilibrium $U^\ast=(0,0)$ by invoking Theorem \ref{thm::stab}  with respect to the above defined norm \eqref{maxwell_norm}, for $\alpha=2$ and 
	\beqn
	\begin{matrix}
		\gamma_1 &=& \min\{\dfrac{1}{2\sigma},\lambda_1^{min},\lambda_2^{min}\}, &
		\gamma_2 &=& \max\{\dfrac{1}{2\sigma},\lambda_1^{max},\lambda_2^{max}\}	
	\end{matrix}
	\eeqn
	where $\lambda_1^{min},\lambda_1^{max}$ are the minimum and maximum eigen values of $\frac{\sigma M_2  -\lambda I}{2 \sigma}$ respectively and similarly $\lambda_2^{min},\lambda_2^{max}$ for $\frac{1}{2}M_1$.
\end{proof}
\end{example}
\section{Systems with boundary interaction: Example of a transmission line system}\label{sec:TLS}
Boundary control of infinite dimensional systems is a well-studied topic. A significant advance in the port-Hamiltonian setting was presented in \citep{stokes}, where the authors extended the classical Hamiltonian formulation of infinite dimensional systems to incorporate boundary energy flow. Most infinite dimensional systems interact with the environment through its boundary, and hence such a formulation has an immediate impact on boundary control of infinite dimensional systems by energy shaping \citep{RamThe}.
In this section, we present the Brayton Moser formulation of infinite dimensional port Hamiltonian systems that interact through boundary. 
We derive admissible pairs and present a new passivity property for the transmission line system described in Example \ref{moti::example_TL}.
\subsection{The Brayton Moser form:}
{\em Spatial domain dynamics:}
The dynamics of the transmission line \eqref{telegraphers} can be written in an equivalent  Brayton Moser form as follows: define a functional $P^a=\int_0^1\text{P}^adz$ where 
\begin{flalign}
\text{P}^a=-\dfrac{1}{2}Ri\wedge \ast i+\dfrac{1}{2}Gv\wedge \ast v-i\wedge \mathrm{d}v =\left(-\frac{1}{2}Ri^2+\frac{1}{2}Gv^2-iv_z\right)dz.
\label{P_bm}
\end{flalign}
\kcn{In order to \ijc{simplify the notation}, we avoid using the differential geometric notation\footnote{\kcn{Note that  the Transmision line system \eqref{telegraphers} can be written in infinite dimensional port Hamiltonian formulation \eqref{pH_2} with $n=p=q=1$, this give rise to real valued ($0-forms$) coenergy variable $i(z,t)$ and $v(z,t)$, which are just functions.}}.} Using the line voltage and current as the state variables, we can rewrite the dynamics of the spatial domain as follows
\beq
\begin{bmatrix}
	-L & 0\\0 & C
\end{bmatrix}\begin{bmatrix}
i_t\\
v_t
\end{bmatrix}
&=& \begin{bmatrix}
	\delta_i P^a\\
	\delta_v P^a
\end{bmatrix} =
\begin{bmatrix}
	-Ri -v_z\\	
	Gv+i_z
\end{bmatrix}.
\label{tele_BM}
\eeq
The above equations, with  $A \text{ diag}~ \{ -L,C\} $, and
$u = ( i(z,t)~ v(z,t) )^{\top}$, can be written in a gradient form
\beq
Au_t=\delta_uP^a. \label{gradform}
\eeq
{\em Boundary dynamics}:
The spatial domain of the transmission line system is represented by a 1-D manifold $Z=(0,1)\in \mathbb{R}$ with point boundaries $\partial Z=\{0,1\}$.
In order to incorporate boundary conditions, we consider the interconnection of the infinite-dimensional system with finite-dimensional systems, via each of the boundary ports. This type of interconnected system is usually referred to as a {\em mixed finite and infinite-dimensional} system. Next, we aim to represent the overall system in BM formulation given in equation \eqref{BM_boundary}. 
Consider now a mixed potential function of the form
\beq
\mathcal{P}(U) &=& P^a(u)+P^0(u_0)+P^1(u_1) \label{P_bm bd}
\eeq
where $U = [u ~ u_0 ~ u_1]^{\top}$, $P^0$ and $P^1$ are the contributions to the mixed potential function arising form the boundary dynamics at $z=0$ and $z=1$ respectively. Similar to \eqref{BM_boundary}, we represent the overall dynamics of mixed finite and infinite-dimensional system in Brayton Moser form. The dynamics evolving on the spatial domain (for $0<z<1$) are given by \eqref{tele_BM} (equivalently \eqref{gradform}). At  $z=0$ the dynamics are
\beq\label{bndry z0}
A_0u_{0t}=\left.\left(\dfrac{\partial P^0}{\partial u_0}-\text{P}^a_{u_z}\right)\right|_{z=0} +B_0 I_0
\eeq
where \[\begin{matrix}
u_0=[i_0,v_0,v_{C0}]^\top, & P^0(u_0)=(v_{C0}-v_0)i_0-\frac{1}{2}R_0i_0^2, &
A^0= \text{ diag} ~ \{ 0,0,-C_0\},
\end{matrix}\]  
with $B_0=[0,0,-1]^\top$ as the input matrix, $I_0$ as input, $\text{P}^a_{u_z}=\dfrac{\partial \text{P}^a}{\partial u_z}$ and $u_{0t}=\dfrac{du_0}{dt}$. 
\\ The dynamics at boundary $z=1$ are 
\beq
A_1u_{1t}=\left.\left(\dfrac{\partial P^1}{\partial u_1}+\text{P}^a_{u_z}\right)\right|_{z=1}  \label{bndry z1}
\eeq
where $u_1=[i_1,v_1,v_{C1}]^\top$, $ P^1= (v_1-v_{C1})i_1-\frac{1}{2}R_1i_1^2$ and $A^1=\text{diag} \{ 0,0,-C_1\}.
$
Together they can be written compactly in the Brayton Moser form as
\beq
\mathcal{A}U_t=\delta_U\mathcal{P}+BI_0 \label{bm_total}
\eeq
$\mathcal{A}=diag\{A,A_0,A_1\}$, $B = \begin{bmatrix} 0 & 0 & B_0^\top &  \boldsymbol{O}_{3} \end{bmatrix}^{\top}$ and  $\boldsymbol{O}_3=[0\;0\;0]$.
\beqn
\delta_U \mathcal{P} = \begin{bmatrix}
	\delta_uP & \left.\left(\dfrac{\partial P^0}{\partial u_0}-\text{P}_{u_z}\right)\right|_{z=0} & \left.\left(\dfrac{\partial P^1}{\partial u_1}+\text{P}_{u_z}\right)\right|_{z=1}
\end{bmatrix}^\top .
\eeqn
\begin{remark}
	Note that the mixed potential functional is not unique. Another choice is $P^b=\int_0^1 \text{P}^bdz$ where
	\beq
	\text{P}^b=-\frac{1}{2}Ri^2+\frac{1}{2}Gv^2+i_zv\label{Pb}.
	\eeq
	This choice of $P^a$ or $P^b$ does not have any effect on spatial domain since it preserve the dynamics \eqref{tele_BM} and \eqref{gradform}, as $\delta_uP^a=\delta_uP^b$.
	If we use $P^b$ as mixed potential function instead of $P^a$, then we need to change  $P^0$ and $P^1$  to $v_{C_0}i_0-\frac{1}{2}R_0i_0^2$ and $-\frac{1}{2}R_1i_1^2-v_{C_1}i_1$ respectively in \eqref{bndry z0}, \eqref{bndry z1}.
\end{remark}

\subsection{ Dirac formulation} The transmission line system in  Brayton-Moser equations \eqref{bm_total} can be equivalently written as
	\beqn
	\left ( \left(-u_t, -u_{0t}, -u_{1t}\right),~B_0^\top u_{0t},~\left( \delta_u \text{P},~P^0_{u_0}-\text{P}_{u_z}|_{z=0}, P^1_{u_1}+\text{P}_{u_z}|_{z=1}\right),~-I_0\right ) \in \mathcal{D}
	\eeqn
	with subspace $D$ defined in Section \ref{subsec::dirac formulation}.
	This gives us the ``power balance equation"  
	\beqn 
	\dfrac{d}{dt}\mathcal{P} &=& \int^1_0 \left(Au_t \cdot u_t\right)dz +A_0u_{0t}\cdot u_{0t}+A_1u_{1t}\cdot u_{1t}+f^\top_{u0} y_0\label{td P+bd}\\
	&=& \int_0^1 \left(u_t^\top\dfrac{A+A^\top}{2}u_t\right)dz+ u_{0t}^\top\dfrac{A_0+A_0^\top}{2}u_{0t}\nonumber 
	+u_{1t}\dfrac{A_1+A_1^\top}{2}u_{1t}+f^\top y \label{td P+bd}\\
	&=&\int_0^1\left(-Li^2_t+Cv^2_t\right)dz-C_0\left(\dfrac{dv_{C0}}{dt}\right)^2-C_1\left(\dfrac{dv_{C1}}{dt}\right)^2+I_0\dfrac{dv_{C0}}{dt}.\label{time der of net P}
	\eeqn
\subsection{Admissible pairs}\label{section:Admissible pairs and stability}
To find admissible pairs for transmission line system with non zero boundary conditions, we need to define $\tilde A$ as the following, which will be clear in the subsequent section.
In general, new $\tilde A$ may contain  $\partial /\partial z$ in its entries (similar to $\ast \mathrm{d}$ in \eqref{genricA} and Remark \ref{genPA}). In this case there will be an additional contribution to the terms in the boundary from $\tilde A$, which will be clear in Proposition \ref{propo:Admissible}. To account this contribution we split such $\tilde A$ as $\tilde{A}_{nd}+\tilde{A}_d\frac{\partial}{\partial z}$.
\begin{definition}{\bf Admissible Pairs.}\label{Admissble_cond}
	Denote $\tilde{\mathcal{P}}=\int_Z\tilde{\text{P}^a}+\tilde{P}^0+\tilde{P}^1$ and $\tilde{\mathcal{A}}=\text{diag}~\{\tilde{A}, \tilde{A}_0, \tilde{A}_1\}$, further $\tilde{A}$ is $\tilde{A}_{nd}+\tilde{A}_d\frac{\partial}{\partial z}$. We call $\tilde{\mathcal{P}}$ and $\tilde{\mathcal{A}}$ admissible pairs if they satisfy the following:
	\begin{itemize}
		\item[(a)] ~\;\;\;\;\; $\tilde{P}^a\geq 0$, $\tilde A_d^\top=\tilde A_d$ 
		and $u_t^\top\tilde{A}_{nd}u_t\leq 0$ such that 
		\beq \tilde{A}u_t=\delta_u\tilde{\text{P}^a} \label{def1a}\eeq
		\item[(b)]  ~\;\;\;\;\;$\tilde{P}^0\geq 0$ and $u_{0t}^\top\tilde{A}_0u_{0t}\leq 0$ such that 
		\beq \left(\tilde{A}_0+\dfrac{1}{2}\tilde{A}_d\right)u_{0t}=\left.\left(\dfrac{\partial \tilde{P}}{\partial u_0}-\tilde{\text{P}^a}_{u_z}\right)\right|_{z=0}+\tilde B_0I_0 \label{def1b} \eeq
		\item[(c)]  ~\;\;\;\;\; $\tilde{P}^1\geq 0$ and $u_{1t}^\top\tilde{A}_1u_{1t}\leq 0$ such that 
		\beq \left(\tilde{A}_1-\dfrac{1}{2}\tilde{A}_d\right)u_{1t}=\left.\left(\dfrac{\partial \tilde{P}}{\partial u_1}+\tilde{\text{P}}_{u_z}^a\right)\right|_{z=1}  \label{def1c}\eeq
		\item[(d)] Together we can write them as 
		\beq
		\tilde{\mathcal{A}}U_t = \begin{matrix}
			\delta_U \tilde{\mathcal{P}}+ \tilde BI_0, & y_0 = - \tilde B^\top_0 u_{0t}.
		\end{matrix} \label{def1d}
		\eeq
	\end{itemize}
\end{definition}
We now have the following Proposition.
\begin{proposition}\label{propo:Admissible}
	If $\tilde{\mathcal{P}}=\int_Z\tilde{\text{P}^a}+\tilde{P}^0+\tilde{P}^1$ and $\tilde{\mathcal{A}}= \text{diag}~\{\tilde{A}, \tilde{A}_0, \tilde{A}_1\}$ satisfy the Definition \ref{Admissble_cond} then $\dot{\tilde{\mathcal{P}}}\leq I_0^\top y_0$, that is the system is passive with  port variables $I_0$ and $y_0$.
\end{proposition}
\begin{proof}
	The time derivative of $\tilde{P}_d\geq 0$ along the trajectories of (\ref{def1a}-\ref{def1c}) is
\beqn
	\dot{\tilde{P}}&=& \int_0^1\left(\delta_u\tilde{P^a}.u_t\right)dz+\left.\left(\dfrac{\partial \tilde{P}}{\partial u_0}-\tilde{\text{P}^a}_{u_z}\right)\right|_{z=0}\cdot u_{0t}+\left.\left(\dfrac{\partial \tilde{P}}{\partial u_1}+\tilde{\text{P}^a}_{u_z}\right)\right|_{z=1}\cdot u_{1t}\\
	&=& \int_0^1\left(\tilde{A}u_t.u_t\right)dz+u_{01}^\top\left(\tilde{A}_0+\dfrac{1}{2}\tilde{A}_d\right) u_{0t}+u_{1t}^\top\left(\tilde{A}_1-\dfrac{1}{2}\tilde{A}_d\right) u_{1t}+I_0^\top y_0\\
	&=& \int_0^1u_t^\top\left(\tilde{A}_{nd}+\tilde{A}_d\dfrac{\partial}{\partial z}\right)u_t dz +I_0^\top y_0+u_{01}^\top\left(\tilde{A}_0+\dfrac{1}{2}\tilde{A}_d\right) u_{0t}\\&&+u_{1t}^\top\left(\tilde{A}_1-\dfrac{1}{2}\tilde{A}_d\right)u_{1t}\\
	&=& \int_0^1\left(u_t^\top\tilde{A}_{nd}u_t\right)dz+\int_0^1\left(u_t^\top\tilde{A}_d\dfrac{\partial}{\partial z}u_t\right)dz+I_0^\top y_0+u_{01}^\top\left(\tilde{A}_0+\dfrac{1}{2}\tilde{A}_d\right) u_{0t}\\&&+u_{1t}^\top\left(\tilde{A}_1-\dfrac{1}{2}\tilde{A}_d\right)u_{1t}\\
	&=& \int_0^1\left(u_t^\top\tilde{A}_{nd}u_t\right)dz+\dfrac{1}{2}\int_0^1\dfrac{\partial}{\partial z}\left(u_t^\top\tilde{A}_du_t\right)dz+I_0^\top y_0+u_{0t}^\top\left(\tilde{A}_0+\dfrac{1}{2}\tilde{A}_d\right) u_{0t}\\&&+u_{1t}^\top\left(\tilde{A}_1-\dfrac{1}{2}\tilde{A}_d\right)u_{1t}\\
	&=& \int_0^1\left(u_t^\top\tilde{A}_{nd}u_t\right)dz+u_{01}^\top\tilde{A}_0 u_{0t}+u_{1t}^\top\tilde{A}_1 u_{1t}+I_0^\top y_0\\
	&\leq&  I_0^\top y_0.
\eeqn
\end{proof}
{\em Admissible pairs for the spatial domain:} First we derive admissible pairs for spatial domain of the transmission line, that is we find $(\tilde{P}^a,\tilde{A})$ satisfying Definition \eqref{def1a}. Next, we find suitable $(\tilde{P}^0,\tilde{A}_0)$ and $(\tilde{P}^1,\tilde{A}_1)$ satisfying \eqref{def1b} and \eqref{def1c} respectively so that we achieve the passivity as stated in Proposition \ref{propo:Admissible}.\\
We construct a new mixed potential $\tilde P$ (for spatial domain) in a similar procedure as followed in \citep{BraMos1964II}
\beq
\tilde{P}^a &=& \lambda P^a +\frac{1}{2} \int_0^1 \delta_u {P}^{a\top} M \delta_u {P}^a dz.\label{Tilde_P}
\eeq
We choose
$M=\begin{bmatrix}
\frac{\alpha}{R} & m_2\\
m_2 & \frac{\beta}{G}
\end{bmatrix}$  
where $\alpha , \beta,m_2$ are positive constants satisfying $\alpha\frac{L}{R}=\beta \frac{C}{G}$ and $\lambda$ is a dimensionless constant. Such a choice will be clear in the following discussions, which will eventually lead to a stability criterion. Note that $\tilde{P}$ still have units of power. To simplify the calculations we define new positive constants $\theta$, $\gamma$ and $\zeta$ as follows:
\beq
\begin{matrix}
	\theta\deff \alpha\frac{L}{R}=\beta \frac{C}{G},& m_2\deff \frac{2\gamma}{CR+LG}, &
	\zeta \deff  \frac{2\gamma}{\sqrt{LC}(\alpha+\beta)} \implies & m_2= \frac{\zeta \theta}{\sqrt{LC}}
\end{matrix}. \label{var}
\eeq 	
To show that $\tilde{P}^a\geq 0$ we start by simplifying the right hand side of \eqref{Tilde_P} in the following way.
Define 
\beq
\begin{matrix}
	\Delta \deff \left( \zeta \sqrt{\frac{C}{2}}(Ri+v_z)- {\sqrt{\frac{L}{2}}}(Gv+i_z) \right).
\end{matrix} \label{delta12}
\eeq
Using \eqref{var},\eqref{delta12}, and after some calculations, we can show that
\begin{eqnarray*}
	\frac{1}{2}\left<\delta_u P,M \delta_u P\right>	&=& \Delta^2+\frac{\alpha}{2R}(1-\zeta^2)(Ri+v_z)^2.
\end{eqnarray*} 
With $P^a$ as the mixed potential functional for transmission line, we calculate $\tilde{P}^a$ using \eqref{Tilde_P} as follows
\beqn \tilde{\text{P}}^a&=&\lambda \text{P}^a+\Delta^2+\frac{\alpha}{2R}(1-\zeta^2)(Ri+v_z)^2\\
&=& \lambda\left(-\frac{1}{2R}\left[\left(Ri+v_z\right)^2-v_z^2\right]+\frac{1}{2}Gv^2\right)+\Delta_2^2+\frac{\alpha}{2R}(1-\zeta^2)(Ri+v_z)^2\\
&=&\frac{\alpha(1-\zeta^2)-\lambda}{2R}(Ri+v_z)^2+\Delta^2+\frac{\lambda}{2R}v_z^2+\frac{\lambda}{2}Gv^2.
\eeqn
This means $\tilde{P}^a = \int_0^1 \text{P}^a dz\ge 0$,  for 
%
\begin{align}
0\leq \lambda\leq \alpha(1-\zeta^2), ~~
0 \leq \zeta^2 \leq 1.
\label{imp_cond}
\end{align}
Further, if we choose $\tilde A$ as
\begin{eqnarray}
\tilde{A}=\begin{bmatrix}
L(\lambda -\alpha- m_2 \frac{\partial }{\partial z} )& C(Rm_2+\frac{\beta}{G}\frac{\partial }{\partial z})\\
L(Gm_2+\frac{\alpha}{R}\frac{\partial }{\partial z})& -C(\lambda+\beta+ m_2\frac{\partial }{\partial z})
\end{bmatrix} \label{Anew}
\end{eqnarray}
then, this  $\tilde A$  together with $\tilde{P}^a$  will satisfy the gradient form \eqref{def1a}. Next we can decompose $\tilde{A}=\tilde{A}_{nd}+A_d\frac{\partial }{\partial z}$ with 
\beq\label{A_nd and A_d}
\tilde{A}_{nd}=\begin{bmatrix}
	L(\lambda -\alpha )& CRm_2\\
	LGm_2& -C(\lambda+\beta)
\end{bmatrix},\tilde{A}_d=\begin{bmatrix}
- m_2L  &\beta\frac{C}{G}\\
\alpha\frac{L}{R}& - m_2C
\end{bmatrix}
\eeq
and $\tilde{A}_{nd}$ is negative semi definite as long as
\beq 
-\beta \leq \lambda \leq \alpha, \text {and }(\lambda -\alpha)(\lambda+\beta)+\dfrac{(\alpha+\beta)^2}{4}\zeta^2 \leq 0\label{lambda}\eeq 
and noting that $\alpha\frac{L}{R}=\beta \frac{C}{G}$ from \eqref{var}, we can show that $\tilde{A}_d$ is symmetric.
\begin{proposition}\label{prop bm}
	If there exist  non zero constants $\alpha, \beta, \lambda$ and $ \zeta$ satisfying \eqref{var}, \eqref{imp_cond}, and \eqref{lambda} then $(\tilde{P^a}, \tilde A)$  is an admissible pair for the transmission line. The transmission line system with zero boundary energy flow is thus stable. 
\end{proposition}
\begin{proof}
	From \eqref{var} we define $\tau \deff \dfrac{\alpha}{\beta}=\dfrac{RC}{LG}$. Given a transmission line, $R,C,L$ and $G$ are fixed. $\tau \geq 0$ is now related to system parameters and thus can be treated as one. Let $\lambda^{'}=\dfrac{\lambda}{\beta}$. Using this in \eqref{imp_cond} and \eqref{lambda} we get
	\beq 0\leq &\lambda^{'}&\leq \tau(1-\zeta^2)\label{qqq1}\\
	(\lambda^{'} -\tau)(\lambda^{'}+1)&+&\frac{(\tau+1)^2}{4}\zeta^2 \leq 0. \label{qqq2}
	\eeq
	Now we have to show that for all $\tau \geq 0$, there exists a pair of $\lambda^{'}$ and $\zeta$ that satisfies equation \eqref{qqq1} and \eqref{qqq2}. Given a $\zeta \in (0,1)$, we obtain $\lambda^{'}\in[0,\tau(1-\zeta^2)]$ (using equation \eqref{qqq1}). Showing that \eqref{qqq2} has one positive and one negative root concludes the proof. Using the fact that a quadratic equation with roots $r_1$ and $r_2$ have opposite signs iff $r_1r_2\leq 0$, equation \eqref{qqq2} leads to
	\beqn
	\frac{(\tau+1)^2}{4}\zeta^2-\tau \leq 0 \Rightarrow \zeta^2 \leq \dfrac{4\tau}{(1+\tau)^2}.
	\eeqn
	Note that this is a valid condition on $\zeta$ since $\forall\; \tau \geq 0$, $\dfrac{4\tau}{(1+\tau)^2}\leq 1$.  
	Therefore  $\forall \;\zeta \in [0,\frac{4\tau}{(1+\tau)^2}]$ there exists a $\lambda^{'}$ which satisfies \eqref{qqq1} and \eqref{qqq2}. Consequently,  $(\tilde{P^a}, \tilde A)$ satisfies the admissible pair's Definition \ref{Admissble_cond}a. This implies stability of transmission line system with zero boundary conditions \citep{BraMir1964}.	
\end{proof}

{\em Admissible pairs for boundary dynamics: }
Assume that $m_2$ and $\theta$ satisfy $m_2=\frac{C_1R_1^2}{L}=\frac{C_1}{C}$ and $\theta=C_1R_1=C_0R_0$. Next we show that $(\tilde{P}^a, \tilde{A})$, together with 
\beq \label{admissible::boundary}
&&\begin{matrix}
	\tilde{P}^0 = \frac{1}{2R_0}(v_0-v_{C_0})^2 & \tilde{P}^1 = \frac{1}{2R_1}(v_1-v_{C_1})^2
\end{matrix}\nonumber \\
&&\tilde{A}_0= \begin{bmatrix}
	-(m_2L+R_0^2C_0) & 0 & R_0C_0 \\0& -(C_0+m_2C)& C_0 \\-R_0C_0 & -C_0 &0
\end{bmatrix},\tilde{A}_1=\begin{bmatrix}
0 & 0 & -C_1R_1 \\0& 0& C_1 \\C_1R_1 & -C_1 &0
\end{bmatrix}\nonumber
\eeq
satisfy Definition \ref{Admissble_cond}.
Now considering the left hand side of \eqref{def1c} with $\lambda=1$
\beqn
\left.\left(\dfrac{\partial \tilde{P^1}}{\partial u_1}+\tilde{\text{P}}^a_{u_z}\right)\right|_{z=1}
&=&\begin{bmatrix}
	m_2Li_{1t} -\theta v_{1t}\\
	m_2Cv_{1t}-\theta i_{1t}\\
	-i_1
\end{bmatrix}=\begin{bmatrix}
m_2Li_{1t} -\theta v_{1t}\\
m_2Cv_{1t}-\theta i_{1t}\\
-C_1v_{C_1t}
\end{bmatrix}
=\begin{bmatrix}
	-C_1R_1v_{C_1t}\\
	C_1v_{C_1t}\\
	-C_1v_t+C_1R_1i_t
\end{bmatrix}\\&=&\begin{bmatrix}
0 & 0 & -C_1R_1 \\0& 0& C_1 \\C_1R_1 & -C_1 &0
\end{bmatrix}\begin{bmatrix}
i_{1t}\\v_{1t}\\v_{C_1t}
\end{bmatrix}.
\eeqn
We can see that $\tilde{A}^1$ is skew symmetric. Similarly we can show that $\tilde{P}^0$ and $\tilde{A}^0$ preserves boundary and satisfies \eqref{def1b}. 
\subsection{Passivity}
\begin{proposition}
	Transmission line system defined by (\ref{telegraphers}-\ref{z1})  is passive with storage function   $\tilde{P}=\tilde{P}^a+\tilde{P}^0+\tilde{P^1}$ and port variables $I_0$ and  $\frac{dv_{C_0}}{dt}$.
\end{proposition}
\begin{proof}
	From definition \eqref{Admissble_cond} the time derivative of $\tilde{P}$ along the trajectories of (\ref{telegraphers}-\ref{z1}) gives
	\beq \label{passivity}
	\dot {\tilde{P}} &\leq&  I_0\frac{dv_{C_0}}{dt}
	\eeq
	which concludes the proof.
\end{proof}
\section{Casimirs and conservation laws}\label{seccas}
We obtain conservation laws which are independent from the mixed potential function, as follows:
For simplicity, we consider the case of systems without dissipation. We further assume that the energy and the
co-energy variables are related via a linear relation, given by
\beq
\alpha_p = \ast \epsilon \;  e_p \; \text{and}\;  \alpha_q = \ast \mu\; e_ q
\eeq
we can write \eqref{proof_Mix_1} in the following way
\beq
\begin{bmatrix}
	-\mu & 0\\ 0 & \epsilon
\end{bmatrix}\begin{bmatrix}
\dot{e}_q\\\dot{e}_p
\end{bmatrix}&=& \begin{bmatrix}
\ast \delta_{e_q}P\\ \ast \delta_{e_p}P
\end{bmatrix}.
\label{eq_cas1}
\eeq
Consider a function $ C  : \Omega^{n-p}(Z) \times \Omega^{n-q}(Z) \times Z \rightarrow \mathbb R$, which satisfies
\beq
\mathrm d (\ast \delta_{e_p}  C) = 0, ~~ \mathrm d (\ast \delta_{e_q}  C) = 0.
\eeq
The time derivative of $C(e_p,e_q)=\int_Z\text{C}(e_p,e_q)$ along the trajectories of \eqref{eq_cas1} is
\beqn
&&\dfrac{d}{dt} C(e_q,e_p) = \int_Z \left(\delta_{e_q}C \wedge \dot e_q+\delta_{e_p}C \wedge \dot e_p\right) \\
&=& \int_Z \left(-\delta_{e_q}C \wedge  \ast \dfrac{1}{\mu}\mathrm{d}e_p(-1)^{(n-q)\times q}+\delta_{e_p}C \wedge \ast \dfrac{1}{\epsilon}(-1)^{pq}\mathrm{d}e_q(-1)^{(n-p)\times p}\right) \\
&=& \int_Z \left((-1)^{(n-q).q+1}\dfrac{1}{\mu}\mathrm{d}e_p \wedge  \ast \delta_{e_q}C+(-1)^p \dfrac{1}{\epsilon}\mathrm{d}e_q \wedge \ast \delta_{e_p}C\right) \\
&=& \int_Z \left((-1)^{(n-q). q +1}\dfrac{1}{\mu}[\mathrm{d}(e_p \wedge  \ast \delta_{e_q}C)+(-1)^qe_p \wedge  \mathrm{d}(\ast \delta_{e_q}C)]\right.\\&&\left.+(-1)^p \dfrac{1}{\epsilon}[\mathrm{d}(e_q \wedge \ast \delta_{e_p}C)+(-1)^pe_p \wedge \mathrm{d}(\ast \delta_{e_p}C)]\right)\\
& = & \int_{\partial Z} \left (e_q \wedge \ast \delta_{e_p}C) \mid_{\partial Z}+(e_p \wedge  \ast \delta_{e_q}C) \mid_{\partial Z} \right ).
\eeqn
This implies that $\dot{ C}$ is function of boundary elements, representing a conservation law.\\ Additionally, if $\ast \delta_{e_p}C = \ast \delta_{e_q}C =0 $, then $d C/ dt = 0$. $ C$ is then called a Casimir function.
\subsection{Example: Transmission Line}
In case of the lossless transmission line, the total current and voltage
\begin{align}
\begin{matrix}
 C_I = \int_0^1 i(t,z) dz \hspace{3mm}&  C_v = \int_0^1 v(t,z) dz
\end{matrix}
\end{align}
are the systems conservation laws.
This can easily be inferred by the following
\begin{align*}
\frac{d}{dt} C_I &= -\int_0^1 \frac{1}{L}\frac{\partial v}{\partial z} = \left.\frac{v}{L} \right|_0 -  \left.\frac{v}{L}\right|_1\\
\frac{d}{dt}  C_v &= -\int_0^1 \frac{1}{C}\frac{\partial i}{\partial z} = \left.\frac{i}{C}\right|_0 -  \left.\frac{i}{C}\right|_1.
\end{align*}
\subsection*{Lossy Transmission line($R\neq 0,G\neq0$):}
Consider a functional $ C=\int_0^1 \bar{\text{C}}(i,v)dz$, where $\bar{\text{C}}(i,v)$ satisfies
\beq\label{A}
\begin{matrix}
	\frac{R}{L}\delta_i  C&=&\frac{1}{C}\frac{\partial }{\partial z}\delta_v  C, &\;\;\;\ & &
	\frac{G}{C}\delta_v  C&=&\frac{1}{L}\frac{\partial }{\partial z}\delta_i  C
\end{matrix}
\eeq
such as:
\beqn
 C(i,v)&=& \int_0^1 \left(\dfrac{\sqrt{G}}{C} \cosh(\omega z)i+\dfrac{\sqrt{R}}{L} \cosh(\omega z)v\right) dz
\eeqn
where $\omega =\sqrt{RG}$. It can be shown that the above functional satisfying \eqref{A} is a conservation law for lossy transmission line system $(R\neq0, G\neq 0)$ by evaluating the time derivative of $ C$ , that is
\beqn
\dfrac{d}{dt} C(i,v) &=& -\left.\left(\delta_iC v+\delta_vC i \right)\right|^1_0.
\eeqn

\subsection{Example: Maxwell's equations}
In case of Maxwell's equations with no dissipation terms, it can easily be checked that that the magnetic field
intensity $\int_Z \mathcal H$ and the electric field intensity $\int_Z \mathcal B$ constitute the conserved quantities. This can be seen via the following expressions:
\begin{align*}
& \int_Z\frac{d}{dt} \mathcal H = -\int_{\partial Z} \frac{1}{\mu} \mathcal E \\
& \int_Z\frac{d}{dt} \mathcal E = \int_{\partial Z} \frac{1}{\epsilon} \mathcal H.
\end{align*}
Another class of conserved quantities can be identified in the following way: Using \eqref{proof_Mix}, the system of equations can be rewritten as (when $R=0,G=0$)
\beq
\begin{bmatrix}
	-\mu & 0\\ 0 & \epsilon
\end{bmatrix}\begin{bmatrix}
\dot{e}_q\\  \dot{e}_p
\end{bmatrix}=
\begin{bmatrix}
	\ast \mathrm{d}e_p(-1)^{(n-q)\times q} \\
	\ast (-1)^{pq}\mathrm{d}e_q(-1)^{(n-p)\times p}
\end{bmatrix}.
\label{eq_cas3}
\eeq
Note that
\begin{align*}
\mathrm d \left ( {\mu} \ast \dot e_q\right ) & = \mathrm d (\ast \ast\mathrm{d}e_p)(-1)^{(n-q)\times q}  = 0\\
\mathrm d \left ({\mu} \ast \dot e_p\right ) & = \mathrm d (\ast \ast \mathrm{d}e_q)(-1)^{(n-p)\times p+pq} = 0.
\end{align*}
This means that  $\mathrm d (\mu \ast  e_q), ~ \mathrm d (\epsilon \ast  e_p)$ are differential forms which do not vary with time.\\
In terms of {\em Maxwells Equations} this would mean $\mathrm d (\mu \ast \mathcal H )$  is a constant three-form representing the charge density and $\mathrm d ( \epsilon \ast \mathcal E )$ is actually zero. In standard electromagnetic texts these would mean $\nabla \cdot \mathcal  D = J$, and $\nabla \cdot \mathcal B = 0$, representing respectively the Gauss' electric and magnetic law.
\section{Boundary control of transmission line system}\label{sec:control}
In this section we consider the stabilization problem of transmission line system in Example \ref{moti::example_TL} at a nontrivial equilibrium point via boundary control. The control objective is to regulate the voltage at the capacitor $C_1$ to $v_{C1}^\ast$ using the current source $I_0$ connected at $z=0$. 
We use the new passivity property \eqref{passivity} derived in Proposition \ref{prop bm}, that is
\beq\label{control_passivity}
\dfrac{d}{dt}\tilde{\mathcal{P}} \leq I_0\dfrac{dv_{C0}}{dt}
\eeq
in achieving the boundary control objective.\\
{\em Boundary control:}
The argument used here is same as that presented in \citep{mtns,Hugo}, where
the authors have presented a boundary control law for a
mixed finite and infinite-dimensional system via energy
shaping methods. But in this case, the passive maps $I_0$ and $v_{C_0}$ (obtained using energy as storage function) do not work due to dissipation obstacle as shown in Proposition \ref{prop::dissp}.  Therefore we propose a boundary
control law via shaping the power of the infinite-dimensional
system. Towards achieving this, we adopt control by interconnection methodology using the passivity property \eqref{passivity}.
As in the finite dimensional case, the method relies
on finding Casimir functions for the closed-loop system \citep{pasumarthy2007achievable}. Consider the controller of the form:
\beq \label{controller_states}\begin{matrix}
	\dot \eta=u_c,\hspace{5mm} & y_c=\dfrac{\partial H_c(\eta)}{\partial \eta}
\end{matrix} \eeq
where $\eta$, $u_c$ and $y_c$ are respectively the state, input and output of the controller. $H_c(\eta)$ denotes the power function of the controller. The interconnection between the system and controller is given by
\beq\label{interconnection}
\begin{bmatrix}
	I_0\\u_c
\end{bmatrix}&=& \begin{bmatrix}
0 & 1\\-1 &0
\end{bmatrix}\begin{bmatrix}
\dfrac{dv_{C0}}{dt}\\ y_c
\end{bmatrix}.
\eeq
%
%
{\em Casimirs:}
It can be easily shown that functions $C(\eta,v_{C_0})=\eta + v_{C_0}$ is Casimir for the closed loop system. Time differential of $C(\eta,v_{C_0})$ is (along \eqref{controller_states} and \eqref{interconnection})
\beqn
\dot C &=& \dot \eta+\dfrac{dv_{C_0}}{dt}= 0
\eeqn
Now the plant state and controller state are related by $\eta =-v_{C_0}+c$, $c$ is a constant(we can take it to be zero if the initial condition of the plant is known). Using this we choose the Hamiltonian of the controller to be
\beqn
H_c(\eta)&=&-v_{C_0}i_0^\ast+\frac{1}{2}K_I(v_{C_0}-v_{C_0}^\ast)^2
  \eeqn
where $K_I\geq 0$ is tuning parameter. We further modify this in the following way (such modification will be useful in power shaping)
\beqn
H_c(\eta)&=& -v_{C_0}i_0^\ast\pm v_0i_0^\ast\pm v_1i_1^\ast+\frac{1}{2}K_I(v_{C_0}-v_{C_0}^\ast)^2\\
&=& (v_0-v_{C_0})i_0^\ast +\left(v_1i_1^\ast -v_0i_0^\ast\right) -v_1i_1^\ast +\frac{1}{2}K_I(v_{C_0}-v_{C_0}^\ast)^2\\
&=& -i_0i_0^\ast R_0+\int_0^1 \dfrac{\partial }{\partial z}\left(vi^\ast \right)dz-v_1i_1^\ast +\frac{1}{2}K_I(i_0-i_0^\ast)^2\\
&=&  -i_0i_0^\ast R_0+\int_0^1 \left(v_zi^\ast+vi^\ast_z \right)dz-v_1i_1^\ast+\frac{1}{2}K_I(v_{C_0}-v_{C_0}^\ast)^2. 
\eeqn
Using this controller Hamiltonian, we will shape the closed-loop mixed mixed potential functional. Let $c\in \mathbb{R}$ be a constant. Consider
\beqn
&&P_d = \tilde{P}^a+\tilde{P}^0+\tilde{P}^1+H_c(\eta)+c \\
&=&\hspace{-3mm} \int_0^1\left(\frac{\alpha(1-\zeta^2)-1}{2R}(Ri+v_z)^2+\Delta_2^2+\frac{1}{2R}v_z^2+\frac{1}{2}Gv^2\right)dz+\dfrac{1}{2R_0}(v_0-v_{C_0})^2\\&&+\dfrac{1}{2R_1}(v_1-v_{C_1})^2-i_0i_0^\ast R_0+\int_0^1 \left(v_zi^\ast+vi^\ast_z \right)dz+\frac{1}{2}K_I(i_0-i_0^\ast)^2+c\\
&=& \int_0^1\left(\frac{\alpha(1-\zeta^2)-1}{2R}(Ri+v_z)^2+\Delta^2+\frac{1}{2R}v_z^2+v_zi^\ast+\frac{1}{2}Gv^2+vi^\ast_z\right)dz\\
&&+\dfrac{1}{2}R_0i_0^2+\dfrac{1}{2}R_1i_1^2-i_0i_0^\ast R_0+c\pm\int_0^1\left(Ri^{\ast 2}+\dfrac{i_z^{\ast 2}}{2G}\right)dz\pm\dfrac{1}{2}R_0i_0^{\ast 2} 
\\&&+\frac{1}{2}K_I(v_{C_0}-v_{C_0}^\ast)^2\\
&=& \int_0^1\left(\frac{\alpha(1-\zeta^2)-1}{2R}(Ri+v_z)^2+\Delta^2+\dfrac{1}{2R}\left(v_z^2+2v_zRi^\ast +R^2i^{\ast 2}\right)\right.\\&&\left.+\frac{1}{2G}\left(Gv+i_z^{\ast }\right)^2\right)dz+\dfrac{1}{2}R_0\left(i_0^2-2i_0i_0^\ast +i_0^{\ast 2}\right)+\dfrac{1}{2}R_1i_1^2-i_0i_0^\ast R_0+c\\&&-\int_0^1\left(Ri^{\ast 2}+\dfrac{i_z^{\ast 2}}{2G}\right)dz-\dfrac{1}{2}R_0i_0^{\ast 2}+\frac{1}{2}K_I(v_{C_0}-v_{C_0}^\ast)^2 .\\
&=& \hspace{-3mm}\int_0^1\left(\frac{\alpha(1-\zeta^2)-1}{2R}(Ri+v_z)^2+\Delta^2+\dfrac{1}{2R}\left(v_z+Ri^\ast \right)^2+\frac{1}{2G}\left(Gv+i_z^\ast \right)^2\right)dz\\&&+\frac{1}{2}R_0\left(i_0-i_0^\ast \right)^2+\dfrac{1}{2}R_1i_1^2. 
\eeqn
By choosing $c=\int_0^1\left(Ri^{\ast 2}+\dfrac{i_z^{\ast 2}}{2G}\right)dz+\dfrac{1}{2}R_0i_0^{\ast 2}$, we can see that 
\beq\label{PD}
P_d\hspace{-3mm}&=&\hspace{-3mm} \int_0^1\left(\frac{\alpha(1-\zeta^2)-1}{2R}(Ri+v_z)^2+\Delta^2+\dfrac{1}{2R}\left(v_z+Ri^\ast \right)^2+\frac{1}{2G}\left(Gv+i_z^\ast \right)^2\right)dz\nonumber\\&&+\dfrac{1}{2}R_0\left(i_0-i_0^\ast \right)^2+\dfrac{1}{2}R_1i_1^2+\frac{1}{2}K_I(v_{C_0}-v_{C_0}^\ast)^2
\eeq
has a minimum at equilibrium of \eqref{eq_telegraphers}, \eqref{eq_z0} and \eqref{eq_z1}. 
The time derivative of $P_d$ along (\ref{telegraphers}-\ref{z1}), \eqref{controller_states} and \eqref{interconnection} is
\beq \label{Pddot}
\dfrac{d}{dt}P_d&\leq&\left(I_s+K_I(v_{C_0}-v_{C_0}^\ast)-i_0^\ast\right)\frac{dv_{C_0}}{dt}
\eeq
\subsection*{Stability analysis}
Denote $\Delta U=(\Delta i, \Delta v,\Delta i_0,\Delta v_0,\Delta v_{C_0},\Delta i_1,$ $\Delta v_1,\Delta v_{C_1})$. Consider the following norm
\beq\label{norm_trans}
\norm{\Delta U}^2&=&\int_0^1\left((R\Delta i+\Delta v_z)^2+\Delta v_z^2+\Delta v^2\right)dz+\Delta i_0^2+\Delta i_1^2 +\Delta v_{C_0}^2
\eeq
\begin{proposition}\label{prop::TL_control_stabl}
	The transmission line system \eqref{bm_total} in closed-loop with control
	\beq\label{PI}
	I_0&=& i_0^\ast-K_P\frac{dv_{C_0}}{dt}-K_I(v_{C_0}-v_{C_0}^\ast),\;\;\; K_P,K_I \geq 0
	\eeq
	is asymptotically stable at the operating point $U^\ast=(i^\ast, v^\ast,i_0^\ast,v_0^\ast,v_{C_0}^\ast,i_1^\ast,v_1^\ast,v_{C_1}^\ast)$ as defined in \eqref{eq_telegraphers}, \eqref{eq_z0} and \eqref{eq_z1}. 
\end{proposition}
\begin{proof}
	From \eqref{PD}, we can show that
	\beq \label{stab_proof_1}
	\begin{matrix}
		P_d(U)>0, & \hspace{-3mm}\forall \;U\neq U^\ast, P_d(U)=0 \hspace{-2mm} & \text{if and onlyif}\hspace{-2mm}&U=U^\ast, &\hspace{-3mm} \text{and} & \delta_UP_d(U^\ast)=0 .
	\end{matrix}
	\eeq
	Moreover, $\mathcal{N}(\Delta U)=P_d(U^\ast+\Delta U)-P_d(U^\ast)$
\begin{flalign*}
&= \int_0^1\left(\frac{\alpha(1-\zeta^2)-1}{2R}(R\Delta i+\Delta v_z)^2+\dfrac{1}{2R}\Delta v_z^2+\frac{1}{2}G\Delta v^2\right)dz\\&\hspace{0.5cm}+\dfrac{1}{2}R_0\Delta i_0^2+\dfrac{1}{2}R_1\Delta i_1^2+\frac{1}{2}K_I\Delta v_{C_0}^2.&&
\end{flalign*}
For
\begin{flalign*}
\begin{matrix}
    \gamma_1&=&\min\left\{\frac{\alpha(1-\zeta^2)-1}{2R},\hspace{3mm}\frac{1}{2R},\hspace{3mm}\frac{1}{2}G,\hspace{3mm}\frac{1}{2}R_0,\hspace{3mm}\frac{1}{2}R_1,\hspace{3mm}\frac{1}{2}K_I\right\},\\
    \gamma_2&=&\max\left\{\frac{\alpha(1-\zeta^2)-1}{2R},\hspace{3mm}\frac{1}{2R},\hspace{3mm}\frac{1}{2}G,\hspace{3mm}\frac{1}{2}R_0,\hspace{3mm}\frac{1}{2}R_1,\hspace{3mm}\frac{1}{2}K_I\right\},
\end{matrix}
\end{flalign*}
we have the following 
	\beq \label{stab_proof_2}
	\gamma_1\norm{\Delta U}^2\leq P_d(U^\ast+\Delta U)-P_d(U^\ast)\leq \gamma_2\norm{\Delta U}^2.
	\eeq
	Finally, using \eqref{Pddot} and \eqref{PI} the time derivative $\dot{P}_d$ is
	\beq \label{stab_proof_3}
	\dot{P}_d&=& \int_0^1\left(u_t^\top\tilde{A}_{nd}u_t\right)dz+u_{01}^\top\tilde{A}_0 u_{0t}+u_{1t}^\top\tilde{A}_1 u_{1t}+ \left(I_0+K_I(v_{C_0}-v_{C_0}^\ast)-i_0^\ast\right)\frac{dv_{C_0}}{dt}\nonumber\\
	&=& -K\left(\int_0^1\left(i_t^2+v_t^2\right)dz+i_{0t}^2+v_{C_0t}^2\right)\leq 0 .
	\eeq
	\kcn{Arnold's first stability theorem (Theorem \ref{thm::stab}) can be proved using \eqref{stab_proof_1}, \eqref{stab_proof_2} and \eqref{stab_proof_3}.} Hence, the transmission line system \eqref{bm_total} in closed-loop is Lyapunov stable at $U^\ast$ with respect to the norm $\norm{\cdot}$ defined in \eqref{norm_trans} . Further from \eqref{eq_telegraphers}, \eqref{eq_z0} and \eqref{eq_z1}, one can show that $\dot{P_d}=0$ iff $U=U^\ast$. Thereby, we conclude the proof by invoking LaSalle$\rq$s invariance principle \citep[see Theorem 5.19]{RamThe}.
\end{proof}
\section{Alternative passive maps} \label{sec::Infinite dimensional port Hamiltonian systems}
   %
   In Chapter 2.4.4, we have seen that admissible pairs for finite-dimensional systems lead to new constraints on system parameters. In the case of infinite-dimensional systems, the interconnected boundary elements give additional constraints on the system, which makes the problem of finding admissible pairs even more difficult. As we argued in finite dimensional case, that these restrictions are mainly due to imposing `gradient structure'. Motivated by this, we now extend the framework developed in Chapter 3.2 to infinite-dimensional systems. 
   We present our result for a general infinite dimensional port-Hamiltonian system defined using  Stokes Dirac structure \cite{stokes}.
%
%
As an example, we present Maxwell's equations in $\mathbb{R}^3$ with non-zero boundary energy flows.
Consider the following storage function
\beq\label{infinite_dim_storage_co_energy}
S(e_p,e_q)&=& \frac{1}{2}\int_Z\left(\dot{e}_q\wedge \ast \mu \dot{e}_q+\dot{e}_p\wedge \ast \epsilon \dot{e}_p\right).
\eeq
\begin{proposition}
	The infinite dimensional port Hamiltonian system defined in \eqref{pH_2} is passive with storage function \eqref{infinite_dim_storage_co_energy} and port variables $\dot f_b$ and $-\dot e_b$.
\end{proposition}
\begin{proof}
	We can now rewrite the spatial dynamics of the infinite-dimensional port-Hamiltonian system, in terms of the co-energy variables as \cite{IFAC}
	\beq\label{infinite_dim_eqn_co_energy}
	\begin{matrix}
		-\ast\epsilon \dot{e}_p &=&  \ast Ge_p + (-1)^r\mathrm{d}e_q\\
		-\ast \mu \dot{e}_q&=& \mathrm{d}e_p + \ast Re_q	.
	\end{matrix}
	\eeq
Using \eqref{infinite_dim_eqn_co_energy}, the time derivative of the storage function $S(e_p,e_q)$ \eqref{infinite_dim_storage_co_energy} can be simplified as follows
\beqn
\dot{S}(e_p,e_q)&=& \frac{1}{2}\int_Z\left(\dot{e}_q\wedge \ast \mu \ddot{e}_q+\dot{e}_p\wedge \ast \epsilon \ddot{e}_p\right)\\
&=&-\int_Z\left(\dot{e}_q\wedge \left(\mathrm{d}\dot e_p + \ast R\dot e_q\right)\right.\left.+\dot{e}_p\wedge\left(\ast G\dot e_p + (-1)^r\mathrm{d}\dot e_q\right)\right)\\
&=& -\int_Z\left(\dot{e}_q\wedge \ast R \dot e_q+\dot e_p\wedge \ast G \dot e_p\right)+(-1)^{(n-q) \times  q}\int_Z\left(\mathrm{d}(\dot e_p\wedge \dot e_q)\right)\\
&\leq & (-1)^{(n-q)}\int_{\partial Z}\left.(\dot e_q\wedge \dot e_p)\right.\\
&=& -\int_{\partial Z}\dot e_b\wedge \dot f_b.
\eeqn
This implies that the system \eqref{pH_2} is passive with port variables $\dot f_b$ and $-\dot e_b$.
\end{proof}
Next using Maxwell's equations as examples we will derive new passive maps for systems non zero boundary conditions. 
\subsection{Example: Maxwell's equations}\label{section:Maxwell}
\begin{example}[Maxwell's equations]
	Consider Maxwell's equations presented in Example \ref{Example::Maxwell} with boundary interaction, defined by 
%
\begin{align}
\begin{matrix}
-\mu \mathcal{H}_t &=& \ast d \mathcal{E}\\
\epsilon\mathcal{E}_t&=&-\sigma \mathcal{E}+\ast d \mathcal{H}\\
0&=&-(\mathcal{H}+\ast \sigma_d \mathcal{E})|_{\partial Z}+J^s
\end{matrix}
\label{grad_maxwell1}
\end{align}
where $\sigma_d$ is specific conductance at boundary and $J^s$ is the source connected at boundary.
\end{example}
\begin{proposition}
	Consider the following functional
	\beq\label{Max_storage}
	S(\mathcal{H}_t,\mathcal{E}_t)=\dfrac{1}{2}\int_Z\left(\mathcal{H}_t\wedge \ast \mu \mathcal{H}_t+\mathcal{E}_t\wedge \ast \epsilon  \mathcal{E}_t\right)
	\eeq
	The system of equations \eqref{grad_maxwell1} are passive with storage function \eqref{Max_storage} and port variable $\dfrac{d}{dt}J^s$ and $\mathcal{E}_t|_{\partial Z}$.
\end{proposition}
\begin{proof}
	The time derivative of the $S(\mathcal{H}_t,\mathcal{E}_t)$ \eqref{Max_storage} along \eqref{grad_maxwell1} is
	\beqn
\dfrac{d}{dt} S(\mathcal{H}_t,\mathcal{E}_t)&=&-\int_Z\left(\mathcal{E}_t\wedge\ast  \sigma \mathcal{E}_t\right)+\int_{\partial Z}\left(\mathcal{H}_t\wedge \mathcal{E}_t\right)
	\eeqn
	But at boundary $\partial Z$ we have $(\mathcal{H}+\ast \sigma_d \mathcal{E})|_{\partial Z}=J^s$ using this we get
	\beqn
	\dfrac{d}{dt} S(\mathcal{H}_t,\mathcal{E}_t)&\leq&\int_{\partial Z}\left(\mathcal{H}_t\wedge \mathcal{E}_t\right)\\
	&=& \int_{\partial Z}\left(\left(\dfrac{d}{dt}J^s-\ast \sigma_d\mathcal{E}_t\right)\wedge \mathcal{E}_t\right)\\
	&\leq & \int_{\partial Z}\left(\dfrac{d}{dt}J^s\wedge \mathcal{E}_t\right),
	\eeqn	
	concludes the proof.
\end{proof}
\begin{ctrlobj}
The control objective is to stabilize the system at $\mathcal{E}|_{\partial Z}=\mathcal{E}^\ast$. At equilibrium we have 
\beqn
J^{s\ast}=\mathcal{H}^\ast+\ast \sigma_d\mathcal{E}^\ast.
\eeqn
\end{ctrlobj}
In Section 5.6, we have used energy-Casimir methodology for boundary control of transmission line system. Instead, we now use the control methodology presented in Chapter 3, by noting that the output port variable $\mathcal{E}_t$ is integrable with time. Consider the closed-loop storage function of the form
\beqn
S_d=S(\mathcal{H}_t,\mathcal{E}_t)+f(\mathcal{E})
\eeqn
where the functional $f(\mathcal{E}):\Omega^1(Z)\rightarrow \mathbb{R}$, is chosen such that $S_d$ has a minimum at the desired operating point. One possible choice is
\beqn
S_d=S(\mathcal{H}_t,\mathcal{E}_t)+\dfrac{K}{2}\int_{\partial Z}\left((\mathcal{E}-\mathcal{E}^\ast)\wedge (\mathcal{E}-\mathcal{E}^\ast)\right)
\eeqn
where $K\geq 0$. The time derivative of $S_d$ along the trajectories of \eqref{grad_maxwell1} is
\beqn
\dfrac{d}{dt}S_d&=&\dfrac{d}{dt}S(\mathcal{H}_t,\mathcal{E}_t)+K\int_{\partial Z}\left((\mathcal{E}-\mathcal{E}^\ast)\wedge \mathcal{E}_t\right)\\
&\leq &\int_{\partial Z}\left(\left(\dfrac{d}{dt}J^s+K(\mathcal{E}-\mathcal{E}^\ast)\right)\wedge \mathcal{E}_t\right)
\eeqn
Now choosing 
\beqn
J^s&=&-K\int_0^t(\mathcal{E}(z,\tau)-\mathcal{E}^\ast)d\tau-\alpha (\mathcal{E}(z,t)-\mathcal{E}^\ast)\\&&+\mathcal{H}^\ast+\ast \sigma_d\mathcal{E}^\ast
\eeqn
where $\alpha \geq 0$ and we get $\dot S_d\leq 0$ and $S_d$ has a minimum at $\mathcal{E}=\mathcal{E}^\ast$ and $\mathcal{E}_t=0$, $\mathcal{H}_t=0$ further at this equilibrium we have $J^{\ast}=\mathcal{H}^\ast+\ast \sigma_d\mathcal{E}^\ast$. Finally using a similar argument presented in Proposition \ref{prop::TL_control_stabl} we can conclude the stability.
\section{Conclusions}\label{sec:conclusion}
In this Chapter, we presented a methodology to overcome the dissipation obstacles in the case of infinite-dimensional systems, thus paving a way for passivity based control techniques. The basic building block was to write the system equations in the Brayton-Moser form. However, to effectively use the method, we need to construct admissible pairs for a given system, which aids in stability analysis and also in deriving new passivity properties. We presented a systematic way to derive these admissible pairs and prove the stability of Maxwell's equations. In the case of non-zero boundary energy flows, we used the transmission line system (as an example) and identified its admissible pair conditions. This resulted in a new passivity property with current and derivative of voltage as input and output port variables respectively, at the boundary. Using the new passive map, a PI controller was constructed to solve the boundary control problem.  Moreover, we extended the differential-passivity like passive maps presented in Chapter 3.2 to infinite dimensional systems.

\afterpage{\blankpage}
\chapter{Primal-dual dynamics of constrained optimization}
	    The applications of convex optimization are ubiquitous in various fields of research \cite{ben2001lectures} such as, resource allocation \cite{ibaraki1988resource}, utility maximization \cite{kelly1998rate}, etc. Numerous methods are proposed to solve these optimization problems \cite{boyd2004convex}. Solution techniques in a distributed setting have gained importance in recent times \cite{xiao2006optimal}.  One  of  the standard  tools  for  designing  algorithms  to  solve  such  optimization  problems  is  through  primal-dual  gradient  method \cite{kose1956solutions,arrow1958studies}. Gradient-based methods are a well-known class of mathematical routines for solving convex optimization problems. These gradient algorithms have much to gain from a control and dynamical systems perspective, to have a better understanding of the underlying system theoretic properties (such as stability, convergence rates, and robustness). The convergence of gradient-based methods and Lyapunov stability, relate the solution of the optimization problem to the equilibrium point of a dynamical system. 

        The Krasovskii-Lyapunov function \cite{krasovskiicertain} is  particularly suited for establishing stability of the continuous time gradient laws, as the equilibrium point (or solution of the optimization problem) is not known apriori.	In \cite{feijer2010stability}, the authors use  Krasovskii-Lyapunov function and hybrid Lasalle's invariance principle \cite{lygeros2003dynamical} to prove asymptotic stability of a network optimization problem. The gradient structure of the primal-dual equations characterizing the optima of a  convex optimization with only equality constraint admit a Brayton Moser (BM) form \cite{jeltsema2009multidomain}. Further, using the duality between energy and co-energy the BM form is partially transformed  into a port-Hamiltonian (pH) form \cite{stegink2017unifying}. These transformations pave the way for passivity/stability analysis using (i) the invariance principle for discontinuous Caratheodory systems \cite{cherukuri2016asymptotic} and (ii) an incremental passivity property for the misfit dynamics.
	    \textcolor{black}{Using the input/output dissipative properties \cite{l2gain}, the authors in \cite{simpson2016input}  provided robustness analysis  for  primal-dual dynamics of convex optimization problems with equality constraint. }
		\newpage
The contents in this chapter are organized as follows: We start by presenting the necessary and sufficient conditions for optimality of convex optimization problems. Next, we show that the primal-dual dynamics of optimization problem with only equality constraint can be written in Brayton-Moser formulation. Consequently, we use the passive maps derived in Chapter \ref{ch:con_fin_dim_sys}. The primal-dual dynamics of inequality constraints are modeled as a state dependent switching system. We first show that each switching mode is passive and the passivity of the system is preserved under arbitrary switching using hybrid passivity tools, a methodology similar to switched Lyapunov functions for stability analysis of switch system. Finally, the two systems, (i) one derived from the Brayton Moser formulation and (ii) the state dependent switching system, are interconnected in a manner such that the equilibrium is the solution of the convex  optimization problem. The proposed methodology is demonstrated by finding the optimal separating hyperplane using support vector machine methodology.
\section{Convex optimization}
In this section, we present a brief overview of mathematical tools in convex optimization, that will be useful in the subsequent sections. The standard form of a convex optimization problem contains three parts:

\begin{minipage}{0.05\linewidth}\vspace{-0.7cm}
(i) 
\end{minipage}
\begin{minipage}{0.95\linewidth}
A continuously differentiable convex function $f(x):\mathbb{R}^n\rightarrow \mathbb{R}$ to be minimized over $x$,
\end{minipage}

\begin{minipage}{0.05\linewidth}
(ii) 
\end{minipage}
\begin{minipage}{0.95\linewidth}
affine equality contraints $h_{i}(x)=0, \hspace{0.4cm} i = 1,\hdots , m$,
\end{minipage}

\begin{minipage}{0.05\linewidth}\vspace{-0.7cm}
(iii)
\end{minipage}
\begin{minipage}{0.95\linewidth}
continuously differentiable convex inequality constraints of the form $g_{i}(x)\leq 0$, \hspace{0.4cm} $i = 1,\hdots , p$.
\end{minipage}

\noindent This can be written in the following form, commonly known as the primal formulation:
    \begin{equation}\label{standard_SOP}
	\begin{aligned}
	& \underset{x \in \mathbb{R}^{n}}{\text{minimize}}
	& & f(x)\\
	& \text{subject to}
	& & h_{i}(x)=0 \hspace{0.4cm} i = 1,\hdots , m\\
    &
    & & g_{i}(x)\leq 0 \hspace{0.4cm} i = 1,\hdots , p\\
	\end{aligned}
	\end{equation}
\textbf{Karush-Kuhn-Tucker (KKT) conditions}: If the solution $x^\ast$ is optimal to the convex optimization problem \eqref{standard_SOP} then these exists $\lambda_i\in \mathbb{R}$, $i = 1,\hdots , m$ and $\mu_{i}\geq 0$, $i = 1,\hdots , p$ satisfying the following KKT conditions
\begin{eqnarray}\label{standard_KKT}
	\nabla_{x} f(x^{*})+\sum_{i=1}^{m}\lambda_{i}\nabla_{x} h_{i}(x^{*})+\sum_{i=1}^{m}\mu_{i}\nabla_{x} g_{i}(x^{*})=0,\nonumber\\
	h_{i}(x^{*})=0 \;\;\;\forall i \in \{1,\ldots, m\},\\
   g_j(x^\ast)\leq 0,\;\;\mu_j\geq 0,\;\; \mu_jg_j(x^\ast)=0\;\;\;\forall j \in \{1,\ldots, p\}.\nonumber
\end{eqnarray}
\begin{remark}
Note that the KKT conditions presented above in equation \eqref{standard_KKT} are only necessary conditions. We next present the requirements under which KKT conditions becomes sufficient.
\end{remark}
\noindent We now define the Lagrangian of the convex optimization \eqref{standard_SOP} as
 \begin{eqnarray}\label{convex_Lagrangian}
 \mathcal{L}(x,\lambda,\mu)=f(x)+\sum_{i=1}^{m}\lambda_{i} h_{i}(x)+\sum_{i=1}^{m}\mu_{i}g_{i}(x)
 \end{eqnarray}
 and the Lagrange dual function as
  \begin{eqnarray}\label{convex_Lagrange_dual}
 L_d(\lambda,\mu)=   
	\begin{aligned}
	& \underset{x \in \mathbb{R}^{n}}{\text{minimize}}
	& & L(x,\lambda,\mu) 
	\end{aligned}
 \end{eqnarray}
 \noindent giving us the following dual problem (correspnding to the primal problem \eqref{standard_SOP})
     \begin{equation}\label{standard_SOP_dual}
	\begin{aligned}
	& \underset{\lambda \in \mathbb{R}^{m},\;\mu\in \mathbb{R}^{p}}{\text{maximize}}
	& &  L_d(\lambda,\mu)\\
	& \text{subject to}
	& & \mu_i\geq 0 \hspace{0.4cm} i = 1,\hdots , p.\\
	\end{aligned}
	\end{equation}
    \begin{remark}
    Dual problem  is always convex, because $L_d$ is always a concave function even when the primal \eqref{standard_SOP} is not convex. If $f^\ast$ and $L_d^\ast$ denotes the optimal values of primal and dual problems respectively, then $L_d^\ast \leq f^\ast$. Therefore dual formulations are used to find the best lower bound of the optimization problem \cite{boyd2004convex,ben2001lectures}. Further, the negative number $L_d^\ast -f^\ast$ denotes the {\em duality gap}. In the case of zero duality gap, we say that the problem \eqref{standard_SOP} {\em satisfies strong duality}.
    \end{remark}
    \begin{definition}\textbf{Slater's conditions}. We call the convex optimization problem \eqref{standard_SOP} satisfies Slater's conditions if there exists an $x$ such that 
    $h_{i}(x)=0 \hspace{0.4cm} i = 1,\hdots , m$ and $ g_{i}(x)< 0 \hspace{0.4cm} i = 1,\hdots , p$. This implies that inequality constraints are strictly feasible.
    \end{definition}
    \begin{remark} If a convex optimization problems \eqref{standard_SOP} satisfies  Slater's conditions then the optimal values of primal and dual problems are equal, that is, \eqref{standard_SOP} satisfies strong duality. Further, in this case the KKT conditions becomes necessary and sufficient.
    \end{remark}
 %
 %
 %
 %
 %
	\section{The BM formulation: For equality constraint}
    In this section, we consider a convex optimization problem with only equality constraints. We show that the primal-dual gradient equations pertaining to this, have a naturally existing Brayton-Moser formulation. Thereafter, we leverage the analysis presented in Chapter 3 to find new passive maps associated with the primal-dual dynamics. Consider the following constrained optimization problem 
	\begin{equation}\label{SOP}
	\begin{aligned}
	& \underset{x \in \mathbb{R}^{n}}{\text{minimize}}
	& & f(x)\\
	& \text{subject to}
	& & h_{i}(x)=0 \hspace{0.4cm} i = 1,\hdots , m
	\end{aligned}
	\end{equation}
	where $f: \mathbb{R}^{n} \rightarrow \mathbb{R}$ is twice continuously differentiable $(C^2)$ and strictly convex and $h_{i} (\in C^2): \mathbb{R}^{n} \rightarrow \mathbb{R}$ is affine. 
 Assume (i) that the objective function has a positive definite Hessian $\nabla_{x}^{2}f(x)$ and (ii) that the problem (\ref{SOP}) has a finite optimum, and \kcb{Slater's condition is satisfied (i.e., the constraints are feasible) and strong duality holds} \cite{boyd2004convex}. The solution $x^{*}$ is an optimal solution to \eqref{SOP}
	if there exists $\lambda^{*} \in \mathbb{R}^{m}$ such that the following KKT conditions are satisfied:
	\begin{eqnarray}\label{mainKKT}
	\begin{aligned}
	\nabla_{x} f(x^{*})+\sum_{i=1}^{m}\lambda_{i}\nabla_{x} h_{i}(x^{*})=0\\
	h_{i}(x^{*})=0 \;\;\;\forall i \in \{1,\hdots, m\}
	\end{aligned}
	\end{eqnarray}
	The Lagrangian of \eqref{SOP} is given by
	\begin{equation}\label{MainLag}
	\mathcal{L}(x,\lambda)=f(x)+\sum_{i=1}^{m}\lambda_{i} h_{i}(x)
	\end{equation}
	Since strong duality holds for \eqref{SOP},
	$(x^{*},\lambda^{*})$ is a saddle point of the Lagrangian $\mathcal{L}$, that is,
    \begin{equation}
(x^\ast,\lambda^\ast)=arg\max_{\lambda}\left(arg\min_{x}\mathcal{L}(x,\lambda)\right),
\end{equation}
    if and only if $x^{*}$ is an optimal solution to primal problem \eqref{SOP} and $\lambda^{*}$ is optimal solution
	to its dual problem. Now, consider the dynamics given by
	\begin{equation}\label{maindyn1}
	\begin{split}
	-\tau_{x}\dot{x}&=\nabla_{x} \mathcal{L}(x,u)+u\\
	\vspace{-100mm}
	\tau_{\lambda_{i}}\dot{\lambda}_{i}&= \nabla_{\lambda}\mathcal{L}(x,u),\;\;~\hspace{1cm}
	y =-x,
	\end{split}
	\end{equation}
    or equivalently,
    \begin{equation}\label{maindyn}
	\begin{split}
	-\tau_{x}\dot{x}&=\nabla_{x} f(x)+\sum_{i=1}^{m}\lambda_{i}\nabla_{x} h_{i}(x)+u\\
	\vspace{-100mm}
	\tau_{\lambda_{i}}\dot{\lambda}_{i}&= h_{i}(x),\;\;~\hspace{2cm}
	y =-x,
	\end{split}
	\end{equation}
	where $\tau_{x}, \tau_{\lambda}\deff\text{diag}\{\tau_{\lambda_{i}},\ldots,\tau_{\lambda_{m}}\}$ are positive definite matrices and $u, y \in \mathbb{R}^n$. The unforced system of equations, obtained by setting $u=0$ in \eqref{maindyn}, represent primal-dual dynamics corresponding to \eqref{MainLag}. Moreover, the equilibrium point corresponds to the solution of the KKT conditions \eqref{mainKKT}.
\begin{remark}The discrete time primal-dual gradient descent equations of convex optimization problem \eqref{SOP} are 
\begin{eqnarray*}
x(t_{k+1})&=&x(t_k)-\eta_{x}\nabla_x\mathcal{L}(x,u)\\
\lambda(t_{k+1})&=&\lambda(t_k)+\eta_{\lambda}\nabla_{\lambda}\mathcal{L}(x,u),\;\;\; k\in \mathbb{Z}^+.
\end{eqnarray*}
where $\eta_x>0$ and $\eta_{\lambda}>0$ represents the step size. Further these are equivalent to the continuous time equations \eqref{maindyn1}, if the step sizes are chosen as $\eta_x=\Delta T\tau_x^{-1}$ and $\eta_{\lambda}=\Delta T\tau_{\lambda}^{-1}$, where $\Delta T=t_{k+1}-t_k$.
\end{remark}    
   Note that the primal-dual equations, expressed in \eqref{maindyn1}, closely resemble a pseudo-gradient structure. The pseudo-gradient structure in turn constitutes the basic skeleton of the Brayton-Moser itself. 
%
This key observation motivated us to look for connections between convex optimization and Brayton-Moser formulation. The following example better illustrates the idea we wish to outline here. Consider the unforced parallel RLC circuit with $V_s=0$ in Example \ref{Example::ch2::prlc::BM}. The equations governing the dynamics of the circuit, in BM formulation, is given by
      \begin{align*}
      \begin{matrix}
			-L\dot{i}&=& \nabla_{i}P\\
			C\dot{v}&=&\nabla_{v}P
	  \end{matrix}  
 		\end{align*}
Note that these equations can also be interpreted as the {\em continuous time gradient descent dynamics} for finding the quantity given by
\begin{equation*}
		   \max_{v}\left(\min_{i}P(i,v)\right),
\end{equation*}
where $P$ essentially acts as a Lagrangian. This example immediately points out that replacing $P$ with $\mathcal L$ would allow us to study the stability aspects of the primal-dual dynamics, under the same lens of the results derived in the earlier chapters. To do this, we start with BM formulation of convex optimization problem \eqref{SOP}.
        
	
	Let $z=(x,\lambda)$. The continuous time gradient laws \eqref{maindyn}, associated with \eqref{SOP}, naturally admit a Brayton-Moser (BM) formulation
	\begin{equation}
	Q(z)\dot{z}=\nabla_z P(z)+u\label{BM}
	\end{equation}
	with $Q(z)=\text{diag}\{-\tau_{x},\tau_{\lambda}\}$ and $P(z)=f(x)+\lambda^{\top} h(x)$ is a scalar function of the state. We next utilize this BM formulation to present a passivity property similar to the one derived in Proposition \ref{prop_control_finite}.
	\begin{proposition}\label{prop::eq_const}
		Let $\bar{z}=(\bar{x}, \bar{\lambda})$ satisfy \eqref{mainKKT}. Assume $h(x)$ is  convex and  $f(x)$ strictly convex. Then the system of equations \eqref{maindyn} are passive with port variables $(\dot{u},\dot{y})$ \cite{ICCvdotidot}.  Further every solution of the unforced version ($u=0$) of  \eqref{maindyn} asymptotically converges to $\bar{z}$.
	\end{proposition}
	\begin{proof}
		We can now consider storage function $\tilde{P}$ of the form \eqref{kras_lypunov} with  $M = diag\{\tau_{x},\tau_{\lambda}\}$, resulting in
		\begin{eqnarray}\label{eq_const_P}
		\tilde{P}&=&\frac{1}{2}\dot{z}^{T}M\dot{z} = \frac{1}{2}\dot{x}^{T}\tau_{x}\dot{x}+\frac{1}{2}\dot{\lambda}^{T}\tau_{\lambda}\dot{\lambda}
		\end{eqnarray}
		The time derivative of the storage function \eqref{eq_const_P} along the system of equations \eqref{maindyn} can be computed as
		\beqn
		\dot{\tilde{P}}&=&-\dot{x}^\top\nabla_x^2f(x) \dot{x}-\dot{x}^\top  \dot{u}
		\leq  -\dot{x}^\top  \dot{u}= \dot u ^{\top}\dot y
		\eeqn
		which implies that the system \eqref{maindyn} is passive. 
	 Further for $u=0$ we have $\dot{\tilde{P}}=0$ $\implies$ $\dot{x}=0$ ( $x$ is some constant). Using this in the first equation of \eqref{maindyn} we get that $\lambda$ is a constant, proving asymptotic stability of $\bar{z}$.
	\end{proof}
\begin{remark}The second method of Lyapunov method hinges on finding a suitable Lyapunov function that decreases along the system trajectories. Knowing the equilibrium point of a dynamical system is not necessary, but definitely helps in constructing the Lyapunov function. Krasovskii-type Lyapunov functions \cite{khalil1996noninear} \cite{krasovskiicertain} is one such function which does not require the information about the equilibrium point of the dynamical system explicitly. Proving stability using Krasovskii-type Lyapunov implies that the distance between the trajectories is decreasing, which is essentially the fundamental idea in contraction analysis. Utilizing Lyapunov analysis, we can show that trajectories converge to a limit set.  However the limit set is not known apriori. 
Hence the use of Krasovskii-type storage function \eqref{eq_const_P} in Proposition \ref{prop::eq_const} is particularly suited, as the solution of the optimization problem (or the equilibrium point) is not known a priori.
%
%
%
\end{remark}
	%
	%
\section{Switch system formulation: For inequality constraint}
We now define the inequality constraint $g_i(\tilde{u})\leq 0$  as the following hybrid dynamics
	\begin{equation}\label{IED}
	\tau_{\mu_i}\dot \mu_i=(g_i(\tilde u))^+_{\mu_i}
	\end{equation}
	where  $\tilde{u}\in \mathbb{R}^n$ and $i\in \{1\cdots p\}$. The positive projection of $g_i(\tilde{u})$ can be written as 
	\beq\label{IED1}
	(g_i(\tilde u))_{\mu_i}^+&=&\left\{\begin{matrix} 
		g_i(\tilde u) &\mu_i>0\\\max\{0,g_i(\tilde{u})\} & \mu_i=0
	\end{matrix} \right. 
	\eeq
This is introduced in \cite{kose1956solutions}, where the authors construct a dynamical system which converges to the stationary solution of saddle value problems. These equations are proposed in such a way that, if the initial condition of $\mu(t)$ is non-negative, then the trajectories $\mu(t)$ always stay inside positive orthant $\mathbb{R}^+$. 
Note that the discontinuity in the above equations \eqref{IED1} occurs when $g_i(\tilde{u})<0$ and $\mu_i=0$, the value of $(g_i(\tilde{u}))^+_{\mu_i}$ switches from $g_i(\tilde{u})$ to $0$. This ensures that the $\mu_i$'s does not go below zero. To make this more visible, we redefine these equations equivalently as follows; 
	\beq 
	(g_i(\tilde u))_{\mu_i}^+&=&\left\{\begin{matrix}
		g_i(\tilde u) &\; (\mu_i>0 \;\;\text{or}\;\; g_i(\tilde{u})>0)\\0 &\;\text{otherwise}
	\end{matrix} \right. 
	\eeq The projection is said to be active in the second case. Let $\mathcal{P}$ represent the power set of $\{1\cdots p\}$, then we define the function $\sigma:[0,\;\infty)\rightarrow \mathcal{P}$ as follows
	\beq\label{sigma_map} \sigma(t)=\{i \mid \; \mu_i(t)=0\;\;\text{and }\;g_i(\tilde{u})\leq0\, \forall i \in \{1, ..., p\}\}\eeq
	where the projection is active. With $\sigma(t)$ representing the switching signal, equation \eqref{IED} now takes the form of a switched system 
	\beq \label{active_const_def}
	\tau_{\mu_i}\dot \mu_i=g_i(\tilde u,\sigma)&=&\left\{\begin{matrix}
		g_i(\tilde u) ;&\; i\notin \sigma(t)\\0 ;&\;i\in \sigma(t)
	\end{matrix} \right. 
	\eeq
	The overall dynamics of the $p$ inequality constraints $g_i(\tilde{u})\leq 0$ $\forall i\in \{1\cdots p\}$ can be written in a compact form as:
	\begin{equation}\label{IEdyn}
	\tau_{\mu}\dot{\mu}=g(\tilde u,\sigma)
	\end{equation}
	where $\mu_i$ and $g_i(\tilde{u},\sigma)$ are $i^{th}$ components of  $\mu$ and $g(\tilde{u},\sigma)$ respectively.
	%
	%
	%
	It is well known that a sufficient condition for a switched system to be passive system is that the storage function should be common for all the individual subsystems \cite{zhao2006notion}. In general it is not easy to find such storage functions. Here we use passivity property defined with `switched storage functions'\cite{HybridPassive}. Consider the following storage function(s) 
	\beq\label{storage_fun_ineq_const}
	S_{\sigma_q}(\mu)&=&\dfrac{1}{2}\sum_{i\notin\sigma_q}^{}\dot \mu_i^2\tau_{\mu_i}\;\;\; \forall \sigma_q \in \mathcal{P}
	\eeq
	\begin{proposition} \label{prop:ineq_passivity}
		The switched system \eqref{IEdyn} is  passive with switched storage functions $S_{\sigma_q}$ (defined one for each switching state $\sigma_q\in\mathcal{P}$ ), input port $u_s=\dot{\tilde{u}}$ and  output port $y_s=\dot{\tilde{y}}$ where $\tilde{y}=\sum_{\forall i} \mu_i\nabla_{\tilde u}g_i(\tilde{u})$. That is, for each $\sigma_p\in \mathcal{P}$ with the property that for every pair of switching times $(t_i,t_j)$, $i<j$ such that $\sigma(t_i)=\sigma(t_j)=\sigma_p\in \mathcal{P}$ and $\sigma(t_k)\neq \sigma_p$ for $t_i<t_k<t_j$, we have
		\begin{eqnarray}\label{Hybrid_passivity}
		S_{\sigma_p}(\mu(t_j))-S_{\sigma_p}(\mu(t_i))\leq \int_{t_i}^{t_j}u_s^\top y_sdt.
		\end{eqnarray}
	\end{proposition}
	\begin{proof}	
We start with analyzing the passivity property for a time interval say $[0,\; \tau_{\sigma})$ with fixed $\sigma(t)$. The time derivative of the storage function $S_{\sigma}(\mu)$ is 
	\begin{eqnarray}
\label{Ssigmadot}
\begin{aligned}
\dot{S}_{\sigma}
&=\sum_{i\notin\sigma}^{}\dot \mu_i\ddot \mu_i\tau_{\mu_i}=\sum_{i\notin\sigma}^{}\dot \mu_i\nabla_{\tilde u}g_i^\top\dot{\tilde u}\\
&= \dot{\tilde{u}}^\top\left(\dfrac{d}{dt}\sum_{i\notin \sigma}\mu_i\nabla_{\tilde{u}}g_i -\sum_{i\notin \sigma}\mu_i\nabla_{\tilde{u}}^2g_i \dot{\tilde{u}}\right)\nonumber\\
&= \dot{\tilde{u}}^\top\left(\dot{\tilde{y}} -\sum_{\forall i}\mu_i\nabla_{\tilde{u}}^2g_i \dot{\tilde{u}}\right)\\
&\leq\dot{\tilde{u}}^\top \dot{\tilde{y}}=u_s^\top y_s.
\end{aligned}
\end{eqnarray}
	In step two we use $ \sum_{i\notin\sigma}^{}\mu_i\nabla_{u}g_i=\sum_{\forall i}^{}\mu_i\nabla_{u}g_i$ (which is true since $\mu_i=0$, if $ i\in \sigma $) and in step three we use the convexity of $g$ and non-negativity of the $\mu_i$. The above inequality can be equivalently written as
		\begin{eqnarray}\label{passive_IE}
		S_{\sigma}(\mu(\tau_{\sigma}))-S_{\sigma}(\mu(0))\leq \int_0^{\tau_\sigma}\dot{\tilde{u}}^\top \dot{\tilde{y}} dt
		\end{eqnarray}
	Hence, the system of equation \eqref{IEdyn} represent a finite family of passive systems and \eqref{storage_fun_ineq_const} represents their corresponding storage functions. 
	Note that in the above inequality, supply rate in the right hand side is independent of $\sigma$ (discrete state), where as the storage functions are dependent on $\sigma$. 
	Since this is not sufficient to prove the passivity property of \eqref{IEdyn}, we further need to analyse the behaviour of the storage functions at all switching times. Let $\sigma(t) \in \mathcal{P}$ denotes current active projection set as defined in \eqref{sigma_map}, then we have the following scenarios:\\
    \begin{minipage}{0.05\linewidth}
    \vspace{-5cm}
    (i)
    \end{minipage}
    \begin{minipage}{0.95\linewidth}
    \vspace{0.25cm}
    For some $i\notin \sigma(t^-)$, let the projection of $i^{th}$ constraint ($g_i(\tilde{u})\leq 0$) become active (i.e $ \mu_i$ reaches $0$ when $g_i(\tilde{u})<0$) at time $t$. This implies a new element $i$ is added to the projection set, $i\in\sigma(t)$. 
		The storage function \eqref{storage_fun_ineq_const} decreases by loosing the term $\tau_{\mu_i}\dot{\mu}_i^2$ from the summation. 
		The term in the storage function corresponding to this $i$ will not appear in \eqref{storage_fun_ineq_const} as $i\in \sigma(t)$. This happens discontinuously because $g_i(\tilde{u},\sigma)$ switches from $g_i(\tilde{u})< 0$ to $0$. Hence
		\beq\label{inactive2active_constraint} S_{\sigma(t)}(\mu(t))< S_{\sigma(t^-)}(\mu(t^-)) \eeq
    \end{minipage}
    
     \begin{minipage}{0.05\linewidth}
    \vspace{-4cm}
    (ii)
    \end{minipage}
    \begin{minipage}{0.95\linewidth}
    In the case when the projection of an active constraint $i \in \sigma(t^-)$ becomes inactive i.e $i \notin \sigma(t)$, a new term $\tau_{\mu_i}\dot{\mu}_i^2$ is added to the summation of the storage function \eqref{storage_fun_ineq_const}. But this happens in a continuous way because $g_i(\tilde{u},\sigma)$ has to increase from $g_i(\tilde{u})< 0$ to $g_i(\tilde{u})> 0$ by crossing $0$. By continuity argument we have
    \vspace{-0.5cm}
		\beq\label{active2inactive_constraint} S_{\sigma(t)}(\mu(t))= S_{\sigma(t^-)}(\mu(t^-)).\\\nonumber
        \eeq
    \end{minipage}
    		These situations are depicted in Fig. \ref{SFV}.
		\begin{figure}[t]
			\center
			\includegraphics[width=0.8\textwidth]{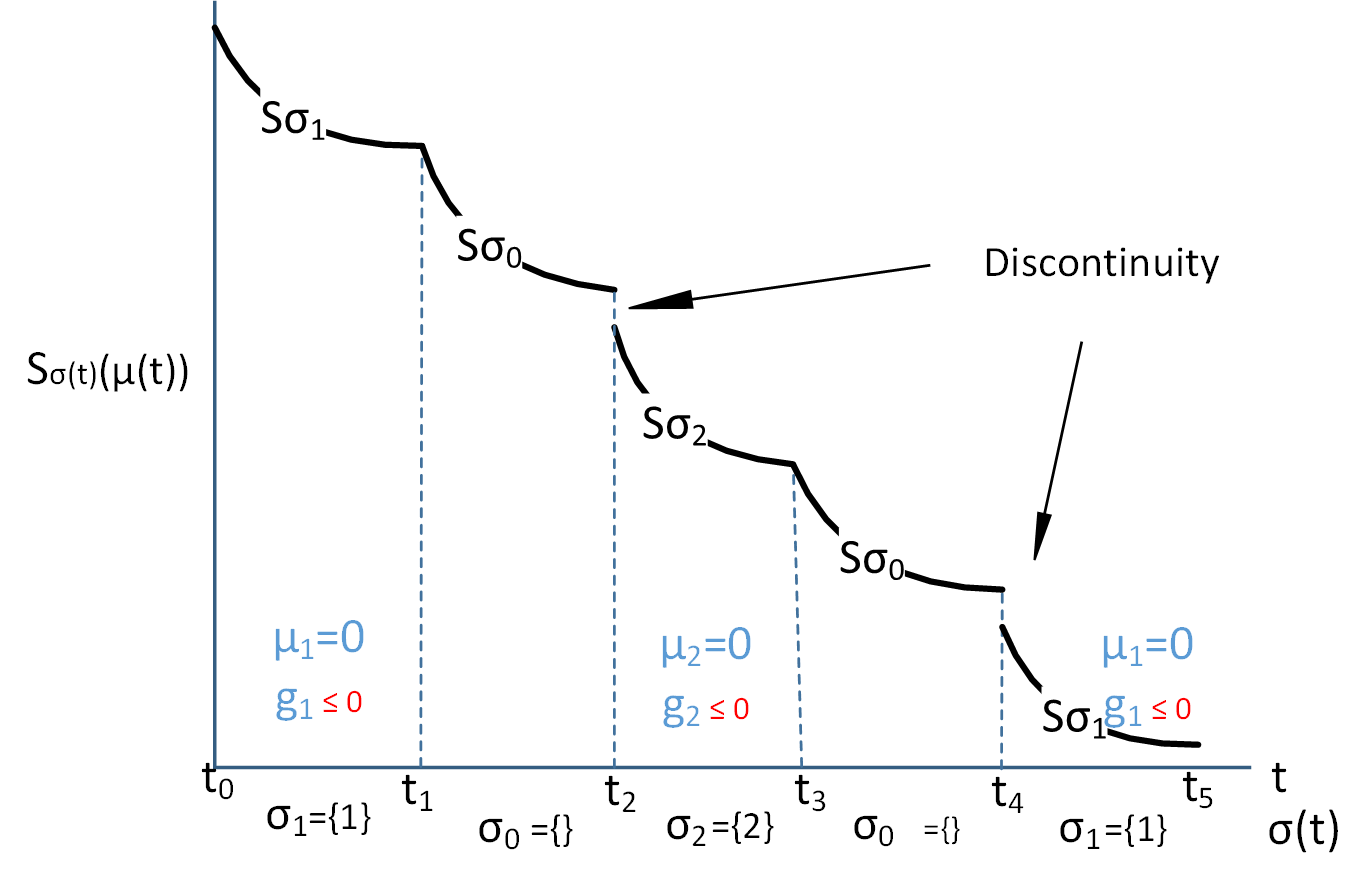}
			\caption{Example for time evolution of storage function with  two inequality constraints ($p=2$). Note that case (i) appears at switching time $t_2$, $t_4$ and case (ii) at $t_1$, $t_3$.}
			\label{SFV}
		\end{figure}
	Now consider a $\sigma_p\in \mathcal{P}$ as given in the Proposition 4.2.
	 We assume that there are $N$ switching times between $t_i$ and $t_j$. 
	 Noting that the storage function is not increasing at switching times we have, 
	  \begin{eqnarray*}
	 S_{\sigma(t_j)} &\leq& S_{\sigma(t_j^-)}
	 \leq  S_{\sigma(t_{i+N})} +\int_{t_{i+N}}^{t_j}\dot{\tilde{u}}^\top \dot{\tilde{y}} dt\\
	 &\leq & S_{\sigma(t_{i+N}^-)} +\int_{t_{i+N}}^{t_j}\dot{\tilde{u}}^\top \dot{\tilde{y}} dt\\
	  &\leq & S_{\sigma(t_{i+N-1})} +\int_{t_{i+N-1}}^{t_{i+N}}\dot{\tilde{u}}^\top \dot{\tilde{y}} dt +\int_{t_{i+N}}^{t_j}\dot{\tilde{u}}^\top \dot{\tilde{y}} dt\\
	    &\leq & S_{\sigma(t_{i})} +\int_{t_{i}}^{t_{i+1}}\dot{\tilde{u}}^\top \dot{\tilde{y}} dt +\cdots+\int_{t_{i+N}}^{t_j}\dot{\tilde{u}}^\top \dot{\tilde{y}} dt\\
	        &= & S_{\sigma(t_{i})} +\int_{t_{i}}^{t_j}\dot{\tilde{u}}^\top \dot{\tilde{y}} dt.
	  \end{eqnarray*}
	 Above we used \eqref{passive_IE}, \eqref{inactive2active_constraint} and \eqref{active2inactive_constraint}. 
	 We thus conclude the system is passive with port variables $(\dot{\tilde{u}}, \dot{\tilde{y}})$. 
	 \end{proof}
	\begin{proposition}\label{prop::ass_stab_ineq}
		The equilibrium set $\Omega_e$ defined by constant control input $\tilde{u}=\tilde{u}^\ast$ of \eqref{IED}
		\beqn\label{IE_equilibrium}
		\Omega_e=\left\{(\bar{\mu},\tilde{u}^\ast)\left|g_{i}(\tilde{u}^{*})\leq 0,\;\; \bar{\mu}_{i}g_{i}(\tilde{u}^{*})=0 \hspace{0.2cm} \forall i \in \{1,\hdots, p\}\right.\right\}
		\eeqn
		is asymptotically stable.
	\end{proposition}
	\begin{proof} 
	From \eqref{Ssigmadot}, \eqref{inactive2active_constraint} and \eqref{active2inactive_constraint} in Proposition \ref{prop:ineq_passivity}, we can infer that the Lyapunov function \eqref{storage_fun_ineq_const} is non-increasing for a constant $\tilde{u}=\tilde{u}^\ast$, concluding Lyapunov stability.  Now we use hybrid Lasalle's theorem condition \cite{lygeros2003dynamical} to show that $\Omega_e$ is the maximal positively invariant set, defined by \\
	   (i) $\dot{S}_{\sigma}(\mu(t))=0$ for fixed $\sigma$. This is can be verified by substituting $\tilde{u}=\tilde{u}^\ast$ a constant  in \eqref{Ssigmadot}.\\
	    (ii) $S_{\sigma(t^-)}(\mu(t^-))=S_{\sigma(t)}(\mu(t))$ if $\sigma$ switches between $\sigma(t^-)$ to $\sigma(t)$ at time $t$. In \eqref{IED}, if $g_i(\tilde{u}^\ast)<0$ and the corresponding $\mu_i^\ast>0$ then $\mu_i$ linearly converges to zero, causing a discontinuity in the Lyapunov function $S_{\sigma}(\mu(t))$ ( case-i of Proposition \ref{prop:ineq_passivity}). This does not happen if either
	     \beq\label{equild_cond12}
	        g_i(\tilde{u}^\ast)<0 ~\text{and} ~ \mu_i^\ast =0~\text{or}~
	         g_i(\tilde{u}^\ast)=0 ~\text{and}~  \mu_i^\ast\geq 0
	     \eeq
	     because both conditions imply $\dot{\mu}_i=0$. 
	We now prove that the trajectories of \eqref{IED} are bounded for $\tilde{u}=\tilde{u}^\ast$.
	Consider the quadratic norm $V(\mu)=\frac{1}{2}(\mu-\bar{\mu})^\top \tau_{\mu}(\mu-\bar{\mu})$. Next, using (8), (9) and (17) together with $g^+_i(\tilde{u})_{\mu_i}\leq g_i(\tilde{u})$, we show that the $V(\mu)$ is non-increasing 
	\begin{eqnarray*}
	\dot{V}
	&=&(\mu-\bar{\mu})^\top g^+(\tilde{u}^\ast)_{\mu}\\
	&\leq&(\mu-\bar{\mu})^\top g(\tilde{u}^\ast)\\
	&=&\sum_{\forall i\notin \sigma(t)}^{}(\mu_i-\bar{\mu}_i)^\top g_i(\tilde{u}^\ast)+\sum_{\forall i\in \sigma(t)}(\mu_i-\bar{\mu}_i)^\top g_i(\tilde{u}^\ast)\\
	&=&\sum_{\forall i\notin \sigma(t)}(\mu_i-\bar{\mu}_i)^\top g_i(\tilde{u}^\ast)\\
	&=&\sum_{\forall i\notin \sigma(t)}\mu_i^\top g_i(\tilde{u}^\ast)-\sum_{\forall i\notin \sigma(t)}\bar{\mu}_i^\top g_i(\tilde{u}^\ast)\\
	&=&\sum_{\forall i\notin \sigma(t)}\mu_i^\top g_i(\tilde{u}^\ast)\\&\leq& 0
	\end{eqnarray*}
	In step one we used \eqref{IED}, in step two we used the fact that $g^+_i(\tilde{u})\leq g_i(\tilde{u})$, in step three and four we used \eqref{active_const_def} , in step six we used \eqref{equild_cond12} and finally in step seven we again used \eqref{active_const_def}. 
	This implies that the trajectories of \eqref{IED} are bounded for $\tilde{u}=\tilde{u}^\ast$.
	If $g_i(\tilde{u}^\ast)>0$, 
	$\mu_i$ increases linearly, contradicting the boundedness of the trajectories. The proof follows by noting that conditions in \eqref{equild_cond12} represent $\Omega_e$ set.
	\end{proof} 
	The most interesting  property of passive systems is their modular nature. One can define power conserving interconnections (such as Newton law's or Kirchoff's current/voltage laws) between these systems, and show that the overall system is passive and there by stable. In the next section we make use of this property, to include inequality constraint \eqref{IED} in the optimization problem \eqref{SOP}.
	\section{The overall optimization problem}
    \label{chap::4::sec::3}
We now define a power conserving interconnection between passive systems associated with optimization problem with an equality constraint \eqref{maindyn} and an inequality constraint \eqref{IED} (see Fig. \ref{fig:my_label}).
			\begin{figure}[h!]
		\centering
		\includegraphics[width=1\linewidth]{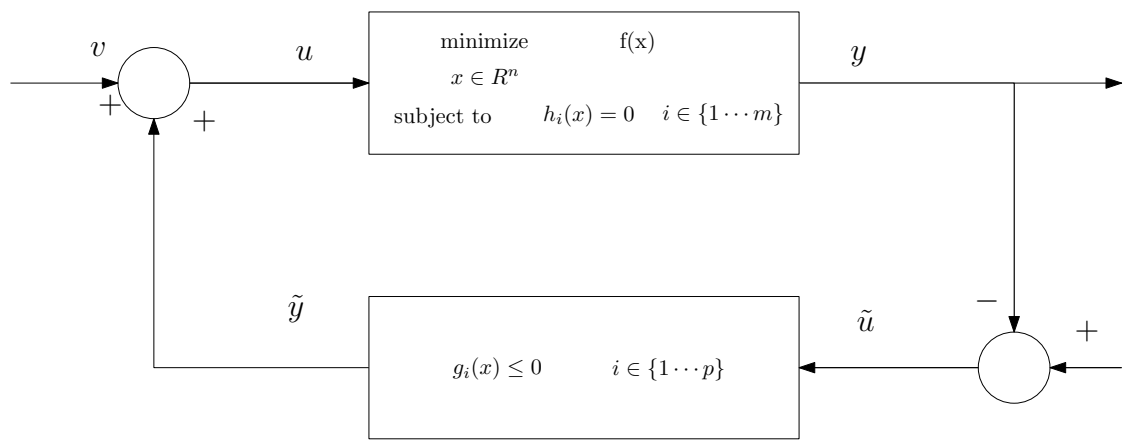}
		\caption{Interconnected optimization}
		\label{fig:my_label}
	\end{figure}
	\begin{proposition} \label{interconnectpassive}
		Consider the interconnection of passive systems \eqref{maindyn} and  \eqref{IED}, via the following interconnection constraints $ u=\tilde{y}+v~~ \text{and} ~~\tilde{u}=-y+\tilde{v}, ~ v\in \mathbb{R}^p, ~ \tilde v\in \mathbb{R}^n $. 
	For $\tilde v=0$, the interconnected system is then passive with port variables $\dot{v}$, $-\dot{x}$. 
		Moreover for $v=0$ and $\tilde v=0$ the interconnected system represents the primal-dual gradient dynamics of the optimization problem \eqref{standard_SOP}
		and the trajectories converge asymptotically to the optimal solution of \eqref{standard_SOP}.
	\end{proposition}
    \begin{proof}
    Define the storage function $\tilde{S}_{\sigma}(x,\lambda,\mu)=\tilde{P}(x,\lambda)+S_{\sigma}(\mu)$.
	The time differential of $\tilde{S}_{\sigma}(x,\lambda,\mu)$ with $\tilde{v}=0$ is
	\begin{eqnarray*}
		\dot{\tilde{S}}_{\sigma}(x,\lambda,\mu)
		&=&-\dot{u}^\top \dot{x}+\dot{\tilde{u}}^\top\dot{\tilde{y}}\leq -\dot{v}^\top \dot{x}
	\end{eqnarray*}
	The interconnection of  \eqref{maindyn} and \eqref{IED} (see Fig. \ref{fig:my_labela}), with $v=0$ and $\tilde{v}=0$, gives
	\begin{eqnarray}\label{primal-dual-dyn}
	-\tau_{x}\dot{x}&=&\left(\nabla_{x} f(x)+\sum_{i=1}^{m}\lambda_{i}\nabla_{x} h_{i}(x)+\sum_{i=1}^{p}\mu_{i}\nabla_{x} g_{i}(x)\right)\nonumber\\
	\tau_{\lambda_{i}}\dot{\lambda_{i}}&=& h_{i}(x)\nonumber\\
	\tau_{\mu_i}\dot{\mu_{i}}&=& \begin{cases}
	g_{i}(x) \;\;\;\;\;\;\;\;\;\;\;\;\;\;\text{if}\; \mu_{i}>0 \;\; \forall i \in \{1,\hdots, p\} \label{PriD}\\
	\text{max}(0,g_{i}(x))\;\; \text{if} \;\mu_{i}=0
	\end{cases}
	\end{eqnarray}
	which represent the primal-dual gradient dynamics of \eqref{standard_SOP}.
	Hence the overall system takes the form of primal-dual gradient dynamics representing optimization problem with both equality and in-equality constraints \eqref{standard_SOP}. \\
	When $v=0$ and $\tilde{v}=0$, $\dot{\tilde{S}}_{\sigma}(x,\lambda,\mu) \le 0$, for the interconnected system. Stability can thus be concluded using the relation between passivity and stability \cite{l2gain} and Propositions \ref{prop:ineq_passivity}, \ref{prop::ass_stab_ineq}. 
    \end{proof}
			\begin{figure}[h!]
	\centering
	\includegraphics[width=1\linewidth]{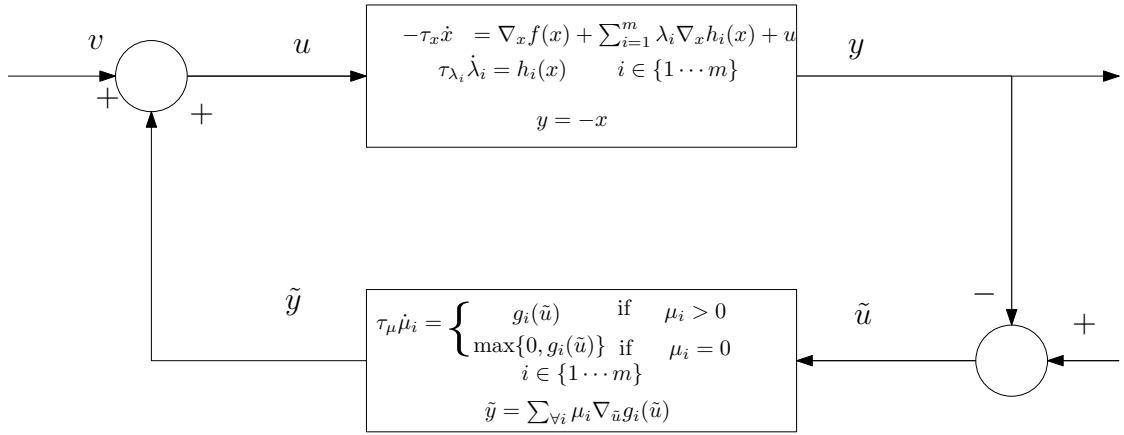}
	\caption{Interconnected primal dual dynamics}
	\label{fig:my_labela}
\end{figure}
The new port variables $\dot{v}$, $-\dot{x}$ can be used to change the convergence rate by damping injection methodology \cite{l2gain}. 
	One possible choice would be $v=k\nabla_x h(x) h(x)$ (note that equality constraint is affine) and $\tilde{v}=0$.
	The resulting dynamics 
    	\begin{eqnarray}\label{maindyna}
	-\tau_{x}\dot{x}&=&\nabla_{x} f(x)+\lambda\nabla_{x} h+\nabla_xg^{\top}\mu+kh(x)\nabla_x h(x)\nonumber\\
	\tau_{\lambda_{i}}\dot{\lambda}_{i}&=& h_{i}(x),\\
	\tau_{\mu_{i}}\dot{\mu}_{i}&=& (g_{i}(x))^{+}_{\mu_i},\;\;\nonumber
	y =-x
	\end{eqnarray}
    represents the primal-dual equations of augmented optimization problem
		\begin{equation}\label{SOP_main1}
		\begin{aligned}
		& \underset{x \in \mathbb{R}^{n}}{\text{minimize}}
		& & f(x) +\dfrac{1}{2}k(h(x))^2\\
		& \text{subject to}
		& & h(x)=0 , ~~g_{i}(x)\leq0 \hspace{0.4cm} i = 1,\hdots , p.
		\end{aligned}
		\end{equation}
        \begin{remark}
        Damping injection increases the convexity of the cost function (in this case we have $k(\nabla_xh(x)^2)-$strong convex function), and the resulting dynamics represents the primal-dual gradient laws of augmented Lagrangian.
        \end{remark}
%
%
%
%
%

		%
		%
		%
		%
		%
\noindent         In the next section, we demonstrate the continuous-time primal-dual algorithm, on the convex optimization formulation of Support Vector Machines (SVM) technique \cite{cortes1995support}. 
        \section{Support Vector Machine (SVM)}
 
Support Vector Machines \cite{cortes1995support} are a class of supervised machine learning algorithms which are commonly used for data classification. 
In this methodology, each data item is a point in $n$-dimensional space that is mapped to a category (or a class). 
Here the aim is to find an optimal separating hyperplane (OSH) which separates both the classes and maximizes the distance to the closest point from either class (as shown in Figure \ref{fig::svm_ex}). These closest points are usually called support vectors (SV). The lines passing through support vectors and parallel to the optimal separating hyperplane are called supporting hyperplanes (SH).

\begin{figure}[h!]
	\centering
	\includegraphics[width=0.8\linewidth]{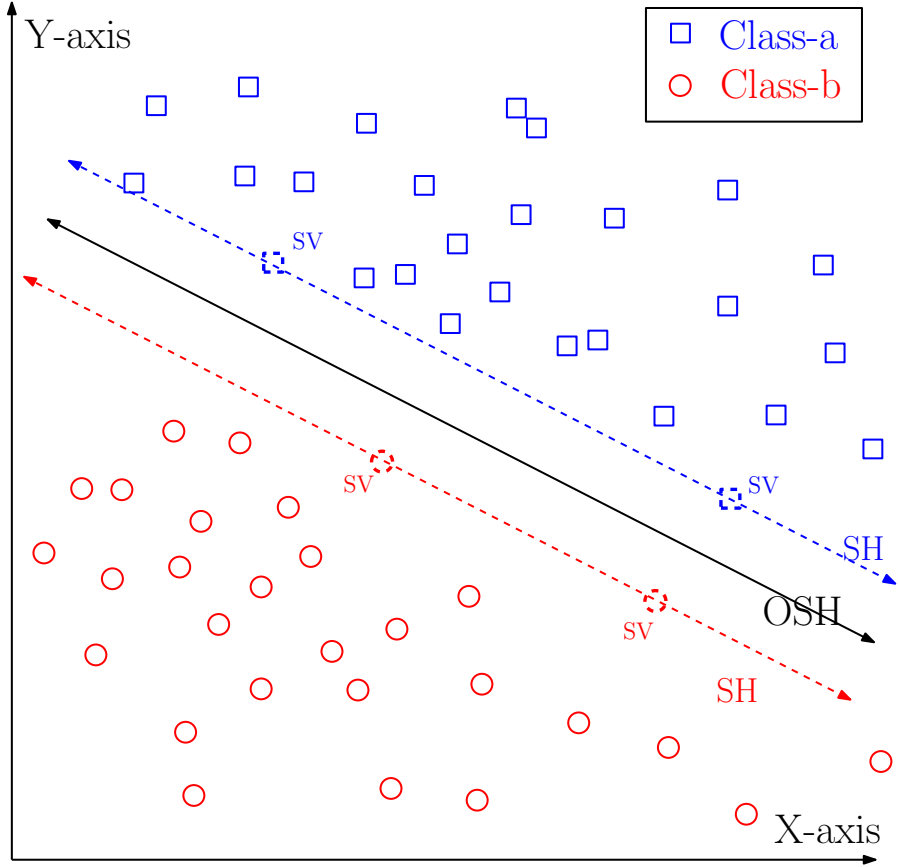}
	\caption{Description of a linear support vector machine.}
	\label{fig::svm_ex}
\end{figure}
{\em Problem formulation:} Consider two linearly separable classes, where each class (say class-$a$, class-$b$) contains a set of $N$ unique data points in $\mathbb{R}^2$. Let $X_a$ and $X_b$ denote the set of points in class-$a$ and class-$b$ respectively. In this methodology we find a hyperplane that separates the classes while maximizing the distance to the closest point from either class. Let $L$ be an affine set that characterizes such a hyperplane, defined as follows
  \beq
 L=\left\{x\in \mathbb{R}^2|x^\top \beta+\beta_0=0 \right\}
 \eeq
 where $\beta=(\beta_1,\beta_2)\in \mathbb{R}^2$ and $\beta_0\in \mathbb{R}$.  Define the map $l:\mathbb{R}^2\rightarrow \mathbb{R}$ by $l(x)=x^\top \beta +\beta_0$. Note the following, for any $x_0\in L$, $l(x_0)=0$ $\implies$ $x_0^\top \beta =-\beta_0$. This implies $l(x)$ can be rewritten as $l(x)=\beta^\top (x-x_0)$, which further implies the unit vector $\hat{\beta}=\dfrac{\beta}{||\beta||}$ is orthogonal to the line defined by the set $L$, that is, $x^\top\beta+\beta_0=0\iff (x-x_0)^\top \beta=0$.
\begin{figure}[h!]
	\centering
	\includegraphics[width=0.5\linewidth]{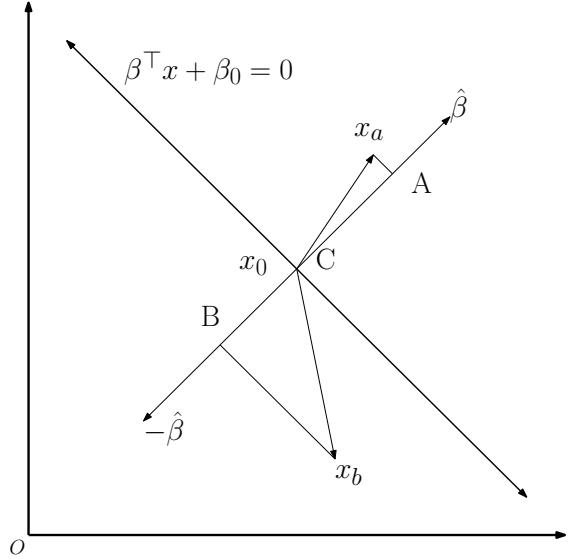}
	\caption{Mathematical formulation of a linear support vector machine, $x_a\in X_a$ (class-$a$) and $x_b\in X_b$ (class-$b$).}
	\label{fig:svm_derivation}
\end{figure}
The distance between the point $x_a\in X_a$ and line $L$ is $|AC|=(x_a-x_0)^\top \hat{\beta}$ (see Fig. \ref{fig:svm_derivation}). Similarly the distance between the point $x_b\in X_b$ and line $L$ is $|BC|=(x_b-x_0)^\top( -\hat{\beta})$. We want to find an optimal separating hyperplane that is at least $M$ units away from all the points. This implies
   \beq\label{svm_constraints}
   \begin{matrix}
 \forall x_a\in X_a, &  	(x_a-x_0)^\top \hat{\beta}     &\geq& M,\\
 \forall x_b\in X_b, & -(x_b-x_0)^\top \hat{\beta}&\geq& M.
   \end{matrix}
   \eeq
   Define $X\deff X_a\cup X_b$, and  $Y\deff Y_a \cup Y_b$ where $Y_a=\underbrace{\{1,\ldots,  1\}}_{\text{n\;\; times}}$ and $Y_b=\underbrace{\{-1,\ldots,  -1\}}_{\text{n\;\; times}}$. 
  The inequality constraints \eqref{svm_constraints} can be rewritten as
   \beq 
   \dfrac{1}{||\beta||}y_i(\beta^\top x_i+\beta_0)\geq M
   \eeq
   where $y_i=1$ if $x_i\in X_a$ (class-$a$), $y_i=-1$ if $x_i\in X_b$ (class-$b$). Finally, finding the optimal separating hyperplane can be proposed as the following optimization problem,
\begin{equation}\label{convex_form1}
\begin{aligned}
& \underset{\beta, \beta_0}{\text{maximize}}
& & M\\
& \text{subject to}
& & \dfrac{1}{||\beta||}y_i(\beta^\top x_i+\beta_0)\geq M,~\;\; \forall x_i\in X,~ y_i\in Y.
\end{aligned}
\end{equation} 
Since $M$ is arbitrary, choosing $M=\dfrac{2}{||\beta||}$ converts \eqref{convex_form1} into a convex optimization problem
\begin{equation}\label{convex_form2}
\begin{aligned}
& \underset{\beta, \beta_0}{\text{minimize}}
& & \dfrac{1}{2}||\beta||\\
& \text{subject to}
& & y_i(\beta^\top x_i+\beta_0)\geq 1,~\;\; \forall x_i\in X,~ y_i\in Y.
\end{aligned}
\end{equation} 
In order to use the primal-dual gradient method proposed in Section \ref{chap::4::sec::3}, we need the cost function to be twice differentiable. But, the cost function $\frac{1}{2}||\beta||\notin C^2$. 
The optimal solution $(\beta^\ast,\beta_0^\ast)$ of \eqref{convex_form2}, is further equivalent to the optimal solution of
\begin{equation}\label{convex_form3}
\begin{aligned}
& \underset{\beta, \beta_0}{\text{minimize}}
& & \dfrac{1}{2}||\beta||^2\\
& \text{subject to}
& & y_i(\beta^\top x_i+\beta_0)\geq 1,~\;\; \forall x_i\in X,~ y_i\in Y.
\end{aligned}
\end{equation}
We now use this convex optimization formulation for support vector machines, and derive its primal-dual gradient dynamics.\\
{\em Continuous time primal-dual gradient dynamics}: Comparing with the convex optimization formulation given in \eqref{standard_SOP}, the cost function is $f(\beta)=\dfrac{1}{2}||\beta||^2$ and inequality constraints are $g_i(\beta,\beta_0)=1-y_i(\beta^\top x_i+\beta_0)$, $i\in \{1,\cdots, 2N\}$. The Lagrangian can be written as
\beq
L(\beta,\mu)=\dfrac{1}{2}||\beta||^2+\sum_{i=1}^{2N}g_i(\beta,\beta_0)\mu_i
\eeq
where $\mu=(\mu_1,\cdots,\mu_{2N})$ denotes the Lagrange variable corresponding to the inequality constraints $g=(g_1,\cdots, g_{2N})$. The primal dual gradient laws given in \eqref{primal-dual-dyn} for the convex optimization problem \eqref{convex_form3} are
 \beqn
\begin{matrix}
	-\tau_{\beta}\dot{\beta}&=&\dfrac{\partial L}{\partial \beta}\\
		-\tau_{\beta_0}\dot{\beta}_0&=&\dfrac{\partial L}{\partial \beta_0}\\
			\tau_{\mu_i}\dot{\mu}_i&=&(g_i(\beta,\beta_0))^+_{\mu_i}\;\; \forall i \in \{1,\hdots, 2N\}
\end{matrix}
\eeqn
equivalently ,
\beq\label{primal-dual_svm}
	-\tau_{\beta}\dot{\beta}&=&\beta-\sum_{i=1}^{2N} \mu_iy_ix_i\nonumber\\
	-\tau_{\beta_0}\dot{\beta}_0&=&-\sum_{i=1}^{2N} \mu_iy_i\\
	\tau_{\mu_i}\dot{\mu}_i&=&\begin{cases}
		g_{i}(\beta,\beta_0) \;\;\;\;\;\;\;\;\;\;\;\;\;\;\text{if}\; \mu_{i}>0 \;\; \forall i \in \{1,\hdots, 2N\} \label{PriD}\\
		\text{max}\{0,g_{i}(\beta,\beta_0)\}\;\; \text{if} \;\mu_{i}=0.
	\end{cases}\nonumber
\eeq 
We now have the following result.
\begin{proposition}
	The primal-dual dynamics \eqref{primal-dual_svm} converges asymptotically to the optimal solution of \eqref{convex_form3}.
\end{proposition}
\begin{proof}
			Since the optimization problem \eqref{convex_form3} has a strictly convex cost function and convex inequality constraints, the result follows from Propositions \ref{prop::eq_const} - \ref{interconnectpassive}.
\end{proof}

\subsection{Simulation Results}
A simulation study is conducted by generating two sets of linearly separable classes having 300 points each, using Normal distribution (see Table \ref{tab:data_gen} for distribution parameters). Figure \ref{fig::svm2} present the evolution of $\beta,\;\beta_0$. 
\begin{table}[h!]
	\centering
	\caption{Distribution parameters}
	\label{tab:data_gen}
	\begin{tabular}{|l|l|l|l|}
		\hline
		\textbf{}        & \textbf{mean}                     & \textbf{Variance}                          & \textbf{No. of data points} \\ \hline
		\textbf{Class-a} & $\begin{bmatrix}0&0\end{bmatrix}$ & $\begin{bmatrix}1&1.5\\1.5&3\end{bmatrix}$ & 300                       \\ \hline
		\textbf{Class-b} & $\begin{bmatrix}0&6\end{bmatrix}$ & $\begin{bmatrix}1&1.5\\1.5&3\end{bmatrix}$ & 300                         \\ \hline
	\end{tabular}
\end{table}
 At equilibrium, the primal-dual dynamics in equation \eqref{primal-dual_svm} results in
 $$\beta^{\ast}=\sum_{i=1}^{2N} \mu_i^{\ast}y_ix_i.$$
 The results depicted in Figure \ref{fig::svm2a} show that the value the Lagrange variables, except ($ \mu_{81},\; \mu_{208},\;\;\mu_{577}$) are identically equal to zero at equilibrium.  Hence 
 $$\beta^\ast=\mu_{81}^\ast x_{81}+\mu_{208}^\ast x_{208}-\mu_{577}^\ast x_{577}$$
 where the data points ($ x_{81},\; x_{208},\;\;x_{577}$) corresponding to these non zero Lagrange variables are support vectors. This implies that the support vectors completely determines the optimal separating hyperplane $\beta^\top x+\beta_0=0$ that separates class-$a$ and class-$b$ (see Fig. \ref{fig::svm1}). However, note that one needs to solve the optimization problem, to find these support vectors.
%
%
    \begin{remark} Remark on Figure \ref{fig::svm1}. In the case study we have 600 inequality constraints. Whenever, an inequality constraint becomes feasible  (i.e. $g_i(\beta,\beta_0)\leq  0$ ) and its corresponding Lagrange variable $\mu_i$ converges to zero, then the closed loop storage function switches to a new storage function that is strictly less than the current one, causing a discontinuity. 
	%
	This is coherent with the Proposition \ref{prop:ineq_passivity}, where passivity property is defined with `multiple storage functions'.
\end{remark}
\begin{remark} Remark on Figure \ref{fig::svm2}. One can see that, all the Lagrange variables except $\mu_{81}$, $\mu_{208}$ and $\mu_{577}$  converged to zero. Moreover, the data points corresponding to these non-zero Lagrange variables are called support vectors, can be seen in Fig. \ref{fig::svm3}
\end{remark}

\begin{remark}Remark on Figure \ref{fig::svm3}.	The three points denoted by $x_{81}$, $x_{208}$ and $x_{577}$ are usually called as support vectors, and the Lagrange variables corresponding to their inequality constraints are non-zero (can be seen in Figure \ref{fig::svm2}). The lines passing through these point and parallel to the separating hyperplane are called supporting hyperplanes. 
\end{remark}
 \begin{figure}[h!]
		\includegraphics[trim={0cm 0cm 0 0cm},scale=0.28]{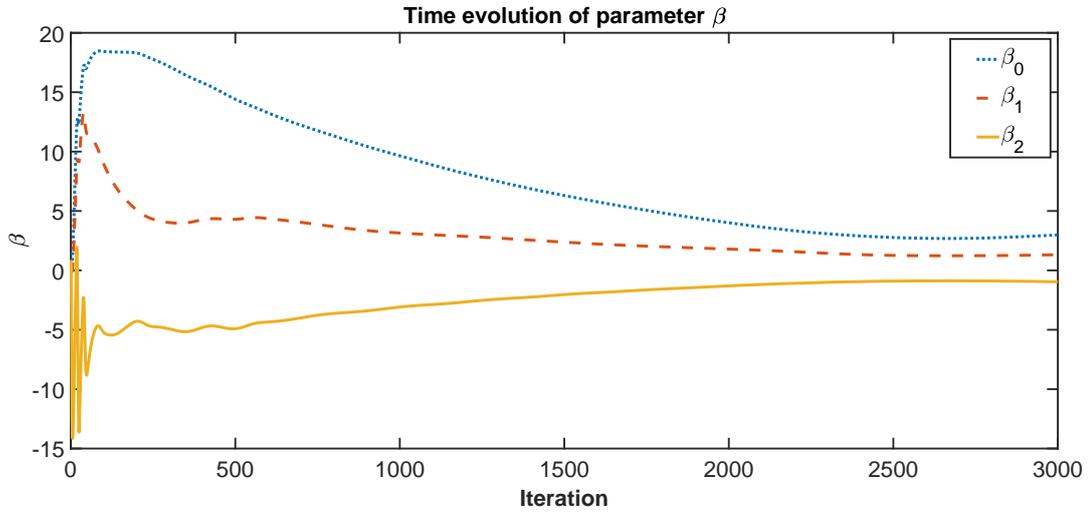}
		\caption{Time evolution of $\beta$ and $\beta_0$ }
		\label{fig::svm2a}
	\end{figure}
	\begin{figure}[h!]
		\includegraphics[trim={8cm 0cm 0 0cm}, scale=0.435]{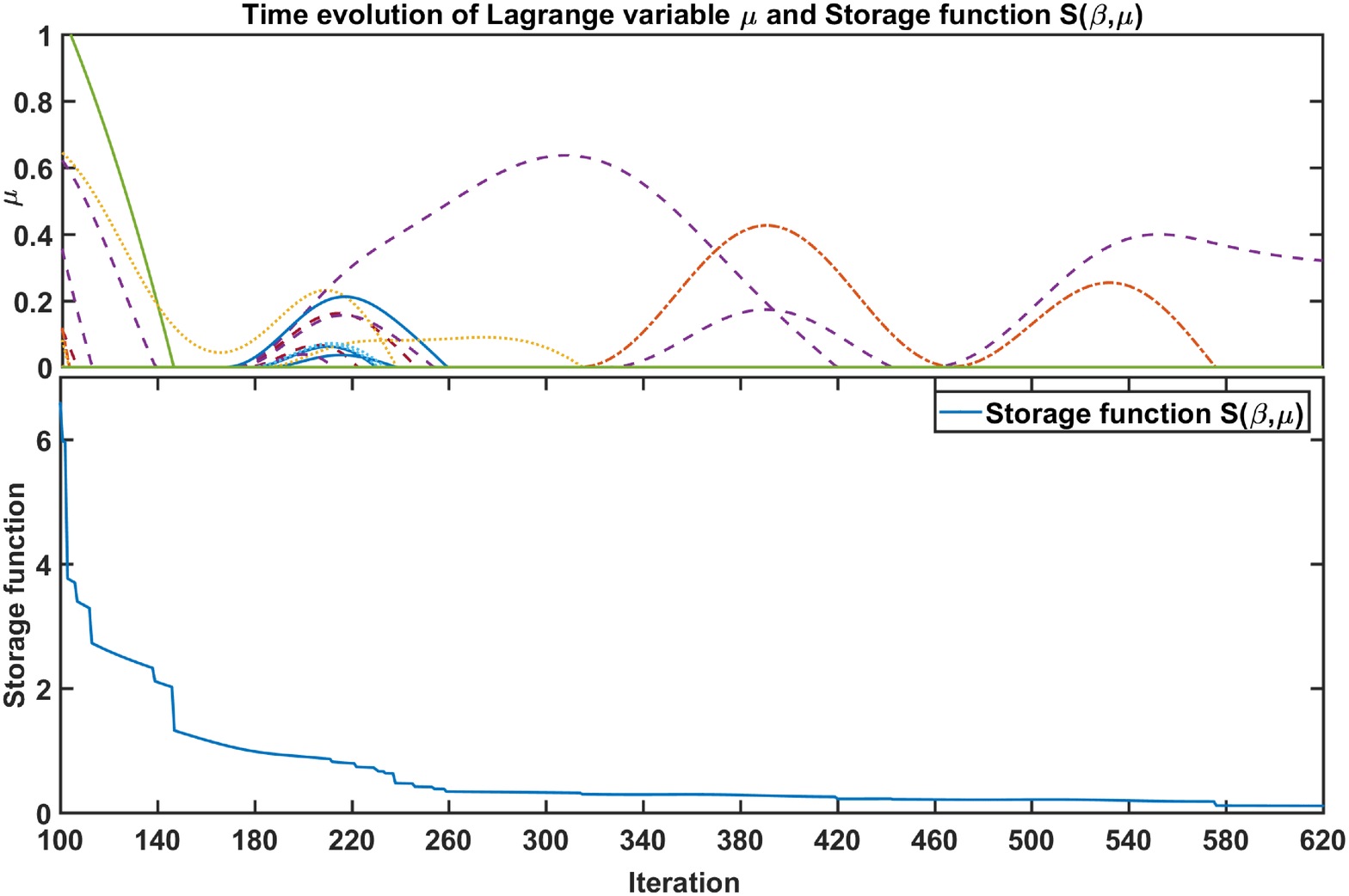}
        		\caption{Time evolution of closed-loop storage function. }
		\label{fig::svm1}
	\end{figure}
			\begin{figure}[h!]
		\includegraphics[trim={5cm 0cm 0 0cm},clip, width=16cm,height=10cm]{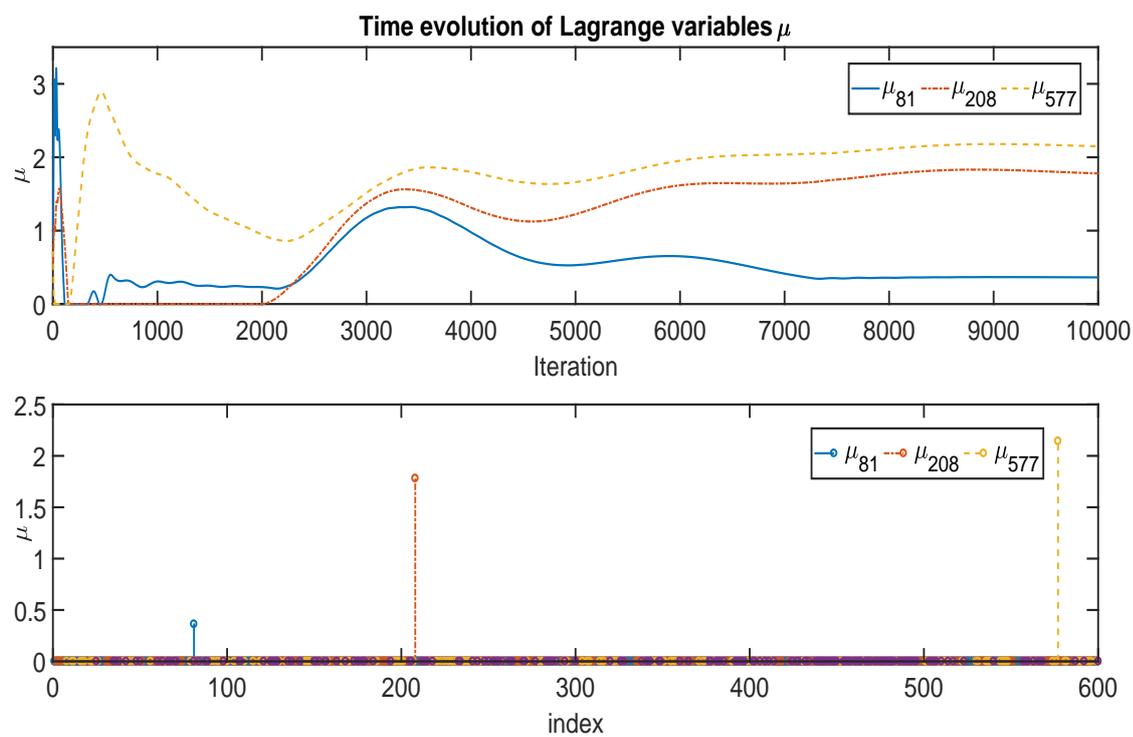}
		\caption{Time evolution of Lagrange variables $\mu_i$, $i\in \{1\cdots 600\}$. }
		\label{fig::svm2}
	\end{figure}
    \begin{figure}[h!]
		\includegraphics[		width=25cm,height=14cm,, trim={0cm 0cm 0 0cm},clip,   angle =90]{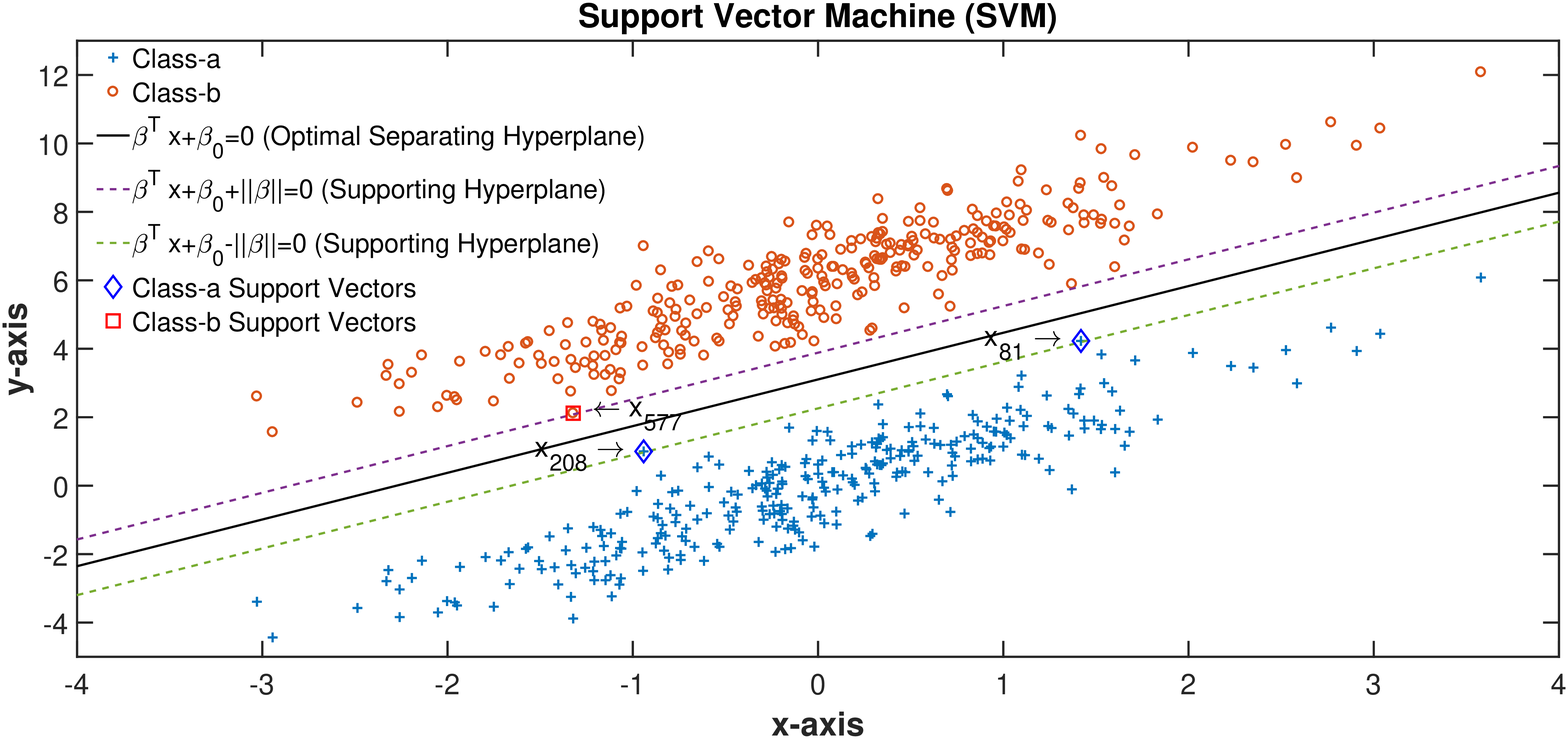}
		\caption{Classification using Support Vector Machine}
		\label{fig::svm3}
	\end{figure}
	\section{Conclusions}\label{concl}
	Starting from an optimization problem with equality constraint we have shown that their primal-dual equations have a naturally existing Brayton Moser representation. 
	Using the interconnection properties of passive systems we extended the optimization problem to include inequality constraints. The overall convergence is guaranteed by proving the asymptotic stability of individual subsystems, whose Lyapunov functions derived from BM formulation have their roots in Krasovskii method. As an example, we have demonstrated the primal-dual algorithm using the convex optimization formulation of SVM technique.
\afterpage{\blankpage}
	\addtocounter{page}{1}%
  \chapter{Concluding remarks}
In this book, we have presented various passivity based methodologies for control and optimization in Brayton-Moser framework. Each chapter has its own concluding remarks with some of the highlights and the limitations of the proposed work. We now discuss the remaining issues and propose some possible directions for future work.
%
%
%
%

In the Chapter 3, we have addressed two major limitations in using Brayton-Moser framework for control by power shaping. 

\hspace{-0.5cm}
\begin{minipage}{0.05\linewidth}
\vspace{-3.9cm}
(i)
\end{minipage}\begin{minipage}{0.95\linewidth}
The first one involves finding of the closed-loop storage function using the admissible pairs. One constructive methodology for finding the closed-loop storage function, that does not requires solving partial differential equations, involves in finding passive maps with integrable output port-variable. We presented a solution for this impediment under the assumption that the 1-forms corresponding to the columns of the input matrix are all closed. 
\end{minipage}
\vspace{0.5cm}

\hspace{-0.5cm}
\begin{minipage}{0.05\linewidth}
\vspace{-4.6cm}
(ii)
\end{minipage}\begin{minipage}{0.95\linewidth}
The second limitation concerns to the existence of admissible pairs. This has been addressed by introducing Krasovskii-type storage functions, which resulted in a new passivity property with integrable port-variables. Towards the end we have extended these results to a class of nonlinear systems characterized by Assumptions \ref{ass::A1}, \ref{ass::A2} and \ref{ass::A3}. Similar kind of studies have been carried out in order to extract new passivity properties of systems, namely differential passivity \cite{van2013differential} and incremental passivity \cite{l2gain}. 
\end{minipage}

%
\noindent {\em Relations to differential and incremental passivity}: 
Consider the prolonged system \cite{cortes2005characterization,crouch1987variational}, that is the original non-linear system \eqref{gen_sys} together with its variational system.
	\begin{eqnarray}\label{incre_extended_sys}
	    \dot{x}&=&f(x)+g(x)u\nonumber \\
	    \dot{\delta x} &=& \left(\dfrac{\partial f}{\partial x}+\dfrac{\partial g}{\partial x} u\right)\cdot \delta x+g(x)\delta u
	\end{eqnarray}
where $\delta x\in \mathbb{R}^n$, $\delta u\in \mathbb{R}^m$ denotes the variation in $x$ and $u$ respectively. The following theorem shows that the prolonged nonlinear system is differentially passive, under Assumptions \ref{ass::A1}, \ref{ass::A2} and \ref{ass::A3}.
	\begin{theorem}\label{thm:conclus::diff_pass} Let the Assumptions \ref{ass::A1}, \ref{ass::A2} and \ref{ass::A3} are satisfied. Then the system of equations \eqref{incre_extended_sys} are passive with port variable $\delta \Gamma=g^\top M \delta x$ and $\delta v= \delta{u}-\alpha u-\beta$. Where $\alpha=-\left(g^\top g\right)^{-1}g^\top \dfrac{\partial g}{\partial x}\delta x$ and $\beta=-g^\top M \delta{x}$.
	\end{theorem}
\noindent The proof of the theorem follows on the similar lines of Theorem \ref{thm::main} with storage function in equation \eqref{kras_lypunov} replaced by $V(x,\delta x) = \dfrac{1}{2}\delta x^\top M\delta x$. Also note that from Assumption \ref{ass::A3}, $Mg(x)$ is integrable. This indicates the existence of  of $\Gamma(x)$ such that $\dfrac{\partial \Gamma}{\partial x}=Mg(x)$ leading to $\delta \Gamma=\dfrac{\partial \Gamma}{\partial x}^\top \delta x=g^\top M \delta x$. 
%
\begin{remark}
In the above Theorem \ref{thm:conclus::diff_pass}, if we consider $g(x)=B$, and $\Gamma(x)=Cx$, where $B\in \mathbb{R}^{n\times m}$ and $C\in \mathbb{R}^{m\times n}$ are constant, then we recover the conditions presented from incremental passivity in \cite{pavlov2008incremental}. 
\end{remark}
\noindent {\em Relation with Integral Quadratic Constraints (IQC) \cite{IQC}:}
Consider the linear system defined by 
\begin{eqnarray*}
	\begin{matrix}
		\dot{x}=Ax+Bu\\y=Cx+Du,
	\end{matrix}
\end{eqnarray*}
where $x\in \mathbb{R}^n$, $u,y\in \mathbb{R}^m$ and $A,B,C,D$ are of appropriate dimensions. Then, one can write the following equivalence
\begin{eqnarray*}
	y=H(u) \iff \dfrac{Y(s)}{U(s)}=G(s)=C(sI-A)^{-1}B+D,
\end{eqnarray*}
where $H$ is a bounded operator and $Y(s)$ \& $U(s)$ are Laplace's transforms of $y(t)$ and $u(t)$, respectively. 	Let $\Pi$ be a bounded and self adjoint operator. Then $u$ satisfies the IQC defined by $\Pi $ if 
\begin{eqnarray}
\left<\begin{bmatrix}
u\\y
\end{bmatrix},\Pi \begin{bmatrix}
u\\y
\end{bmatrix}\right>\geq 0\;\; \forall u\in \mathcal{H}.
\end{eqnarray}
In transfer function domain this translate to 
\begin{eqnarray}
\int_{-\infty}^{\infty}\left(\begin{bmatrix}
U(j\omega)\\Y(j\omega)
\end{bmatrix}^{\ast}\Pi(j\omega) \begin{bmatrix}
U(j\omega)\\Y(j\omega)
\end{bmatrix}\right)\geq 0\;\; \forall u\in \mathcal{H}.
\end{eqnarray}
\noindent We list a few relevant passivity conditions and their equivalent conditions in frequency domain and with respect to IQC operators. A system is passive
\begin{itemize}
	\item[(i)] {
		with port variables $u$ and $y$, if
		\begin{eqnarray}
		<y,u>_T\geq 0\;\; \forall T\geq 0,\;\; \forall u\in \mathcal{H}_e
`		\end{eqnarray}
		where $\mathcal{H}_e$ denotes the extended Hilbert's space. This is equivalent to 
		\begin{eqnarray}\label{pass_uy}
		G(j\omega) +G(j\omega)^\ast \geq 0\;\;\; \forall \omega.
		\end{eqnarray}
		It also implies that the real part of $G(j\omega)$ should be positive, which is well known as  the Positive Real condition. Further the IQC operator $\Pi$ takes the form $\begin{bmatrix}
		\mathrm{O}& \mathrm{I}\\\mathrm{I} & \mathrm{O}
		\end{bmatrix}$, where $O,\; I$ are zero and identity matrix of dimension $m$.}
	\item[(ii)]{with $u$ and $\dot{y}$, as in Brayton Moser framework, if
		\begin{eqnarray}
		<\dot{y},u>_T\geq 0\;\; \forall T\geq 0,\;\; \forall u\in \mathcal{H}_e
		\end{eqnarray}
		This can be simplified and expressed in frequency domain as
		\begin{eqnarray}
		-j\omega (G(j\omega) - G(j\omega)^\ast) \geq 0\;\;\; \forall \omega.
		\end{eqnarray}
		This implies that the imaginary part of the transfer function should be negative, which is termed here as the Negative Imaginary condition. In this case, IQC operator takes the form $\Pi=\begin{bmatrix}
		\mathrm{O}& j\omega \mathrm{I}\\-j\omega \mathrm{I} & \mathrm{O}
		\end{bmatrix}$.}
	\item[(iii)]{with $\dot{u}$ and $\dot{y}$, if,
		\begin{eqnarray}
		<\dot{y},\dot{u}>_T\geq 0\;\; \forall T\geq 0,\;\; \forall u\in \mathcal{H}_e.
		\end{eqnarray}
		This results in Positive Real condition \eqref{pass_uy}. This is not surprising, since it is well-known that if a linear system is passive with respect to port-variables $u$ and $y$, then is passive w.r.t port-variables $\dot{u}$ and $\dot{y}$. Further, in the case of $\dot{u}$ and $\dot{y}$, IQC operator $\Pi$ again takes the same form $\begin{bmatrix}
		\mathrm{O}& \mathrm{I}\\\mathrm{I} & \mathrm{O}
		\end{bmatrix}$, as in the case of $u$ and $y$ as port-variables. 
	}
\end{itemize}
\noindent The following are the important directions for future work:

\hspace{-0.5cm}
\begin{minipage}{0.05\linewidth}
\vspace{-0.85cm}
(i)
\end{minipage}\begin{minipage}{0.95\linewidth}
Finding integrable passive maps with out relying on the assumptions that the input matrix is integrable.
\end{minipage}

\hspace{-0.5cm}
\begin{minipage}{0.05\linewidth}
\vspace{-0.85cm}
(ii)
\end{minipage}\begin{minipage}{0.95\linewidth}
Exploring connections between dynamic feedback passivation and differential passivity.
\end{minipage}
\vspace{0.5cm}

%
%
%
%
%
In Chapter 4, the primal-dual algorithm is treated as interconnected passive systems, (i) convex optimization problem with only equality constraint, (ii) a state dependent switching system for inequality constraint. Recall that in Proposition 4.4, we interconnected these systems using
\beq\label{interconnection_ch6}
\begin{bmatrix}
	u\\\tilde{u}
\end{bmatrix}&=& \begin{bmatrix}
0 & 1\\-1 &0
\end{bmatrix}\begin{bmatrix}
y\\ \tilde{y}
\end{bmatrix}+\begin{bmatrix}
v\\\tilde{v}
\end{bmatrix}
\eeq
where $v$ and $\tilde{v}$ are considered as new input port-variables of the interconnected system. We can use these new port-variables to analyze and improve the primal-dual gradient laws. The following are some of the important ideas that can be  leveraged for future work.
%

\hspace{-0.5cm}
\begin{minipage}{0.05\linewidth}
\vspace{-6.4cm}
(i)
\end{minipage}\begin{minipage}{0.95\linewidth}
{\em Robustness}: To analyze uncertainties in parameters or disturbances such as the numerical error accumulated in the primal and dual variables, one can rewrite interconnection as
	   \begin{equation}
	   u=\tilde{y}+\Delta \tilde{y}~~ \text{and} ~~\tilde{u}=x+\Delta x
	  \end{equation}
	   where $\Delta x$ and $\Delta \tilde{y}$ denotes the numerical error in $x$ (primal variable) and $\tilde{y}$ (a function of dual variable) respectively. These can be treated as external disturbances  creeping in through the interconnected port variables. 
We can provide robustness analysis quantitatively (on sensitivity of the algorithm due to numerical errors), using input/output dissipative properties \cite{l2gain} of these systems.
\end{minipage}

\hspace{-0.5cm}
\begin{minipage}{0.05\linewidth}
\vspace{-5.4cm}
(ii)
\end{minipage}
\begin{minipage}{0.95\linewidth}
In SVM simulation we have seen that there are 600 inequality constraints (each corresponds to a data-point). Usually, real world examples may contain many more data-points. Each data-points gives rise to an inequality constraint, and further leads to a gradient-law. In situations involving large data, it is computationally ineffective to run gradient-descent algorithm using all the data-points. In general this obstacle is circumvented using a variation in gradient descent method called {\em stochastic gradient descent}. Can we propose a passivity based convergence analysis for stocastic gradient descent?
\end{minipage}
\vspace{0.5cm}

\hspace{-0.5cm}
\begin{minipage}{0.05\linewidth}
\vspace{-3.1cm}
(iii)
\end{minipage}\begin{minipage}{0.95\linewidth}
{\em Controller}: Using these new port variables one can interconnect the primal-dual dynamics to a plant, such that the closed-loop system is again a passive dynamical system \cite{stegink2017unifying}. Moreover, one can explore the idea of Barrier functions \cite{boyd2004convex} to derive a bounded controller. Gradient methods are inherently distributed computing methods. Hence the controllers derived from these may inherit this property.
\end{minipage}
\vspace{0.5cm}

In chapter 5, we have discussed modeling and control aspects  of infinite-dimensional port-Hamiltonian systems in  Brayton-Moser framework. In retrospection, we first showed the existence of dissipation obstacle and motivated the need for the new class of power-based storage functions. Next, we presented the Brayton-Moser formulation of an infinite-dimensional port-Hamiltonian system defined by Stokes-Dirac structure. Consequently, we had shown that these Brayton-Moser equations can be written as a distributed parameter dynamical system with respect to a noncanonical Dirac structure. Similar to finite-dimensional systems, the mixed-potential function obtained from Brayton -Moser formulation was indefinite, hence cannot be used to infer any stability or passivity properties. This resulted in the quest for finding new gradient structures called admissible pairs. 
We then presented a systematic way of finding these admissible pairs. Further, these admissible pairs were used to derive (i) the stability analysis of Maxwell's equations in $\mathbb{R}^3$ with zero energy flows through boundary and (ii) the passive maps for boundary controlled transmission line system modeled by Telegraphers equations. By adopting the new passive maps, we presented a boundary control methodology for transmission line system using control by interconnection. 
The following are some of the important future directions:

\begin{minipage}{0.05\linewidth}
	\vspace{-0.75cm}
	(i)
\end{minipage}\begin{minipage}{0.95\linewidth}
Extending Brayton-Moser formulation to a general class of infinite-dimensional systems from various physical domains.
\end{minipage}

\begin{minipage}{0.05\linewidth}
	\vspace{-2.3cm}
	(ii)
\end{minipage}\begin{minipage}{0.95\linewidth}
In Chapter 4, we have seen that Brayton-Moser formulation for finite-dimensional systems helped us analyze optimization problems.  Can we present a similar formulation for optimal control using the Brayton-Moser framework developed for infinite-dimensional systems?
\end{minipage}
\vspace{0.5cm}
\afterpage{\blankpage}
\addtocounter{page}{1}%



\begin{singlespace}
	\bibliography{refs}

\begin{thebibliography}{0}
\expandafter\ifx\csname natexlab\endcsname\relax\def\natexlab#1{#1}\fi
\expandafter\ifx\csname url\endcsname\relax
  \def\url#1{{\tt #1}}\fi
\expandafter\ifx\csname urlprefix\endcsname\relax\def\urlprefix{URL }\fi

\end{thebibliography}


\begin{thebibliography}{10}

\bibitem{l2gain}
Arjan van~der Schaft.
\newblock {\em $L_2$-gain and Passivity Techniques in Nonlinear Control}.
\newblock Springer International Publishing AG, 2017.

\bibitem{OrtVanMasEsc}
Romeo Ortega, Arjan van~der Schaft, Bernhard Maschke, and Gerardo Escobar.
\newblock Interconnection and damping assignment passivity-based control of
  port-controlled {H}amiltonian systems.
\newblock {\em Automatica}, 38(4):585--596, 2002.

\bibitem{ortega2001putting}
Romeo Ortega, Arjan van~der Schaft, Iven Mareels, and Bernhard~M. Maschke.
\newblock Putting energy back in control.
\newblock {\em IEEE Control Systems}, 21(2):18--33, 2001.

\bibitem{maxwell1867governors}
J~Clerk Maxwell.
\newblock On governors.
\newblock {\em Proceedings of the Royal Society of London}, 16:270--283, 1867.

\bibitem{bookgeo}
Vincent Duindam, Alessandro Macchelli, Stefano Stramigioli, and Herman
  Bruyninckx.
\newblock {\em Modeling and Control of Complex Physical Systems: The
  port-{H}amiltonian Approach}.
\newblock Springer Science \& Business Media, 2009.

\bibitem{BraMos1964I}
R.~K. Brayton and J.~K. Moser.
\newblock A theory of nonlinear networks. i.
\newblock {\em Quarterly of Applied Mathematics}, 22(1):1--33, 1964.

\bibitem{BraMos1964II}
R.~K. Brayton and J.~K. Moser.
\newblock A theory of nonlinear networks. ii.
\newblock {\em Quarterly of Applied Mathematics}, 22(2):81--104, 1964.

\bibitem{BraMir1964}
R.~K. Brayton and W.~L. Miranker.
\newblock A stability theory for nonlinear mixed initial boundary value
  problems.
\newblock {\em Archive for Rational Mechanics and Analysis}, 17(5):358--376,
  1964.

\bibitem{smale2000mathematical}
Stephen Smale.
\newblock On the mathematical foundations of electrical circuit theory.
\newblock In {\em The Collected Papers of Stephen Smale: Volume 2}, pages
  951--968. World Scientific, 2000.

\bibitem{jeltsema2009multidomain}
Dimitri Jeltsema and Jacquelien~MA Scherpen.
\newblock Multidomain modeling of nonlinear networks and systems.
\newblock {\em IEEE Control Systems Magazine}, 29(4):28--59, 2009.

\bibitem{chinde2016building}
Venkatesh Chinde, Krishna~Chaitanya Kosaraju, Atul Kelkar, Ramkrishna
  Pasumarthy, S.~Sarkar, and Navdeep~M Singh.
\newblock A passivity-based power-shaping control of building hvac systems.
\newblock {\em Journal of Dynamic Systems, Measurement, and Control},
  139(11):111007--111007--10, 2017.

\bibitem{ICCvdotidot}
Krishna~Chaitanya Kosaraju, Ramkrishna Pasumarthy, Navdeep~M Singh, and A.~L.
  Fradkov.
\newblock Control using new passivity property with differentiation at both
  ports.
\newblock In {\em Indian Control Conference, Guwahati, India}, pages 7--11,
  2017.

\bibitem{ACC18}
Krishna~Chaitanya Kosaraju, Venkatesh Chinde, Ramkrishna Pasumarthy, Atul
  Kelkar, and Navdeep~M Singh.
\newblock Differential passivity like properties for a class of nonlinear
  systems.
\newblock In {\em American Control Conference (ACC), under review}, 2018.

\bibitem{Hugo}
Hugo Rodr{\'\i}guez, Arjan van~der Schaft, and Romeo Ortega.
\newblock On stabilization of nonlinear distributed parameter port-controlled
  hamiltonian systems via energy shaping.
\newblock In {\em 40th IEEE Conference on Decision and Control, Orlando, USA},
  volume~1, pages 131--136, 2001.

\bibitem{AleCla}
Alessandro Macchelli and Claudio Melchiorri.
\newblock Control by interconnection of mixed port {H}amiltonian systems.
\newblock {\em IEEE Transactions on Automatic Control}, 50(11):1839--1844,
  2005.

\bibitem{DimVan07}
Dimitri Jeltsema and Arjan van~der Schaft.
\newblock Pseudo-gradient and {L}agrangian boundary control system formulation
  of electromagnetic fields.
\newblock {\em Journal of Physics A: Mathematical and Theoretical},
  40(38):11627, 2007.

\bibitem{mtns}
Ramkrishna Pasumarthy, Krishna~Chaitanya Kosaraju, and Addarsh Chandrasekar.
\newblock On power balancing and stabilization for a class of
  infinite-dimensional systems.
\newblock In {\em 21st International Symposium on Mathematical Theory of
  Networks and Systems, Groningen, The Netherlands}, 2014.

\bibitem{IFAC}
Krishna~Chaitanya Kosaraju, Ramkrishna Pasumarthy, and Dimitri Jeltsema.
\newblock Alternative passive maps for infinite-dimensional systems using
  mixed-potential functions.
\newblock In {\em 5th IFAC Workshop on Lagrangian and Hamiltonian Methods for
  Nonlinear Control LHMNC, Lyon, France}, volume~48, pages 1--6, 2015.

\bibitem{book}
Krishna~Chaitanya Kosaraju and Ramkrishna Pasumarthy.
\newblock Power-based methods for infinite-dimensional systems.
\newblock In {\em Mathematical Control Theory I}, pages 277--301. Springer,
  2015.

\bibitem{ima2018}
Krishna~Chaitanya Kosaraju, Ramkrishna Pasumarthy, and Dimitri Jeltsema.
\newblock Modeling and boundary control of infinite dimensional systems in the
  brayton-moser framework.
\newblock {\em IMA Journal of Mathematical Control and Information}, pp, 2018.

\bibitem{kosaraju2018stability}
Krishna~Chaitanya Kosaraju, Venkatesh Chinde, Ramkrishna Pasumarthy, Atul
  Kelkar, and Navdeep~M Singh.
\newblock Stability analysis of constrained optimization dynamics via passivity
  techniques.
\newblock {\em IEEE Control Systems Letters}, 2(1):91--96, 2018.

\bibitem{van2014port}
Arjan van~der Schaft and Dimitri Jeltsema.
\newblock Port-hamiltonian systems theory: An introductory overview.
\newblock {\em Foundations and Trends$^{\textregistered}$ in Systems and
  Control}, 1(2-3):173--378, 2014.

\bibitem{courant1990dirac}
Theodore~James Courant.
\newblock Dirac manifolds.
\newblock {\em Transactions of the American Mathematical Society},
  319(2):631--661, 1990.

\bibitem{dorfman1993dirac}
Irene Dorfman.
\newblock {\em Dirac structures and integrability of nonlinear evolution
  equations}.
\newblock Wiley, 1993.

\bibitem{dalsmo1998representations}
Morten Dalsmo and Arjan van~der Schaft.
\newblock On representations and integrability of mathematical structures in
  energy-conserving physical systems.
\newblock {\em SIAM Journal on Control and Optimization}, 37(1):54--91, 1998.

\bibitem{garcia2005control}
Elo{\'\i}sa Garcia-Canseco, Ramkrishna Pasumarthy, Arjan van~der Schaft, and
  Romeo Ortega.
\newblock On control by interconnection of port hamiltonian systems.
\newblock In {\em Proceedings of the 16th IFAC World Congress, Prague, The
  Czech Republic}, volume~38, pages 330--335, 2005.

\bibitem{castanos2009asymptotic}
Fernando Casta{\~n}os, Romeo Ortega, Arjan Van~der Schaft, and Alessandro
  Astolfi.
\newblock Asymptotic stabilization via control by interconnection of
  port-hamiltonian systems.
\newblock {\em Automatica}, 45(7):1611--1618, 2009.

\bibitem{ortega2007control}
Romeo Ortega, Arjan van~der Schaft, Fernando Castanos, and Alessandro Astolfi.
\newblock Control by (state--modulated) interconnection of port--hamiltonian
  systems.
\newblock In {\em 7th IFAC Symposium on Nonlinear Control Systems, Pretoria,
  South Africa}, volume~40, pages 28--35, 2007.

\bibitem{ortega2008control}
Romeo Ortega, Arjan van~der Schaft, Fernando Castanos, and Alessandro Astolfi.
\newblock Control by interconnection and standard passivity-based control of
  port-hamiltonian systems.
\newblock {\em IEEE Transactions on Automatic Control}, 53(11):2527--2542,
  2008.

\bibitem{Guido}
Guido Blankenstein.
\newblock Geometric modeling of nonlinear {RLC} circuits.
\newblock {\em IEEE Transactions on Circuits and Systems I: Regular Papers},
  52(2):396--404, 2005.

\bibitem{blankenstein2003joined}
Guido Blankenstein.
\newblock A joined geometric structure for hamiltonian and gradient control
  systems.
\newblock In {\em IFAC Lagrangian and Hamiltonian Methods for Nonlinear
  Control, Seville, Spain}, volume~36, pages 51--56, 2003.

\bibitem{ElosiaDimOrt}
Elo{\'\i}sa Garc{\'\i}a-Canseco, Dimitri Jeltsema, Romeo Ortega, and Jacquelien
  M.~A. Scherpen.
\newblock Power-based control of physical systems.
\newblock {\em Automatica}, 46(1):127--132, 2010.

\bibitem{van2011relation}
Arjan van~der Schaft.
\newblock On the relation between port-hamiltonian and gradient systems.
\newblock In {\em 18th IFAC World Congress, Milano, Italy}, volume~44, pages
  3321--3326, 2011.

\bibitem{fortney2010dirac}
Jon~Pierre Fortney.
\newblock Dirac structures in pseudo-gradient systems with an emphasis on
  electrical networks.
\newblock {\em IEEE Transactions on Circuits and Systems I: Regular Papers},
  57(7):1732--1745, 2010.

\bibitem{jeltsema2003passivity}
Dimitri Jeltsema, Romeo Ortega, and Jacquelien~MA Scherpen.
\newblock On passivity and power-balance inequalities of nonlinear rlc
  circuits.
\newblock {\em IEEE Transactions on Circuits and Systems I: Fundamental Theory
  and Applications}, 50(9):1174--1179, 2003.

\bibitem{stokes}
Arjan van~der Schaft and Bernhard~M Maschke.
\newblock {H}amiltonian formulation of distributed-parameter systems with
  boundary energy flow.
\newblock {\em Journal of Geometry and Physics}, 42(1):166--194, 2002.

\bibitem{AbrMarRat88}
Ralph Abraham, Jerrold~E Marsden, and Tudor~S Ratiu.
\newblock {\em Manifolds, Tensor Analysis, and Applications}, volume~75 of {\em
  Applied Mathematical Sciences}.
\newblock Springer-Verlag, 2nd edition, 2012.

\bibitem{swaters}
Gordon~E Swaters.
\newblock {\em Introduction to {H}amiltonian Fluid Dynamics and Stability
  Theory}, volume 102.
\newblock CRC Press, 1999.

\bibitem{venkatraman2009energy}
Aneesh Venkatraman and Arjan van~der Schaft.
\newblock Energy shaping of port-hamiltonian systems by using alternate passive
  outputs.
\newblock In {\em Proceedings of the European control conference, Budapest,
  Hungary}, pages 2175--2180, 2009.

\bibitem{venkatraman2010energy}
Aneesh Venkatraman and Arjan van~der Schaft.
\newblock Energy shaping of port-hamiltonian systems by using alternate passive
  input-output pairs.
\newblock {\em European Journal of Control}, 16(6):665--677, 2010.

\bibitem{ortega2003power}
Romeo Ortega, Dimitri Jeltsema, and Jacquelien~MA Scherpen.
\newblock Power shaping: A new paradigm for stabilization of nonlinear rlc
  circuits.
\newblock {\em IEEE Transactions on Automatic Control}, 48(10):1762--1767,
  2003.

\bibitem{ZhenBaoOmer}
Zheng-Hua Luo, Bao-Zhu Guo, and {\"O}mer Morg{\"u}l.
\newblock {\em Stability and Stabilization of Infinite Dimensional Systems with
  Applications}.
\newblock Springer-Verlag, London, 1999.

\bibitem{RamThe}
Ramkrishna Pasumarthy.
\newblock {\em On analysis and control of interconnected finite-and
  infinite-dimensional physical systems}.
\newblock PhD thesis, Twente University Press, 2006.

\bibitem{donaire2016shaping}
Alejandro Donaire, Rachit Mehra, Romeo Ortega, Sumeet Satpute, Jose~Guadalupe
  Romero, Faruk Kazi, and Navdeep~M Singh.
\newblock Shaping the energy of mechanical systems without solving partial
  differential equations.
\newblock {\em IEEE Transactions on Automatic Control}, 61(4):1051--1056, 2016.

\bibitem{mehra2017control}
Rachit Mehra, Sumeet~G Satpute, Faruk Kazi, and Navdeep~M Singh.
\newblock Control of a class of underactuated mechanical systems obviating
  matching conditions.
\newblock {\em Automatica}, 86:98--103, 2017.

\bibitem{gogte2012passivity}
Gaury Gogte, Chinde Venkatesh, Faruk Kazi, Navdeep~M Singh, and Ramkrishna
  Pasumarthy.
\newblock Passivity based control of underactuated 2-d spidercrane manipulator.
\newblock In {\em 20th International Symposium on Mathematical Theory of
  Networks and Systems, Melbourne, Australia}, 2012.

\bibitem{satpute2014geometric}
Sumeet Satpute, Rachit Mehra, Faruk Kazi, and Navdeep~M Singh.
\newblock Geometric--pbc approach for control of circular ball and beam system.
\newblock In {\em 21st International Symposium on Mathematical Theory of
  Networks and Systems, Groningen, The Netherlands}, 2014.

\bibitem{borja2015shaping}
Pablo Borja, Rafael Cisneros, and Romeo Ortega.
\newblock Shaping the energy of port-hamiltonian systems without solving pde's.
\newblock In {\em 54th IEEE Conference on Decision and Control, Osaka, Japan},
  pages 5713--5718, 2015.

\bibitem{khalil1996noninear}
Hassan~K Khalil.
\newblock Noninear systems.
\newblock {\em Prentice-Hall, New Jersey}, 2(5), 1996.

\bibitem{ma2012predictive}
Yudong Ma, Anthony Kelman, Allan Daly, and Francesco Borrelli.
\newblock Predictive control for energy efficient buildings with thermal
  storage: Modeling, stimulation, and experiments.
\newblock {\em IEEE Control Systems}, 32(1):44--64, 2012.

\bibitem{deng2010building}
Kun Deng, Prabir Barooah, Prashant~G Mehta, and Sean~P Meyn.
\newblock Building thermal model reduction via aggregation of states.
\newblock In {\em Proceedings of the 2010 American Control Conference,
  Baltimore, USA}, pages 5118--5123, 2010.

\bibitem{garcia2004new}
Eloisa Garcia-Canseco and Romeo Ortega.
\newblock A new passivity property of linear rlc circuits with application to
  power shaping stabilization.
\newblock In {\em Proceedings of the 2004 American Control Conference, Boston,
  USA}, volume~2, pages 1428--1433, 2004.

\bibitem{lohmiller1998contraction}
Winfried Lohmiller and Jean-Jacques~E Slotine.
\newblock On contraction analysis for non-linear systems.
\newblock {\em Automatica}, 34(6):683--696, 1998.

\bibitem{forni2014differential}
Fulvio Forni and Rodolphe Sepulchre.
\newblock A differential lyapunov framework for contraction analysis.
\newblock {\em IEEE Transactions on Automatic Control}, 59(3):614--628, 2014.

\bibitem{forni2013differential}
Fulvio Forni, Rodolphe Sepulchre, and Arjan van~der Schaft.
\newblock On differential passivity of physical systems.
\newblock In {\em 52nd IEEE Conference on Decision and Control, Florence,
  Italy}, pages 6580--6585, 2013.

\bibitem{van2013differential}
Arjan van~der Schaft.
\newblock On differential passivity.
\newblock volume~46, pages 21--25, 2013.

\bibitem{crouch1987variational}
P~E Crouch and Arjan van~der Schaft.
\newblock {\em Variational and Hamiltonian control systems}.
\newblock Springer-Verlag New York, Inc., 1987.

\bibitem{krasovskiicertain}
N.~N. Krasovskii.
\newblock {\em Certain Problems of the Theory of Stability of Motion [in
  Russian], Fizmatgiz, Moscow}.
\newblock English translation by Stanford University Press, 1963, 1959.

\bibitem{nijmeijer1990nonlinear}
Henk Nijmeijer and Arjan van~der Schaft.
\newblock {\em Nonlinear dynamical control systems}, volume 175.
\newblock Springer-Verlag, New York, 1990.

\bibitem{Jeltsemaa06apower}
Dimitri Jeltsema and Jacquelien M.~A. Scherpen.
\newblock A power-based description of standard mechanical systems.
\newblock {\em Systems \& Control Letters}, 56(5):349--356, 2007.

\bibitem{jeltsema2003dual}
Dimitri Jeltsema and Jacquelien~MA Scherpen.
\newblock A dual relation between port-hamiltonian systems and the
  brayton--moser equations for nonlinear switched rlc circuits.
\newblock {\em Automatica}, 39(6):969--979, 2003.

\bibitem{DimCleOrtSch}
Dimitri Jeltsema, Jesus Clemente-Gallardo, Romeo Ortega, Jacquelien~MA
  Scherpen, and J~Ben Klaassens.
\newblock Brayton-moser equations and new passivity properties for nonlinear
  electromechanical systems.
\newblock {\em proc. Mechatronics 2002, Twente, The Netherlands}, 2002.

\bibitem{KoopDim}
Johan Koopman and Dimitri Jeltsema.
\newblock Casimir-based control beyond the dissipation obstacle.
\newblock volume~45, pages 173--177, 2012.

\bibitem{YanhanMas}
Yann Le~Gorrec, Hans Zwart, and Bernhard Maschke.
\newblock Dirac structures and boundary control systems associated with
  skew-symmetric differential operators.
\newblock {\em SIAM Journal on Control and Optimization}, 44(5):1864--1892,
  2005.

\bibitem{pasumarthy2007achievable}
Ramkrishna Pasumarthy and Arjan van~der Schaft.
\newblock Achievable casimirs and its implications on control of
  port-{H}amiltonian systems.
\newblock {\em International Journal of Control}, 80(9):1421--1438, 2007.

\bibitem{ben2001lectures}
Aharon Ben-Tal and Arkadi Nemirovski.
\newblock {\em Lectures on modern convex optimization: analysis, algorithms,
  and engineering applications}.
\newblock SIAM, 2001.

\bibitem{ibaraki1988resource}
Toshihide Ibaraki and Naoki Katoh.
\newblock {\em Resource allocation problems: algorithmic approaches}.
\newblock MIT press, 1988.

\bibitem{kelly1998rate}
Frank~P Kelly, Aman~K Maulloo, and David~KH Tan.
\newblock Rate control for communication networks: shadow prices, proportional
  fairness and stability.
\newblock {\em Journal of the Operational Research society}, 49(3):237--252,
  1998.

\bibitem{boyd2004convex}
Stephen Boyd and Lieven Vandenberghe.
\newblock {\em Convex optimization}.
\newblock Cambridge university press, 2004.

\bibitem{xiao2006optimal}
Lin Xiao and Stephen Boyd.
\newblock Optimal scaling of a gradient method for distributed resource
  allocation.
\newblock {\em Journal of optimization theory and applications},
  129(3):469--488, 2006.

\bibitem{kose1956solutions}
T~Kose.
\newblock Solutions of saddle value problems by differential equations.
\newblock {\em Econometrica, Journal of the Econometric Society}, pages 59--70,
  1956.

\bibitem{arrow1958studies}
Kenneth~Joseph Arrow, Leonid Hurwicz, Hirofumi Uzawa, and Hollis~Burnley
  Chenery.
\newblock Studies in linear and non-linear programming.
\newblock 1958.

\bibitem{feijer2010stability}
Diego Feijer and Fernando Paganini.
\newblock Stability of primal--dual gradient dynamics and applications to
  network optimization.
\newblock {\em Automatica}, 46(12):1974--1981, 2010.

\bibitem{lygeros2003dynamical}
John Lygeros, Karl~Henrik Johansson, Slobodan~N Simic, Jun Zhang, and Shankar~S
  Sastry.
\newblock Dynamical properties of hybrid automata.
\newblock {\em IEEE Transactions on automatic control}, 48(1):2--17, 2003.

\bibitem{stegink2017unifying}
Tjerk Stegink, Claudio De~Persis, and Arjan van~der Schaft.
\newblock A unifying energy-based approach to stability of power grids with
  market dynamics.
\newblock {\em IEEE Transactions on Automatic Control}, 62(6):2612--2622, 2017.

\bibitem{cherukuri2016asymptotic}
Ashish Cherukuri, Enrique Mallada, and Jorge Cort{\'e}s.
\newblock Asymptotic convergence of constrained primal--dual dynamics.
\newblock {\em Systems \& Control Letters}, 87:10--15, 2016.

\bibitem{simpson2016input}
John~W Simpson-Porco.
\newblock Input/output analysis of primal-dual gradient algorithms.
\newblock In {\em 54th Annual Allerton Conference on Communication, Control,
  and Computing, Monticello, USA}, pages 219--224, 2016.

\bibitem{zhao2006notion}
Jun Zhao and David~J Hill.
\newblock A notion of passivity for switched systems with state-dependent
  switching.
\newblock {\em Journal of control theory and applications}, 4(1):70--75, 2006.

\bibitem{HybridPassive}
M~Zefran, F~Bullo, and M~Stein.
\newblock A notion of passivity for hybrid systems.
\newblock In {\em 40th IEEE Conference on Decision and Control, Orlando, USA},
  2001.

\bibitem{cortes1995support}
Corinna Cortes and Vladimir Vapnik.
\newblock Support-vector networks.
\newblock {\em Machine learning}, 20(3):273--297, 1995.

\bibitem{cortes2005characterization}
Jorge Cort{\'e}s, Arjan van~der Schaft, and Peter~E Crouch.
\newblock Characterization of gradient control systems.
\newblock {\em SIAM Journal on Control and Optimization}, 44(4):1192--1214,
  2005.

\bibitem{pavlov2008incremental}
Alexey Pavlov and Lorenzo Marconi.
\newblock Incremental passivity and output regulation.
\newblock {\em Systems \& Control Letters}, 57(5):400--409, 2008.

\bibitem{IQC}
Alexandre Megretski and Anders Rantzer.
\newblock System analysis via integral quadratic constraints.
\newblock {\em IEEE Transactions on Automatic Control}, 42(6):819--830, June
  1997.

\end{thebibliography}
\end{singlespace} 

  \section*{Publications}
{\bf I -} REFEREED JOURNALS 
\begin{enumerate}
\item K.  Kosaraju,  V.  Chinde,  R.  Pasumarthy,  A.  Kelkar,  and  N.  Singh, ``Stability analysis of constrained optimization dynamics via passivity techniques,'' {\em IEEE Control Systems Letters},  vol.  2,  no.  1,  pp.  91-96, 2018.
\item V. Chinde, K. Kosaraju, A. Kelkar, R.Pasumarthy, S. Sarkar, and N. Singh, 
``A passivity-based power-shaping control of building hvac systems, ''{\em Journal of Dynamic Systems, Measurement, and Control}, vol. 139, no. 11, pp. 111 007-111 007-10, 2017.
\item K. Kosaraju, R. Pasumarthy and D. Jeltsema, ``Modeling and boundary control of infinite dimensional systems in the Brayton Moser framework,'' {\em In IMA Journal of Mathematical Control and Information}, 2017. 
\end{enumerate}
\noindent {\bf II -}  BOOK CHAPTERS
\begin{enumerate}
\item  K. Kosaraju, R. Pasumarthy ``Power based methods for infinite dimensional systems''. In M.K. Camlibel, A.A. Julius, R. Pasumarthy, and J.M.A. Scherpen (eds.), {\em Mathematical Control Theory I: Nonlinear and Hybrid Control systems}, Springer, 2015.
\end{enumerate}
\noindent {\bf III -}  PEER-REVIEWED CONFERENCES
\begin{enumerate}
		\item K. Kosaraju, V. Chinde, R. Pasumarthy, A. Kelkar, and N. M. Singh. ``Differential passivity like properties for a class of nonlinear systems''. {\em In American Control Conference (ACC)}, Milwaukee, USA, 2018.
\item K. Kosaraju, R. Pasumarthy, N. Singh, and A. Fradkov, ``Control using new passivity property with differentiation at both ports,'' {\em Indian Control Conference (ICC)}, Guwahati, India, pp. 7-11, 2017.
\item V. Chinde, K. Kosaraju, A. Kelkar, R. Pasumarthy, S. Sarkar, and N. M. Singh. ``Building HVAC systems control using power shaping approach''. {\em In American Control Conference (ACC)}, Boston, USA, pp. 599-604, 2016.
\item K. C. Kosaraju, R. Pasumarthy, and D. Jeltsema, ``Alternative passive maps for infinite dimensional systems using mixed-potential functions,'' {\em IFAC Workshop on Lagrangian and Hamiltonian Methods for Non Linear Control}, Lyon, France, pp. 1-6, 2015.
\item R. Pasumarthy, K. Kosaraju, and A. Chandrasekar. `` On power balancing and stabilization for a class of infinite dimensional systems''. {\em In the 21st International Symposium on Mathematical Theory of Networks and Systems}, Groningen, The Netherlands, 2014.
%
%
%
%
\end{enumerate}
\end{document}